%% file: lpd_arxiv_v2.tex
\documentclass[aps, pra, reprint, superscriptaddress]{revtex4-2}
\def\preprint{0}
\def\draft{0}
\def\isprl{0}
\input{style/preamble.sty}

\input{style/macros}

\input{style/nsp_style.sty}

\begin{document}
\let\oldacl\addcontentsline
\renewcommand{\addcontentsline}[3]{}

\newcommand{\mytitle}{Classical Simulation Healed by Quantum Entanglement}

\title{\mytitle}
\author{Jue Xu}
\affiliation{QICI Quantum Information and Computation Initiative, Department of Computer Science, School of Computing and Data Science, The University of Hong Kong, Pokfulam Road, Hong Kong SAR, China}
\author{Chu Zhao}
\affiliation{Department of Electrical and Computer Engineering, Duke University, Durham, NC 27708, USA}
\author{Xiangran Zhang}
\affiliation{QICI Quantum Information and Computation Initiative, Department of Computer Science, School of Computing and Data Science, The University of Hong Kong, Pokfulam Road, Hong Kong SAR, China}
\author{Shuchen Zhu}
\affiliation{Department of Mathematics, Duke University, Durham, NC 27708, USA}
\author{Qi Zhao}
\email{zhaoqi@cs.hku.hk}
\affiliation{QICI Quantum Information and Computation Initiative, Department of Computer Science, School of Computing and Data Science, The University of Hong Kong, Pokfulam Road, Hong Kong SAR, China}

\date{\today}
\begin{abstract}

Entanglement is not only the origin of exotic quantum phenomena, but also widely regarded as the fundamental barrier to classical simulation of quantum dynamics.
We overturn this intuition:
for predicting local observables under Hamiltonian evolution,
entanglement in the quantum state actually heals the classical simulation error in the Heisenberg picture.
Classical algorithms that propagate observables in the Pauli basis with proper truncation have shown remarkable empirical success in simulating noiseless quantum dynamics,
yet all prior rigorous guarantees required noise or randomness---leaving the physically relevant noiseless regime without theoretical foundation.
To close this gap,
we prove that the truncation error of Low-weight Pauli Dynamics ($\lpd$) admits an average-case bound without assuming randomness,
provided the state is sufficiently entangled.
Since tensor-network methods efficiently simulate low-entanglement states while $\lpd$ thrives precisely in the complementary regime,
together they extend rigorous classical simulation to longer times,
sharpening the boundary between classically simulable and genuinely quantum dynamics.

\end{abstract}
\maketitle

\mysection[1]{Introduction}
Local observables such as correlators of Hamiltonian dynamics
serve as key indicators of intriguing quantum phenomena
including quantum phase transitions \cite{martinezRealtimeDynamicsLattice2016, bernienProbingManybodyDynamics2017,ebadiQuantumPhasesMatter2021}.
Nevertheless, direct simulation of generic Hamiltonian dynamics is intractable for classical computers due to the exponential growth of complexity with system size.
This difficulty motivated Feynman's proposal of a quantum computer \cite{feynmanSimulatingPhysicsComputers1982, feynmanQuantumMechanicalComputers1985},
making Hamiltonian simulation a primary candidate for practical quantum advantage
\cite{lloydUniversalQuantumSimulators1996,jordanQuantumAlgorithmsQuantum2012,ciracGoalsOpportunitiesQuantum2012,georgescuQuantumSimulation2014,daleyPracticalQuantumAdvantage2022, kimEvidenceUtilityQuantum2023, abaninObservationConstructiveInterference2025}.
Alongside the rapid advancement of quantum simulation \cite{berryEfficientQuantumAlgorithms2007,berrySimulatingHamiltonianDynamics2015,lowOptimalHamiltonianSimulation2017},
classical simulations of Hamiltonian dynamics have also improved significantly for specific cases, challenging quantum simulation advantages,
as shown in \cref{fig:venn}.
For example,
the tensor network methods \cite{vidalEfficientClassicalSimulation2003,vidalEfficientSimulationOnedimensional2004} can efficiently approximate quantum states and operators of low entanglement.
However, long-time Hamiltonian dynamics typically induces large entanglement that hinders classical simulation.
And the cluster expansion methods \cite{wildClassicalSimulationShortTime2023,wuEfficientClassicalAlgorithm2024}
offer solutions for short time evolution with product initial states,
while their empirical performance had remained untested (see Methods for more details).

In the NISQ era, noise in quantum circuits renders quantum simulation experiments \cite{kimEvidenceUtilityQuantum2023} more amenable to classical simulation \cite{begusicFastConvergedClassical2024, tindallEfficientTensorNetwork2024, fontanaClassicalSimulationsNoisy2025}.
Among these classical simulation techniques, the Pauli truncation style algorithms
have received significant attention for their theoretical guarantee,
since first introduced for noisy random circuit sampling \cite{bremnerAchievingQuantumSupremacy2017,gaoEfficientClassicalSimulation2018,rallSimulationQubitQuantum2019,aharonovPolynomialTimeClassicalAlgorithm2023}.
Here, \emph{Pauli truncation} is a general term for the classical algorithms that approximate Hamiltonian dynamics or quantum circuits by truncating Pauli operators obeying certain criteria \cite{rudolphPauliPropagationComputational2026}.
Later, Refs.~\cite{fontanaClassicalSimulationsNoisy2025,shaoSimulatingNoisyVariational2024} proved that truncating high-weight Pauli paths can approximate expectation values of noisy parameterized circuits with bounded average error over parameter space.
Moreover, Pauli truncation was extended to more general cases, such as the arbitrary noises \cite{martinezEfficientSimulationParametrized2025,angrisaniSimulatingQuantumCircuits2026},
any circuits with random input states \cite{schusterPolynomialtimeClassicalAlgorithm2025}, and the worst-case scenarios \cite{gonzalez-garciaPauliPathSimulations2025}.
Nevertheless, their error bounds crucially rely on the damping effect of noise channels as well as the randomness of circuits or states.

\begin{figure}[!t]
    \centering
    \includegraphics[width=0.96\linewidth]{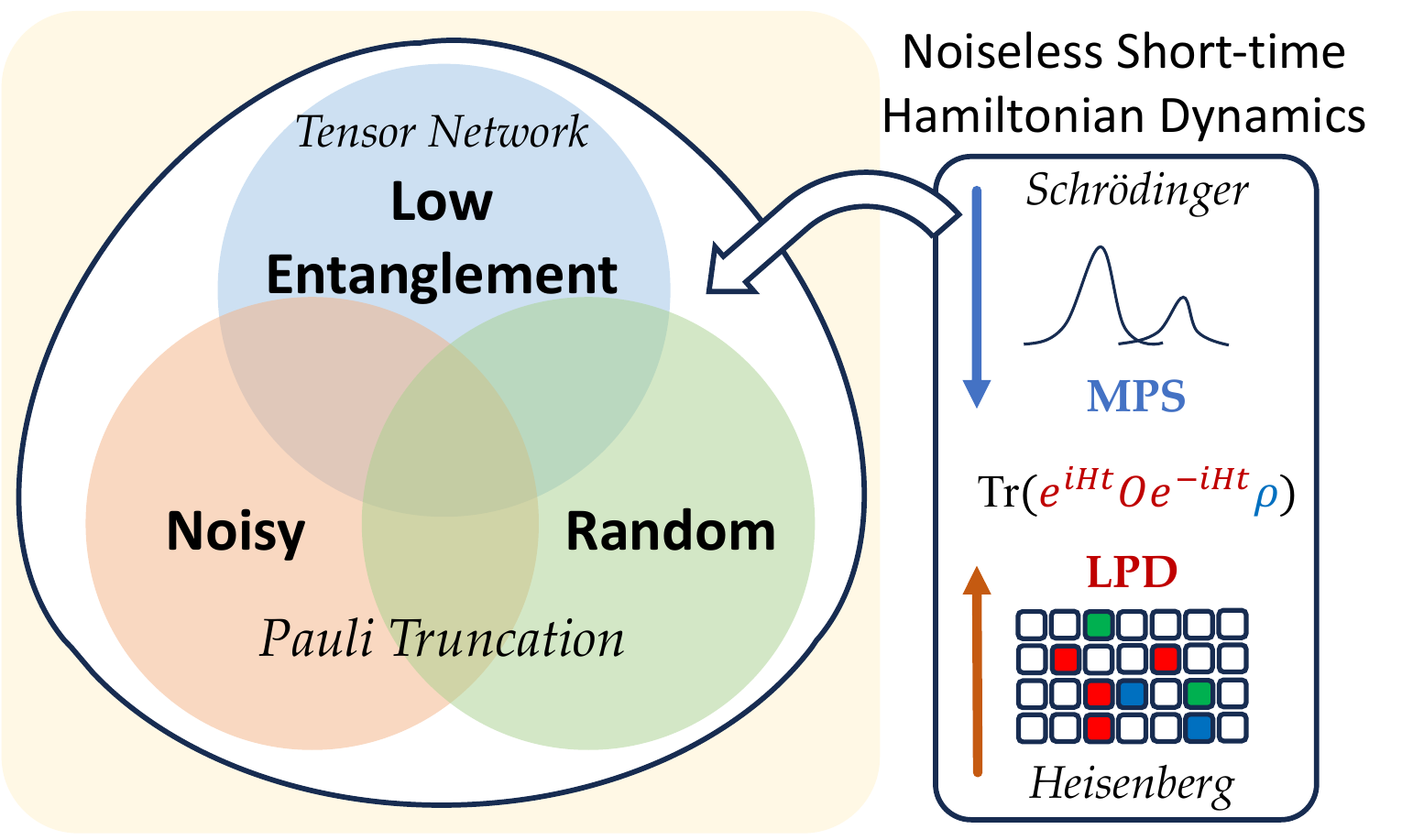}
    \caption{
    Regimes of efficient classical simulations for Hamiltonian dynamics.
    Tensor network methods efficiently represent quantum states with low entanglement,
    whereas existing theoretical guarantees for Pauli-truncation methods rely on noise or randomness.
    Our work shows the state entanglement actually suppresses the Pauli truncation error in the noiseless case, thereby broadening the regime of classical simulation.
    }
    \label{fig:venn}
\end{figure}

\begin{figure*}[!t]
    \centering
    \includegraphics[width=0.99\linewidth]{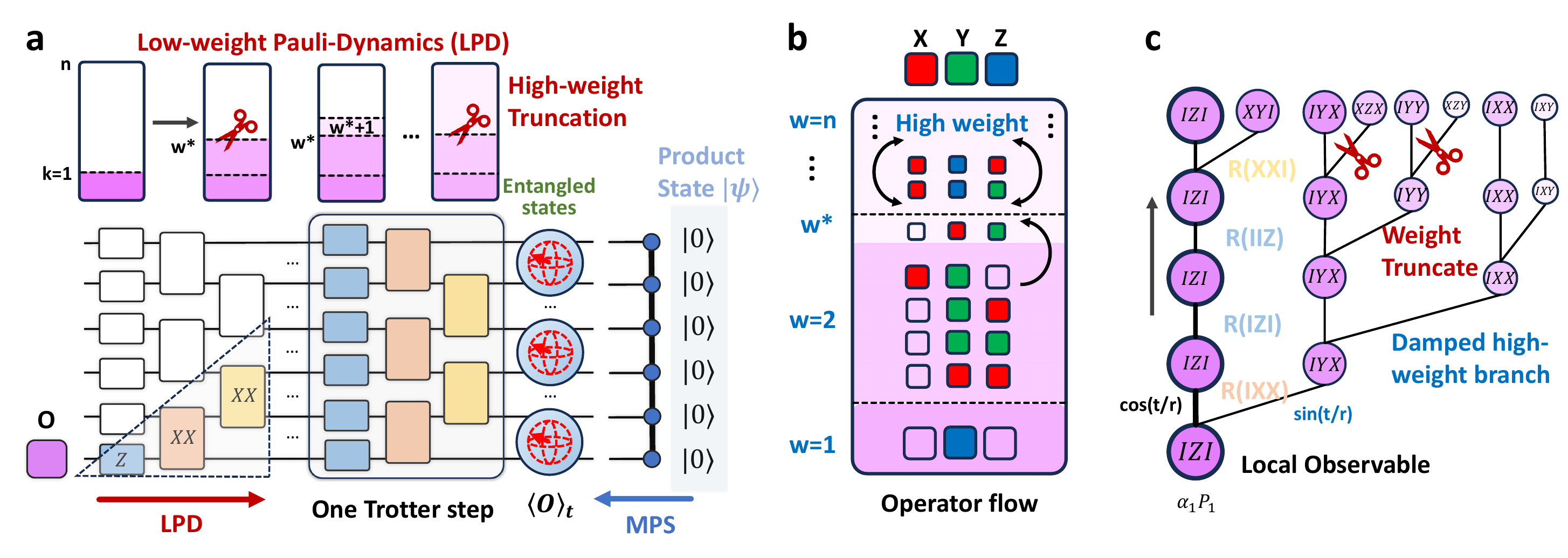}
    \caption{
    The illustration of the $\lpd$ (Low-weight Pauli Dynamics) algorithm.
    (a) The Trotterized backward evolution of local observable with high-weight truncation and light-cone argument.
    The initial product state is evolved to a sufficiently entangled state by the MPS method.
    The expectation value is evaluated by the low-weight observable and the entangled state.
    (b) The local operator flow under local Pauli rotations.
    The norm of Paulis with certain weight flowing to higher weights is damped by small angle $\dt$.
    The darker color represents a larger norm contribution of Paulis with certain weight to the norm of the evolved observable.
    (c) Truncation of high-weight Pauli operators in one step.
    The colors and the sizes of the bubbles represent the magnitudes of the Pauli coefficients.
    In this example, the Pauli operators with weight larger than $w^*=2$ are truncated.
    }
    \label{fig:main}
\end{figure*}

Surprisingly, even without noise and randomness,
the Pauli truncation algorithms still perform well for certain Hamiltonian dynamics.
Begušić et al. \cite{begusicSimulatingQuantumCircuit2025,begusicRealTimeOperatorEvolution2025} proposed the heuristics by truncating small-coefficient Paulis in Trotterized evolution,
while Angrisani et al. \cite{angrisaniClassicallyEstimatingObservables2025} proved truncating high-weight Paulis in noiseless locally scrambling circuits yields bounded average error and it is empirically effective for some Hamiltonians.
These results motivate a systematic study of the theoretical guarantees and limitations of Pauli truncation for noiseless Hamiltonian dynamics without randomness assumptions.

We bridge this gap by analyzing a classical algorithm that we call \emph{Low-weight Pauli Dynamics} ($\lpd$)---the weight-truncation strategy within the Pauli-propagation framework \cite{angrisaniClassicallyEstimatingObservables2025,begusicFastConvergedClassical2024,begusicRealTimeOperatorEvolution2025,rudolphPauliPropagationComputational2026,loizeauQuantumManybodySimulations2025}---and by proving that it is efficient for short-time Hamiltonian dynamics.
The algorithm approximates observables by truncating high-weight Pauli operators at each Trotter step.
Here, the \emph{weight} of a Pauli string is the number of qubits on which it acts nontrivially---that is, with $X$, $Y$, or $Z$ rather than the identity $I$.
Owing to locality of Paulis and light-cone, the growth of Pauli weights are constrained, and the contributions of high-weight Paulis to expectation values are suppressed by small rotation angles even in the absence of noise.
Moreover, we prove that if the input state has sufficient entanglement entropy,
the Pauli truncation error in expectation satisfies a nontrivial average-case bound even without randomness.
We show that the sufficiently entangled states are tractable to tensor network methods
\cite{vidalEfficientSimulationOnedimensional2004}
and they can be prepared by near-term quantum  simulators \cite{zhaoEntanglementAcceleratesQuantum2025}.
Naturally, we combine $\lpd$ with the tensor network state simulation to effectively extend accessible classical simulation time.
On the other hand, our classically evolved low-weight observables essentially reduce the quantum simulation circuit depth for evolving quantum states.
We numerically show the effectiveness of $\lpd$ with the chaotic quantum mixed-field Ising Hamiltonian.

\mysection[1]{Problem Setup and LPD Algorithm}
This work focuses on calculating the expectation of local observables evolved by Hamiltonian dynamics.
Formally,
    given a $\kh$-local
    \footnote{In this paper, the word ``local'' means that all Paulis in the local operator have constant (system-size independent) weights, not necessarily being geometrically local.}
    $n$-qubit Hamiltonian $H=\sum_l^L \alpha_l G_l$ with
    Paulis $G_l$ and $\alpha_l\in\bigO(1)$,
    evolution time $t\in\mathbb{R}$,
    a $\ko$-local observable $O=\sum_j^J \beta_j P_j$ with Paulis $P_j$
    and $\beta_j\in\bigO(1)$,
    and a state $\rho$,
    the \emph{Hamiltonian dynamics expectation}
    problem is to approximate
    \begin{equation}\label{eq:objective}
        \mu(H,t,O,\rho):=\Tr(e^{\ii Ht} O e^{-\ii Ht} \rho),
    \end{equation}
where the local observable is evolved in the Heisenberg picture to exploit constrained light-cones \cite{bravyiClassicalAlgorithmsQuantum2021}.

Our $\lpd$ classical simulation is summarized in Alg.~\ref{alg:noiseless}.
The first step is the Trotter approximation,
where the dynamics evolution $U:= e^{-\ii Ht}$ is approximated
by the Trotter formula $\pf^r$ and $r$ is the number of repeated Trotter steps.
For example, we can approximate the expectation \cref{eq:objective} by the first-order Trotter formula $\pf_1$
\begin{equation}\label{eq:trotter_formula}
    \tilde{\mu}:=\Tr(\rho \pf^{\dagger r}_1 O \pf^{r}_1), \quad
    \pf_1:=\prod_{l=1}^L e^{-\ii\alpha_l G_l t/r}
    .
\end{equation}
Crucially, to suppress the light-cone spread,
the Hamiltonian $H= \sum_{l=1}^L\alpha_l G_l $ is usually regrouped into $H=\sum_{\gamma=1}^{\Gamma} H_\gamma$ where each $e^{-\ii H_\gamma t/r }$ is the product of Pauli rotations that can be implemented simultaneously.
For example, the one-step Trotter evolution of the nearest-neighbor Hamiltonian can be arranged into $\Gamma\in\bigO(1)$ layers of Pauli rotations
as the brickwork circuit in \cref{fig:main}(a).

After Trotterization of the dynamics,
each unitary is a Pauli rotation $e^{-\ii G_l\dt_l/2}$ generated by a $\kh$-local Pauli operator
with weight at most $\kh$, whose conjugation angle is $\dt_l:=2\alpha_l t/r$, of magnitude at most $\dt:=\alpha t/r$ with $\alpha:=2\max_l\abs{\alpha_l}$.
In the Heisenberg picture,
a Pauli rotation acting on a Pauli $P$ is always either non-branching or 2-branching,
depending on their commutation relation
\begin{equation}\label{eq:pauli_rotation_branch}
     e^{\ii G\frac{\dt}{2}} P e^{-\ii G\frac{\dt}{2}} =
    \begin{cases}
        P, &[P,G] = 0 \\
        \cos(\dt) P + \ii \sin(\dt) GP, &\qty{P,G}=0
    \end{cases}
\end{equation}
where $GP=[G,P]/2$
is a new Pauli string with weight changed by at most $\kh-1$ and with a small factor $\sin(\dt)$,
as shown in \cref{fig:main}(b).

Next, $\lpd$ approximates the Trotterized evolution of the observable $O$ by low-weight Pauli operators.
    For each Trotter step,
    we apply the Pauli rotations one-by-one to the Pauli operators in $\tO$ by \cref{eq:pauli_rotation_branch}.
    After applying the last gate in one Trotter step,
    we truncate all high-weight Pauli operators above the threshold $w^*$ as exemplified in \cref{fig:main}(c)
\begin{equation}
    \forall d\in[r],\;
    \tilde{O}^{(d)}_{\le w^*}
    = \sum_{Q:\abs{Q}\le w^*} c_Q Q,
\end{equation}
where $\tO^{(d)}_{\le w^*}$ denotes the truncated observable after applying the last gate in the $d$-th Trotter step and
$\abs{Q}$ denotes the weight of Pauli $Q$, i.e., the number of non-identity Pauli operators in $Q$.
    After all $r$ Trotter steps,
    the expectation of the evolved observable is approximated with the state $\rho$ by $\tilde{\mu}_{\le w^*} = \Tr(\rho \tilde{O}^{(r)}_{\le w^*})$.

$\lpd$ is naturally suited to hybrid protocols.
If the input state is a product state,
we evolve it in the Schr\"odinger picture forwardly for time $t_F$ by tensor network methods, such as the Matrix Product States ($\mps$) \cite{vidalEfficientClassicalSimulation2003,vidalEfficientSimulationOnedimensional2004}.
Once the evolved state is sufficiently entangled, we switch to the Heisenberg picture and backward evolve the input observable by $\lpd$ for $t_B:=t-t_F$.
We will prove the entanglement of the state helps bound the Pauli truncation error of the evolved observable.
We delay the numerical evidence of the hybrid simulation and of the healing effect to \cref{fig:hybrid} and put more details about the matching entanglement condition in
Methods.

On the other hand, the near-term quantum simulators can also prepare entangled states such that our backward-evolved observable effectively reduces the quantum circuit depth for long-time evolution, equivalently extending the simulation time.
And the classical shadow method \cite{huangPredictingManyProperties2020} is suitable for efficient evaluation of the evolved observable on a quantum state.

\mysection[1]{Analysis: error and complexity}
Most prior works on Pauli truncation \cite{aharonovPolynomialTimeClassicalAlgorithm2023,shaoSimulatingNoisyVariational2024,angrisaniClassicallyEstimatingObservables2025} require the randomness in quantum circuits to cancel cross terms to obtain a non-trivial error bound.
This trick is suitable for random circuit sampling or parameterized quantum algorithms,
but not directly apply to Trotterized Hamiltonian dynamics.
To overcome it, Schuster et al. \cite{schusterPolynomialtimeClassicalAlgorithm2025} leveraged the randomness of input states to achieve average-case bounds for any circuits.
We realize an average-case bound with a milder assumption:
a sufficiently entangled input state rather than an ensemble of random states.
More importantly, we bound the high-weight Pauli contribution in noiseless Hamiltonian dynamics
by exploiting the properties of small-angle local Pauli rotations.

    Firstly, to achieve the \emph{Trotter error} \cite{yuObservabledrivenSpeedupsQuantum2025,zhaoEntanglementAcceleratesQuantum2025} in the expectation of a local observable $O$ within error tolerance $\eps$,
        that is, $\abs{\bra{\psi}(\pf^{\dagger r}_p O \pf_p^r - U^\dagger O U)\ket{\psi}}\le \eps$
    for the $p$th-order Trotter formula $\pf_p$,
    we need Trotter steps
    \begin{equation}\label{eq:average_trotter_number}
        r \in \bigO\qty((\Lambda\cdot\norm{O}\cdot\eps^{-1})^{1/p}\cdot t^{1+1/p}),
    \end{equation}
    where
    $\Lambda:=\sum\limits_{\gamma_1,\dots,\gamma_{p+1}=1}^{\Gamma}
    \norm{[H_{\gamma_{p+1}},[H_{\gamma_p},...,[H_{\gamma_2},H_{\gamma_1}]]]}$ is the operator norm of the nested commutator \cite{childsTheoryTrotterError2021}
    and $H=\sum_{\gamma=1}^\Gamma H_{\gamma}$.
This worst-case bound works for any input states.
Furthermore, for random input states
\cite{zhaoHamiltonianSimulationRandom2021,chenAverageCaseSpeedupProduct2024} or a sufficiently entangled input state \cite{zhaoEntanglementAcceleratesQuantum2025},
the operator norm in \cref{eq:average_trotter_number} can be improved to the Pauli 2-norm.
In this paper,
we denote the \emph{Pauli 2-norm}
\footnote{The Pauli 2-norm of a Hermitian operator $A$ is the $l_2$-norm of the Pauli coefficients $\vec{c}$ with $A=\sum_{P\in \pbasis_{n}} \alpha_P P$. It is equivalent to the normalized Schatten 2-norm $\nfnorm{A}:=\sqrt{\Tr(A^\dagger A)/2^n}$ of operator $A$ on $n$ qubits, also known as the normalized Frobenius norm}
as $\nfnorm{A}:=\qty(\frac{1}{2^n}\sum_{s\in \npbasis_{n}} \Tr(As)^2)^{1/2}$,
where $\npbasis_{n}$ is the Pauli basis.
We put the details of the Trotter error in
\cref{apd:trotter_error}.

The key challenge of this work is to bound the
    \emph{Pauli truncation error} in expectation values, i.e.,
    \begin{equation}\label{eq:truncation_error}
        \Delta \tilde{\mu}_{\le w^*}
        := \abs{\bra{\psi}(\tO^{(r)} - \tO^{(r)}_{\le w^*})\ket{\psi}},
    \end{equation}
    where $\tO_{\le w^*}^{(r)}$ is the low-weight approximation of Trotter $\pf$ evolved observable $\tO^{(r)}:=\pf^{\dagger r} O \pf^{r}$ at the last Trotter step with truncation threshold $w^*$.
In the worst-case analysis \cite{dowlingMagicResourcesHeisenberg2025}, \cref{eq:truncation_error} has a loose upper bound by the operator norm of the observable difference
$\norm{\tO^{(r)} - \tO^{(r)}_{\le w^*}}$.
A finer analysis is the average-case one \cite{fullerImprovedQuantumComputation2026,schusterPolynomialtimeClassicalAlgorithm2025},
if the input states are sampled from the 2-design ensemble $\statee_2$,
the average Pauli truncation error has a tighter Pauli 2-norm bound $\nfnorm{\tO^{(r)} - \tO^{(r)}_{\le w^*}}$
\footnote{
The mean-square expectation values can be bounded by the squared Pauli 2-norm of the observable, that is,
 $ \bbE_{\ket{\psi}\sim \statee_2} \qty[\abs{\bra{\psi} O\ket{\psi}}^2] \le \nfnorm{O}^2$
} than the operator norm.

Interestingly, a similar Pauli 2-norm bound can be achieved without random input assumption as long as the observable is local,
i.e., the Paulis in the observable have constant weights,
and the input state is sufficiently entangled on its subsystems.
\begin{lemma}[Pauli 2-norm upper bound on local observable expectation value with entangled states]\label{thm:entangle_ob_bound}
    Given a local observable $O=\sum_j O_j$,
    if the input state $\ket{\psie}$ has the subsystem entanglement entropy $S(\rho_{j,j'})\ge\abs{\supp(O^\dagger_j O_{j'})}-\frac{1}{2}\nfnorm{O}^4/(\sum_j\norm{O_j})^4$ for all subsystems $\supp(O^\dagger_j O_{j'})$,
    then the squared expectation value
    $ \abs{\bra{\psie} O \ket{\psie}}^2 \le 2\nfnorm{O}^2.$
\end{lemma}

We sketch the proof of this lemma in Methods.
Notably, telescoping the observable difference $\tO^{(r)} - \tO^{(r)}_{\le w^*}$ over Trotter steps, each per-step truncated component has system-size independent Pauli weight (at most $w^*\kh^{\Gamma}$ for the first-order formula, by the light-cone) uniformly in $r$ as long as one Trotter step consists of $\bigO(1)$ layers of rotations with pairwise disjoint supports; applying \cref{thm:entangle_ob_bound} per step in the Schr\"odinger picture then requires the entanglement condition of the intermediate Trotter-evolved states.
Meanwhile, the requirement in \cref{thm:entangle_ob_bound} is a \emph{local} one---every subsystem of constant size $k$ close to maximally mixed---rather than a bound on the bipartite entropy across extensive cuts; it is therefore compatible with states of constant bipartite entanglement entropy that remain within the capability of $\mps$, although constant entropy alone does not imply it (a GHZ state carries one bit across every cut yet fails the condition already for $k=2$).
So, the Pauli truncation error analysis now reduces to bounding the Pauli 2-norm of the observable difference.
Leveraging the triangle inequality, it
can be further bounded by the sum of the norms of high-weight components over all $r$ Trotter steps as
\begin{equation}\label{eq:truncation_error_triangle_bound}
    \Delta \tilde{\mu}_{\le w^*}  \le
    \sqrt{2}\sum_{d=1}^{r} \nfnorm{\tO^{(d)}_{\ge w^*+1}},
\end{equation}
where the Pauli 2-norm of high-weight components reads
$\nfnorm{O_{\ge w^*+1}}:= \sqrt{\sum_{w\ge w^*+1} \nfnorm{O_{=w}}^2}$, the component $\tO^{(d)}_{\ge w^*+1}$ is the high-weight part discarded at the end of step $d$ of the truncated trajectory, i.e.\ $(1-\Pi_{\le w^*})\pf^\dagger\tO^{(d-1)}_{\le w^*}\pf$, and
we define the weight-$w$ component of $O$ as
$O_{=w}:=\sum_{s\in \npbasis_{n},\abs{s}=w} \Tr(Os)s$.

Then, we only need to bound the norms of high-weight Paulis of one Trotter step.
To tackle it without the damping effect by noise,
our core idea is to track the damped local flow of norms from certain weights to higher ones driven by small-angle Pauli rotations as in \cref{fig:main}(b) and \cref{fig:typical_case_pauli}(c).
We formalize this intuition as below.
\begin{lemma}[Damped local norm flow]\label{thm:weight_decay}
    Given an initial $\ko$-local observable $O$
    and the product of $\kh$-local Pauli rotations $U_g:=\prod_{l=1}^{g} e^{-\ii G_l\cdot \dt/2}$ with conjugation angle $\dt\in[0,\pi/2]$ (rotations with different angles $\dt_l$ are covered by $\dt:=\max_l\abs{\dt_l}$),
    consider the evolved observable $O^{(g)}:= U_g^\dagger O U_g$,
    the high-weight norm $\cumnorm_{\ge m}^{(g)}$
    satisfies the recursion relation
    \begin{equation}\label{eq:worst_local_flow_k_local}
        \cumnorm_{\ge m}^{(g)}\le
        \cumnorm_{\ge m}^{(g-1)}+\sin(\dt)\cdot \cumnorm_{\ge m-1}^{(g-1)},
    \end{equation}
    where $\cumnorm_{\ge m}^{(g)}:=\nfnorm{O_{\ge w_m+1}^{(g)}}=\qty(\sum_{w>w_{m}} \nfnorm{O_{=w}^{(g)}}^2)^{1/2}$ is a single Pauli 2-norm and $w_m:=\ko+(m-1)(\kh-1)$ denotes the maximum weight that can be achieved by $m-1$ anti-commuting $\kh$-local Pauli rotations from $\ko$.
\end{lemma}
Thus, for an initially low-weight Pauli to become a high-weight one,
it must undergo many small-angle Pauli rotations
and thereby accrue a significant damping factor.
By applying \cref{thm:weight_decay} recursively,
we can bound the cumulation of high-weight norms as
$    \cumnorm_{\ge m}^{(g)}\le \binom{g}{m} \sin^m(\dt) \nfnorm{O}$,
such that we can use it
to bound the total truncation error over all Trotter steps \cref{eq:truncation_error_triangle_bound}.
See the complete proofs in
\cref{apd:pauli_error}.
Once we have the total truncation error in terms of the truncation threshold $w^*$,
one can determine the required $w^*$ to achieve precision $\eps$.
We formalize it as \cref{thm:truncation_threshold} in
Methods.

\begin{figure}[!t]
    \centering
    \includegraphics[width=0.99\linewidth]{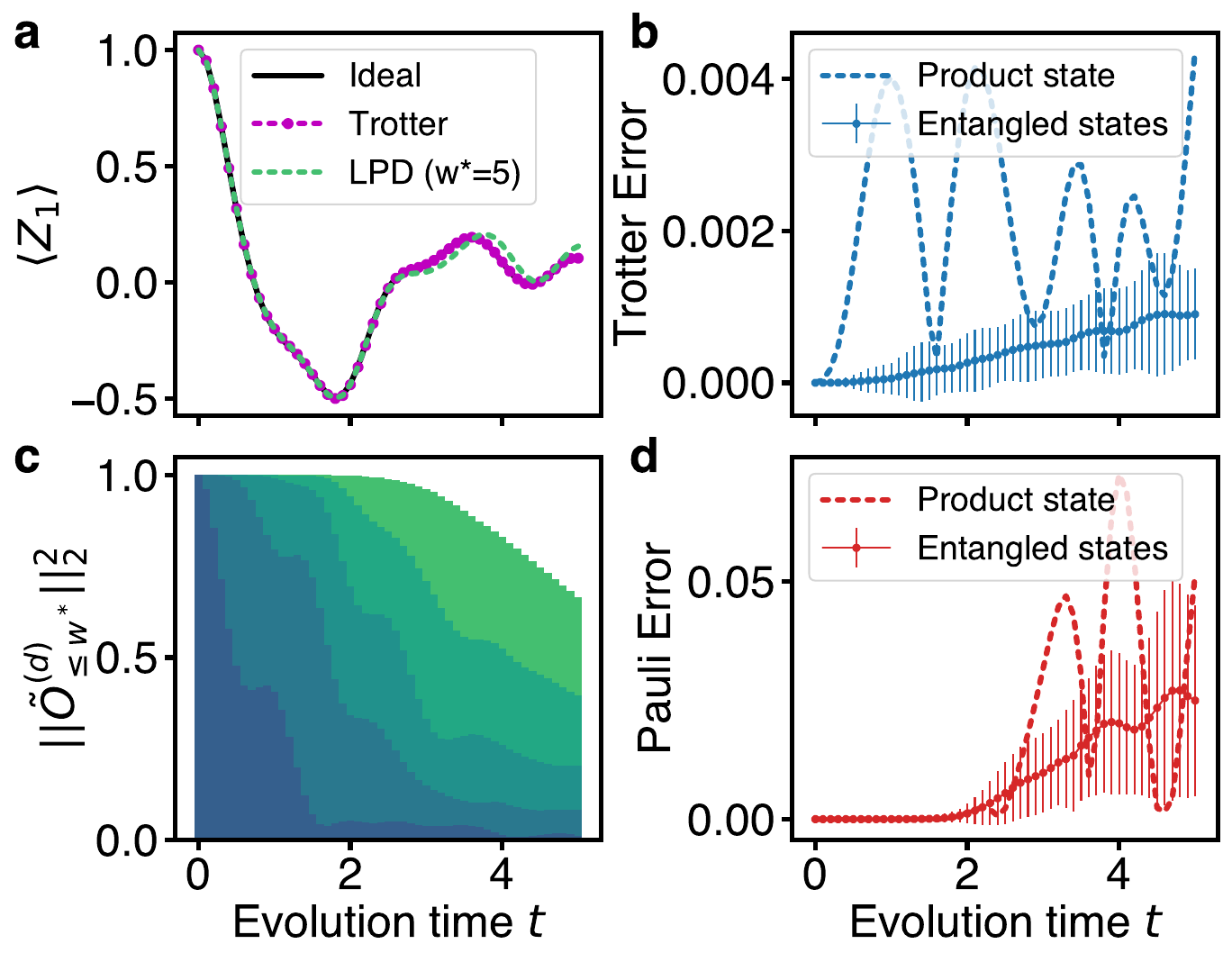}
    \caption{
    The numerical results of Alg.~\ref{alg:noiseless} for the QMFI model.
    We set evolution time $t=5$ and the second-order Trotter formula with $r=50$ steps.
    The initial product state is $\ket{01\cdots01}$ and the truncation threshold is $w^*=5$ in every panel.
    (a) The expectation values at Trotter every step.
    The solid line represents the ideal expectation value.
    The purple dotted line is the Trotter expectation without truncation.
    The dashed line is the one with the truncation.
    (b)
    The dashed line is Trotter error with the product input state,
    while the line with error bars is the average Trotter error over 100 random Clifford states (a state 2-design), which are highly entangled.
    (c) In the short-time regime, the low-weight Paulis (dark green) dominate the norm distribution over weights, while the high-weight Paulis (light green) gradually accumulate in longer time.
    The empty regime is the norm lost by the Pauli truncation.
    (d)
    The dashed line is the Pauli truncation error in expectation values
    with the product input state,
    while the line with error bars is the average Pauli truncation error over the same random Clifford (2-design) ensemble of entangled states.
    }
    \label{fig:typical_case_pauli}
\end{figure}

If the truncation threshold $w^*$ is independent of $n$ and one Trotter step is shallow
(i.e., Pauli rotations can be arranged into $\Gamma\in\bigO(1)$ layers),
there will be polynomially many $\bigO(n^{w^*})$ Pauli operators throughout the low-weight Pauli dynamics.
At last,
the expectation value of the final low-weight observable can also be computed efficiently
as long as the Pauli decomposition of the state can be efficiently computed.
Thus, we have the following polynomial algorithm runtime.
\begin{theorem}[Runtime of Alg.~\ref{alg:noiseless}, informal]\label{thm:runtime}
    Given a local Hamiltonian $H=\sum_{\gamma}^{\Gamma} H_{\gamma}$ with $\Gamma\in\bigO(1)$ layers of terms with pairwise disjoint supports,
    any local observable $O$ normalized so that $\nfnorm{O}=1$,
    evolution time $t< t_0=\frac{1}{c_0\Gamma(\kh-1)\alpha}\in\Omega(1)$
    with a constant $c_0\le2$ (see \cref{apd:thm:c0_bound}) and $\alpha:=2\max_l\abs{\alpha_l}$,
    and a sufficiently entangled input state $\ket{\psi_S}$ as required in \cref{thm:entangle_ob_bound}, under the conditions on $m^*$ and $r$ of \cref{thm:truncation_threshold},
    assume the Pauli coefficients of $H$, $O$, and $\rho=\op{\psie}$ can each be computed in $\bigO(1)$ time.
    There exists a classical algorithm $\lpd$ to approximate the expectation value
    $\Tr(\rho e^{\ii Ht} O e^{-\ii Ht})$
    within error tolerance $2\eps$.
    The runtime is
        $\bigO\qty(r\,n^{ w^*})$
        where $r\in\poly(n,1/\eps)$ is the number of Trotter steps (\cref{eq:average_trotter_number})
        and $w^*=\ko+\bigO\qty(\frac{\log(1/\eps)}{\log(t_0/t)})$ is the truncation threshold.
\end{theorem}
    The time limit $t_0$ is a constant independent of $n$, but depends on the locality, strength, and the number of disjoint-support layers of the Hamiltonian.
    We shall emphasize that in the case of \cref{thm:runtime}
    (i.e., constant time and shallow one-step Trotter),
    the full Trotter circuit is not necessarily a shallow circuit because the number of Trotter steps can be polynomial with $n$.
    If we require precision $\eps\sim 1/\poly(n)$,
    then the algorithm needs quasi-polynomial runtime because $w^*$ becomes $\bigO(\log(n))$.
We put the complete analysis of the runtime and memory complexity in
\cref{apd:sec:proof_runtime}.

\mysection[1]{Numerical results}
To substantiate our theoretical result, we provide concrete numerical evidence,
while Refs.~\cite{shaoSimulatingNoisyVariational2024,angrisaniClassicallyEstimatingObservables2025,begusicRealTimeOperatorEvolution2025,shaoPauliPropagationSimulating2025} have demonstrated the effectiveness of other Pauli truncation algorithms.
We test on the one-dimensional quantum mixed-field Ising (QMFI) model
$H_{\mathrm{MFI}} = \sum_{j=1}^{n-1} X_j X_{j+1} + h_x\sum_{j=1}^n X_j + h_y\sum_{j=1}^n Y_j$
with $h_x=0.8,h_y=0.9$ and the open boundary condition,
of which the dynamics has rapidly growing entanglement \cite{kimBallisticSpreadingEntanglement2013} and operator magic \cite{dowlingMagicResourcesHeisenberg2025}.
We choose $n=10$ qubits to exactly evaluate all errors.
We choose the local observable $O=Z_1$ (Pauli $Z$ on the first qubit) and the truncation threshold $w^*=5$.
\cref{fig:typical_case_pauli}(a, b) shows that
the Trotter errors in expectation values are small and
the truncation errors with $w^*=5$ are small for short evolution time.
\cref{fig:typical_case_pauli}(d) demonstrates the average truncation errors $\mathbb{E}_{\ket{\psi_0}\sim\mathrm{Clif}}\abs{\bra{\psi_0}\tO^{(d)} - \tO_{\le w^*}^{(d)}\ket{\psi_0}}$ of entangled states (the random-Clifford 2-design ensemble; Haar averages are indistinguishable here),
which are much smaller than the truncation errors for the specific product state $\ket{\psi_0}=\ket{01\cdots 01}$.

\begin{figure}[!t]
    \centering
    \includegraphics[width=0.99\linewidth]{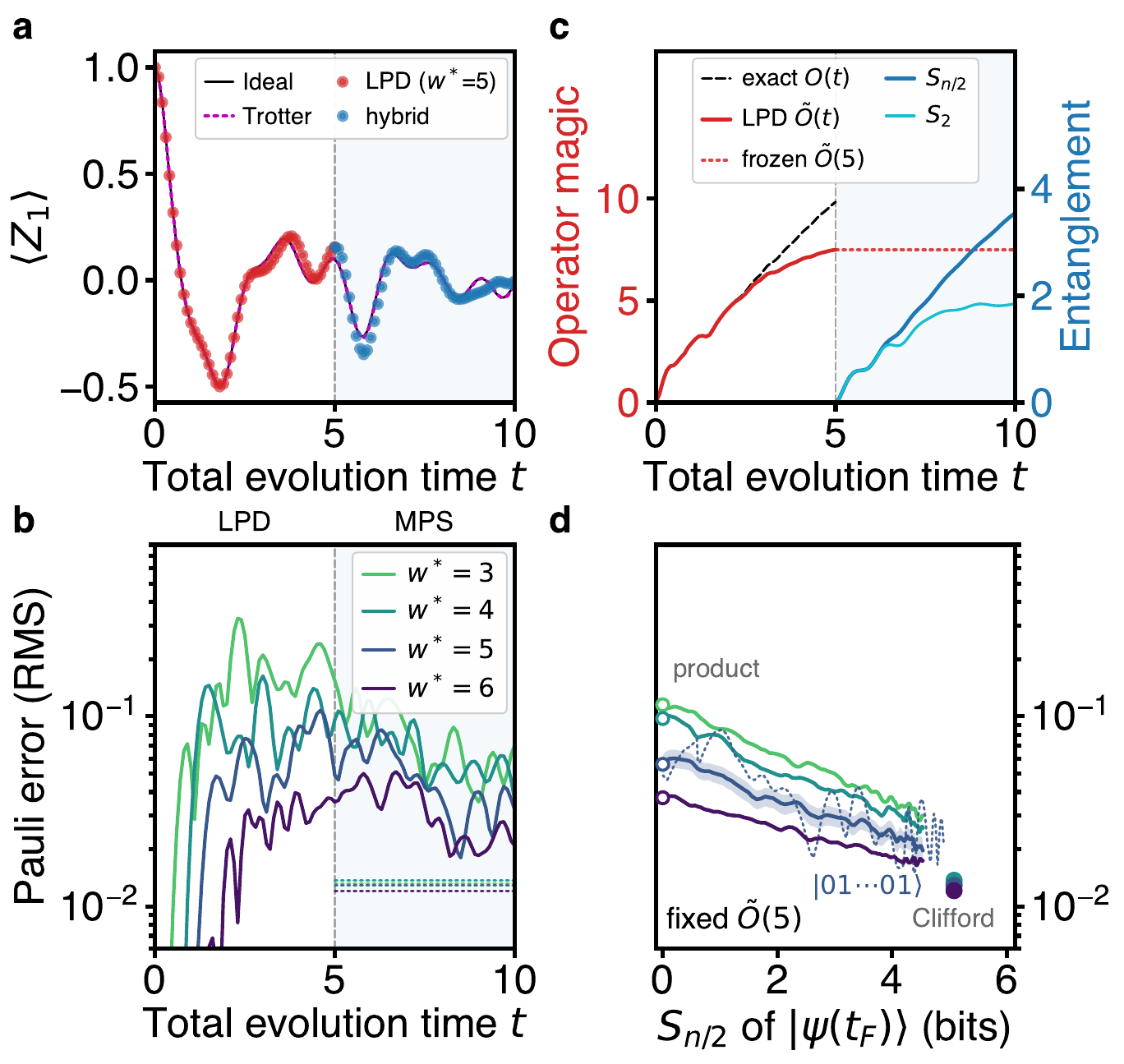}
    \caption{
    $\lpd$ complements $\mps$, and the entanglement of the state heals the Pauli truncation error (QMFI model, $n=12$).
    For $t\le5$ (white region), the observable is backward-evolved by $\lpd$ on the product state $\ket{01\cdots01}$;
    for $t>5$ (shaded region), the truncated observable $\tO(5)$ is frozen and evaluated on the $\mps$-evolved state $\ket{\psi(t-5)}$.
    (a) $\langle Z_1\rangle$ for the ideal, Trotter, $\lpd$ ($w^*=5$) and hybrid dynamics.
    (b) Pauli truncation error (root mean square over all single-site $Z_q$) for $w^*=3$--$6$;
    dotted lines: the same frozen $\tO(5)$ averaged over random Clifford states.
    (c) Operator magic of $\tO(t)$ (left axis) and entanglement entropies of $\ket{\psi(t-5)}$ (right axis).
    (d) Error of the frozen $\tO(5)$ versus the half-chain entanglement entropy of the state it is evaluated on:
    mean over $50$ random product inputs (solid; band: standard error, $w^*=5$), the $\ket{01\cdots01}$ trajectory (dotted),
    and the levels of $50$ random product states (open circles) and $50$ random Clifford states (filled circles, at their mean entanglement entropy).
    }
    \label{fig:hybrid}
\end{figure}

Additionally, we numerically demonstrate the hybrid classical simulation combining our low-weight observable simulation $\lpd$ with the matrix product state simulation $\mps$,
together with the healing of the Pauli truncation error by the entanglement of the state (\cref{thm:entangle_ob_bound}).
In \cref{fig:hybrid}, the observable is backward-evolved by $\lpd$ for $t_B=\min(t,5)$,
and for $t>5$ the truncated observable $\tO(5)$ is held fixed while the state is forward-evolved by $\mps$ for $t_F=t-5$,
i.e., the hybrid expectation value is $\bra{\psi(t-5)}\tO(5)\ket{\psi(t-5)}$.
We set the second-order Trotter formula with $r=100$ steps ($\dt=0.1$) and the bond dimension $\chi=256$, which is exact at this system size,
so that \cref{fig:hybrid}(b) isolates the Pauli truncation error with respect to the Trotter dynamics;
it is averaged in root mean square over the $n$ single-site observables $Z_q$ to remove the zero crossings of individual expectation values.
The Clifford averages in \cref{fig:hybrid}(b) nearly coincide for all $w^*$,
since the truncation error operator $O(5)-\tO(5)$ has normalized 2-norm close to $1$ for every threshold considered.
In \cref{fig:hybrid}(c), we exhibit the rapid growth of two resources \cite{dowlingBridgingEntanglementMagic2025} associated with these two classical algorithms:
the entanglement entropy of the forward-evolved state and the operator magic
\footnote{Here, we use the Shannon entropy of the (normalized) squared Pauli coefficients of an observable $O=\sum_P \alpha_P P$,
that is, $\mathcal{M}(O) = -\sum_{P\in \pbasis} \abs{\alpha_P}^2 \log \qty(\abs{\alpha_P}^2)$
as the operator magic (nonstabilizerness) measure of $O$ \cite{dowlingMagicResourcesHeisenberg2025}.}
of the backward-evolved observable.
It implies that if one only uses either $\mps$ or $\lpd$,
the classical simulation would lose its theoretical guarantee after its resource saturates.
Instead, \cref{fig:hybrid}(a) shows that the $\lpd$-$\mps$ hybrid simulation effectively extends the simulation time while keeping the error small.
Moreover, \cref{fig:hybrid}(b, d) show that the truncation error of the fixed observable $\tO(5)$ is suppressed,
on average over $50$ random product inputs forward-evolved by the same Trotter circuit for $t_F\le10$,
monotonically as the entanglement of the evolved state grows,
from the level of random product states ($0.056$ at $S_{n/2}=0$ for $w^*=5$)
towards the level of random Clifford states ($0.013$ at their mean entanglement entropy $S_{n/2}\approx5.1$ bits, close to the maximum $n/2=6$; the ensemble mean reaches $0.020$ at $S_{n/2}=4.5$ bits),
as predicted by \cref{thm:entangle_ob_bound}.
We put more details about entanglement and magic in
Methods.
\Cref{apd:sec:quantinuum} applies the hybrid protocol to the 36- and 56-qubit digitized Ising quench of Quantinuum's H2 experiment \cite{haghshenasDigitalQuantumMagnetism2026}.
Furthermore, the input state can also be from a real quantum device,
then $\lpd$ can enhance quantum simulation as well \cite{fullerImprovedQuantumComputation2026},
i.e., reducing quantum simulation circuit depth virtually.

\mysection{Discussion and Outlook}
Our work provides theoretical guarantees for the empirical success of the Pauli truncation style algorithms in the noiseless Hamiltonian dynamics, without resorting to randomness assumptions.
Surprisingly, the $\lpd$ algorithm works seamlessly with entangled input states that are believed unfavorable for classical simulation.
Meanwhile, the short-time limitation of $\lpd$ does not rule out the possibility of practical quantum advantage via quantum dynamics simulation.
Instead, $\lpd$ can help reduce quantum simulation circuit depths.
At the same time, the full power of $\lpd$ remains unexplored.
For instance, our theoretical bound could be improved
given specific Hamiltonian information, such as geometric locality and commutation relation.
It is of practical interest to identify typical good states \cite{begusicRealTimeOperatorEvolution2025},
where high-weight Paulis have little contribution to expectation
such that the actual truncation error is much smaller than the theoretical bound.
To accelerate $\lpd$, it is feasible to truncate the high-weight Paulis along with the Paulis of other properties, e.g., small coefficients \cite{begusicRealTimeOperatorEvolution2025,shaoPauliPropagationSimulating2025} or symmetries \cite{tengLeveragingSymmetryMerging2025}.
It also remains to extend our techniques to non-local observables and fermionic systems \cite{millerSimulationFermionicCircuits2025,facelliFastConvergenceMajorana2026}.

\bibliographystyle{style/truncate_ref}
\bibliography{bib/ref_aps}
\clearpage

\begin{center}
{\bf \large {Methods}}
\end{center}

\mysection[3]{Sketch proof of \cref{thm:entangle_ob_bound} (Pauli 2-norm bound)}
We first introduce the upper bound on the expectation value of a local observable regarding the distance between the reduced density matrix and maximally mixed states.
\begin{lemma}[Upper bound on local observable expectation value]\label{mtd:lem:local_observable_2bound}
    Given a $w$-local observable $O= \sum_{j} O_j$
    where each $O_j$ acts non-trivially on the subsystem $\supp(O_j)$ and each has weight at most $w$ (i.e., $\max_j \abs{\supp(O_j)} \le w$),
    the squared expectation value of $O$ with respect to any pure state $\ket{\psi}$ has the upper bound
    \begin{equation}\label{mtd:eq:local_observable_bound}
        \abs{\bra{\psi} O \ket{\psi}}^2
        \le \nfnorm{O}^2 + \sum_{j,j'} \norm{O_j^\dagger O_{j'}} \Tr\abs{\rho_{j,j'}-\frac{I_{\textsf{S}_{jj'}}}{2^{\abs{\textsf{S}_{jj'}}}}},
    \end{equation}
    where $\rho_{j,j'} := \Tr_{[n]\backslash \textsf{S}_{jj'}}(\op{\psi})$ is the reduced density matrix of $\ket{\psi}$ on the subsystem $\textsf{S}_{jj'}:=\supp(O^\dagger_j O_{j'})$.
\end{lemma}
Recall the Pauli 2-norm of $A$ is defined as
$\nfnorm{A}:=\qty(\frac{1}{2^n}\sum_{s\in \npbasis_{n}} \Tr(As)^2)^{1/2}$,
where $\npbasis_{n}$ denotes the Pauli basis.
The complete proof of \cref{mtd:lem:local_observable_2bound} can be found in \cref{apd:sec:proof} and partially from \cite{zhaoEntanglementAcceleratesQuantum2025}.
As its corollary, we have the \cref{thm:entangle_ob_bound} in the main text.
\begin{proof}[Proof of \cref{thm:entangle_ob_bound}]
    Let $S(\rho_A):=-\Tr(\rho_A \log \rho_A)$ denote the entanglement entropy of the reduced state $\rho_A$, with all logarithms in base $2$.
    If the input state $\ket{\psie}$ has the subsystem entanglement entropy $S(\rho_{j,j'})\ge\abs{\supp(O^\dagger_j O_{j'})}-\frac{1}{2}\nfnorm{O}^4/(\sum_j\norm{O_j})^4$ for all subsystem $\supp(O^\dagger_j O_{j'})$,
    the trace distance part (the second term in \cref{mtd:eq:local_observable_bound}) has the upper bounded
    \begin{align}
        &\sum_{j,j'} \norm{O_j^\dagger O_{j'}} \Tr\abs{\rho_{j,j'}-I_{\textsf{S}_{jj'}}/2^{\abs{\textsf{S}_{jj'}}}} \\
        \le& \sum_{j,j'} \norm{O_j^\dagger O_{j'}} \sqrt{2S\qty(\rho_{j,j'}||I_{\textsf{S}_{jj'}}/2^{\abs{\textsf{S}_{jj'}}})} \\
        \le &\sum_{j,j'} \norm{O_j^\dagger}\norm{ O_{j'}} \sqrt{2 \abs{\textsf{S}_{jj'}} -2S(\rho_{j,j'})}\\
        \le&\qty(\sum_j \norm{O_j})^2 \max_{j,j'}\sqrt{2 \abs{\supp(O^\dagger_j O_{j'})} -2S(\rho_{j,j'})} \le \nfnorm{O}^2.
    \end{align}
    where $S(\rho||\sigma):=\Tr(\rho \log \rho - \rho \log \sigma)$ is the relative entropy between $\rho$ and $\sigma$.
    The second line is due to Pinsker's inequality, $\Tr\abs{\rho-\sigma}\le\sqrt{2\ln2\, S(\rho||\sigma)}$ in base $2$, with the factor $\ln2<1$ relaxed to $1$ for a cleaner constant.
    Thus, we have $\abs{\bra{\psie} O \ket{\psie}}^2 \le 2\nfnorm{O}^2$ that is a Pauli 2-norm bound.
\end{proof}

\mysection[3]{Sketch proof of Pauli truncation error}
Next, we sketch proof of the lemma about the damped norm flow of the observable in Trotterized dynamics.
\begin{proof}[Sketch proof of \cref{thm:weight_decay}]
    To grasp the intuition, we consider the simplest non-trivial case
    where all Pauli rotations are 2-local ($\kh=2$) and the initial observable is 1-local ($\ko=1$).
    Since any 2-local Pauli rotation can increase the weights of the Paulis in the observable by at most 1,
    the norm of the high-weight (above $m$) component can be increased by at most the norm above weight $m-1$ with a damping factor $\sin(\dt)$:
    \begin{equation}
        \nfnorm{O_{\ge m+1}^{(g)}}\le
        \nfnorm{O_{\ge m+1}^{(g-1)}}+
        \sin(\dt)\,\nfnorm{O_{\ge m}^{(g-1)}},
    \end{equation}
    where $\nfnorm{O_{\ge m}^{(g)}}$ denotes the Pauli 2-norm of the weight-${\ge}m$ component of the observable $O^{(g)}$ evolved by $g$ 2-local Pauli rotations; working with this single norm, rather than a sum of per-weight norms, is what makes the regrouping step of the general proof valid.
    For the general $k$-local case, we refer to the complete proof in \cref{apd:pauli_error}.
\end{proof}

\begin{algorithm}[!t]
    \label{alg:noiseless}
    \DontPrintSemicolon
    \SetAlCapHSkip{0ex}
    \caption{Low-weight Pauli Dynamics ($\lpd$)}
    \SetKwInOut{Input}{Input}
    \SetKwInOut{Output}{Output}
    \Input{Local Hamiltonian $H$, local observable $\ob$, time $t$, state $\rho$, error tolerance $\eps$}
    \Output{approximation of $\mu=\Tr(e^{\ii Ht}O e^{-\ii Ht} \rho)$}
    \BlankLine
    Set Trotter steps $r$, per-term rotation angles $\dt_l:=2\alpha_l t/r$ with $\abs{\dt_l}\le\dt$ \tcp*{\cref{eq:average_trotter_number}}
    Set Pauli truncation weight $w^*$ \tcp*{Prop.~\ref{thm:truncation_threshold}}
    \If{input state $\rho$ is a product state}{
        Evolve $\rho$ by $\mps$ with $\chi=2^{\bigO(w^*)}$
    }
    $\tilde{O}^{(0)}_{\le w^*} \leftarrow O$ \tcp*{Initialize $O$}
    \For{\texttt{\textup{each Trotter step}} $d\in[r]$}{
        $\tO_{\le w^*}^{(d_{0})}\leftarrow \tO_{\le w^*}^{(d-1)}$ \tcp*{Start of the $d$th step}
        \For{\texttt{\textup{each rotation}} $g\in\qty{e^{-\ii G_l \dt_l/2}}$} {
            $\tilde{O}_{\le w^*}^{(d_{g})}\leftarrow 0$ \tcp*{Rotation $g$ in the $d$th step}
            \For{\texttt{\textup{each Pauli}} $P\in \tilde{O}_{\le w^*}^{(d_{g-1})}$}{
                $\tilde{O}_{\le w^*}^{(d_g)} \leftarrow \tilde{O}_{\le w^*}^{(d_g)} + e^{\ii G_l \dt_l/2} P e^{-\ii G_l \dt_l/2}$  \tcp*{Update}
            }
        }
            $\tilde{O}_{\le w^*}^{(d)} \leftarrow \sum_{Q:\abs{Q}\le w^*} c_Q Q$  \tcp*{Truncation by weight}
    }
    \Return
    Expectation
    $\tilde{\mu}_{\le w^*}=\Tr(\rho \tilde{O}_{\le w^*}^{(r)})$
\end{algorithm}

Then, we formally derive the truncation error bound and the truncation weight threshold of $\lpd$ (Alg.~\ref{alg:noiseless}).
\begin{proposition}[Truncation threshold of $\lpd$]\label{thm:truncation_threshold}
    Given a local Hamiltonian $H=\sum_{l}^L \alpha_l G_l=\sum_{\gamma}^{\Gamma} H_{\gamma}$ of polynomial Paulis that can be arranged into $\Gamma\in\bigO(1)$ layers of terms with pairwise disjoint supports in one-step Trotter,
    any $\ko$-local observable $O=\sum_j \beta_j O_j$ of polynomial Paulis with maximum weight $\ko\in\bigO(1)$,
    evolution time $t< t_0=\frac{1}{c_0\Gamma(\kh-1)\alpha}\in\Omega(1)$ with a constant $c_0\le2$ and $\alpha:=2\max_l\abs{\alpha_l}$, under the mild conditions $m^*\ge1$ and $r\ge\max(5,\,m^*,\,8(m^*+1)^2/\Gamma,\,8e^2(\ko+\kh-1)\alpha t)$ on the ladder index $m^*$ of the threshold $w^*=\ko+m^*(\kh-1)$ (see \cref{apd:thm:one_step_truncation_error}),
    and a sufficiently entangled input state $\ket{\psi_S}$ as required in \cref{thm:entangle_ob_bound}.
    To achieve Pauli truncation error $\Delta\tilde{\mu}_{\le w^*}$ of Alg.\ref{alg:noiseless} within $\eps\nfnorm{O}$,
    it suffices to take
    the truncation threshold
    $w^* = \ko+ \bigO\qty(\frac{\log(1/\eps)}{\log(t_0/t)})$.
\end{proposition}
The conditions on $r$ are constants independent of $n$ (and, for $t<t_0$, of $t$),
so they hold automatically once $r$ grows with $n$ or $1/\eps$ for Trotter accuracy,
and otherwise cost only a constant factor in $r$, which never weakens the bound since $c_0$ decreases with $r$ (\cref{apd:thm:c0_bound}).
\begin{proof}[Sketch proof of \cref{thm:truncation_threshold}]
    Recall, assume the input state is sufficiently entangled as \cref{thm:entangle_ob_bound},
    then by the triangle inequality,
    the total Pauli truncation error of $\lpd$ \cref{eq:truncation_error} has the Pauli 2-norm upper bound
    \begin{equation}\label{em:eq:tri_bound}
        \Delta \tilde{\mu}_{\le w^*} \le
        \sqrt{2}\sum_{d=1}^{r} \nfnorm{\tO_{\ge w^*+1}^{(d)}},
    \end{equation}
    where the right-hand side is the sum of the Pauli 2-norms of the high-weight components  $\tO_{\ge w^*+1}^{(d)}$ at the end of every Trotter step before truncation.

    On the other hand,
    we apply \cref{thm:weight_decay} recursively to bound the norm of the high-weight Paulis in the evolved observable $O^{(g)}$ without truncation
    \begin{align}\label{em:eq:cumulation}
        \cumnorm_{\ge m}^{(g)}
        \le \binom{g}{m} \sin^m(\dt) \nfnorm{O}
        \le \qty(\frac{eg\alpha t}{mr})^m  \nfnorm{O},
    \end{align}
    where $\cumnorm_{\ge m}^{(g)}:=\nfnorm{O^{(g)}_{\ge w_m+1}}$ with $w_m:=\ko+(m-1)(\kh-1)$, so that a Pauli operator of weight $w>w_m$ requires at least $m$ anti-commuting Pauli rotations to be produced from weight $\ko$.
    To bound \cref{em:eq:tri_bound} by \cref{em:eq:cumulation}, we need to count how many  rotations effectively applied in one Trotter step.

Rather than counting effective rotations along individual branches, we bound the inflow of one entire layer of disjoint-support rotations by an operator-norm (Schur-test) estimate of its action on the Pauli coefficient vector: crossing one rung of the weight ladder within a single layer costs at most a factor $\bigO\qty(w_m \sin(\dt))$, uniformly over coefficient vectors. Combining the $\Gamma$ layers per Trotter step with the rung-linear rates $w_m=\ko+(m-1)(\kh-1)$ yields the constants $C_1$ and $C_2=c_0\Gamma(\kh-1)\alpha$ below, with a constant $c_0\le2$ under the conditions of \cref{thm:truncation_threshold}; the accounting is manifestly system-size independent.
Combining these observations,
we can finally bound the Pauli truncation error in the expectation
\begin{align}\label{em:eq:truncation_error}
    \Delta\tilde{\mu}_{\le w^*}
    &\le
    \frac{t}{r} \cdot C_1\cdot\sum_{d=1}^{r}  \qty(\frac{t}{r}\cdot d\cdot C_2)^{m^*} \cdot \nfnorm{O} \\
    &\le t \cdot C_1\cdot \qty(t\cdot C_2)^{m^*} \cdot \nfnorm{O},
\end{align}
where the truncation threshold is parameterized on the weight ladder as $w^*=\ko+m^*(\kh-1)=w_{m^*+1}$,
$C_2=c_0\Gamma(\kh-1)\alpha\le2\Gamma(\kh-1)\alpha$ is controlled by the constants $\ko$, $\kh$, $\Gamma$, and $\qty{\alpha_l}$ alone (the factor $c_0\le2$ depending in addition on $r$ and $m^*$), while $C_1$ depends on these and grows polynomially in $m^*$, as $(m^*+2)^{1+\ko/(\kh-1)}$ up to constants.
To avoid error blow-up, we need $t\cdot C_2< 1$,
demanding the evolution time $t$ is smaller than a constant; explicitly, $t<t_0=\frac{1}{c_0\Gamma(\kh-1)\alpha}$.
With the short-time condition,
let the right-hand side of \cref{em:eq:truncation_error} smaller than $\eps \nfnorm{O}$.
We obtain the truncation threshold $w^* = \ko + \bigO\qty(\frac{\log(1/\eps)}{\log(t_0/t)})$,
which is constant for constant precision $\eps$ (\cref{apd:thm:truncation_threshold_entangled}).
\end{proof}

\mysection[3]{Sketch proof of \cref{thm:runtime} (LPD Runtime)}
The runtime analysis is mainly to bound the number of Paulis in all steps.
First of all, at the beginning of each truncated Trotter step,
all Paulis in the observable have weight at most $w^*$ by our $\lpd$ algorithm.
So, the observable has at most $\sum_{w=1}^{w^*}3^{w} \tbinom{n}{w} \le \sum_{w=1}^{w^*}\qty(\frac{3en}{w})^{w}\in\bigO(n^{w^*})$ Paulis.
After applying the Pauli rotations,
the Pauli operators would proliferate and the weights may also increase.
The good news is that the weight increase is system-size independent
because a local Pauli rotation only acts on constant number of qubits of a Pauli
and generate at most one new Pauli with at most a constant weight increase once.
Therefore, if each Trotter step is a shallow (constant depth) circuit,
by the constrained light-cone expansion, the maximum weights of Paulis at the end of every Trotter step are still system-size independent and the number of Paulis is still bounded by $\bigO(n^{w^*})$.

The last step of $\lpd$ is to calculate the expectation value of the observable with respect to a given state $\rho$,
that is $\Tr(\tO^{(r)}_{\le w}\rho)$.
Assume each coefficient of $\rho$ in Pauli basis can be computed in $\bigO(1)$ time,
then the expectation value can be computed in time proportional to the number of Pauli operators in the final observable by orthogonality of Pauli operators.
So, the evaluation of the expectation does not incur runtime more than one step Pauli propagation.
In the hybrid protocol the input is the forward-evolved matrix-product state, for which each Pauli expectation costs $\poly(n,\chi)$ rather than $\bigO(1)$, and the next section accounts for this evaluation cost.
The $\bigO(1)$ assumption holds exactly for stabilizer inputs, whose Pauli expectations $\bra{\psi}P\ket{\psi}\in\{0,\pm1\}$ follow from the stabilizer formalism.
Correspondingly, the memory requirement is the maximum number of Pauli operators during the evolution $\bigO(n^{w^*})$,
which is also polynomial if the truncation weight is system-size independent.

\mysection[3]{LPD complements MPS}
Since $\lpd$ backward-evolves the observable in the Heisenberg picture
and $\mps$ forward-evolves the state in the Schr\"odinger picture~\cite{vidalEfficientClassicalSimulation2003,vidalEfficientSimulationOnedimensional2004},
the two methods complement each other
and can be combined into a hybrid protocol with a rigorous efficiency guarantee and extended simulation capabilities
(cf.\ mixed Schr\"odinger--Heisenberg tensor-network simulation \cite{begusicFastConvergedClassical2024} and the meet-in-the-middle option of Pauli propagation \cite{rudolphPauliPropagationComputational2026}).

\begin{proposition}[Hybrid $\mps$--$\lpd$ simulation]\label{mtd:prop:hybrid_efficiency}
    Given a one-dimensional local Hamiltonian $H = \sum_\gamma^\Gamma H_\gamma$ with $\Gamma\in\bigO(1)$ layers of terms with pairwise disjoint supports,
    a local observable $O$,
    a product input state $\rho(0) = \op{\psi(0)}$,
    total evolution time $t = t_F + t_B$ with $t_B < t_0$,
    and $r = r_F + r_B$ Trotter steps chosen as in \cref{thm:runtime},
    suppose that
    (i) throughout the first $r_F$ Trotter steps, after every gate, the evolved state has Schmidt rank at most $\chi$ across every cut, and
    (ii) the states $\pf^{d}\ket{\psi(0)}$ for $d = r_F,\dots,r$ satisfy the entanglement condition of \cref{thm:entangle_ob_bound} as required in \cref{thm:runtime}.
    Then the hybrid $\mps$--$\lpd$ protocol approximates
    $\Tr(\rho(0)\, e^{\ii Ht}\, O\, e^{-\ii Ht})$
    within additive error $\bigO(\eps\nfnorm{O})$ in time
    \begin{equation}\label{mtd:eq:hybrid_runtime}
        \bigO\!\left(n\chi^3 r_F \;+\; n^{w^*}\qty(r_B + n\chi^3)\right)
    \end{equation}
    and memory $\bigO(n\chi^2 + n^{w^*})$,
    where $w^*$ is the truncation threshold of \cref{thm:truncation_threshold}.
\end{proposition}
\begin{proof}
    Under (i), the TEBD evolution with bond dimension $\chi$ represents $\ket{\psi_F}:=\pf^{r_F}\ket{\psi(0)}$ exactly~\cite{vidalEfficientClassicalSimulation2003,vidalEfficientSimulationOnedimensional2004}.
    Each gate update costs $\bigO(\chi^3)$
    and $\bigO(n)$ gates are applied per Trotter step,
    so the forward phase costs $\bigO(n \chi^3 r_F)$.
    The hybrid estimate is $\bra{\psi_F}\tO^{(r_B)}_{\le w^*}\ket{\psi_F}$, the $\lpd$ estimate with input state $\ket{\psi_F}$ and $r_B$ backward steps.
    Since $\bra{\psi(0)}\pf^{\dagger r} O \pf^{r}\ket{\psi(0)} = \bra{\psi_F}\pf^{\dagger r_B} O \pf^{r_B}\ket{\psi_F}$,
    it targets the same $r$-step Trotter circuit as $\lpd$ alone,
    so the truncation and Trotter error analyses of \cref{thm:runtime} apply verbatim with hypothesis (ii),
    at cost $\bigO(r_B\, n^{w^*})$ for the backward evolution.
    Finally, each of the $\bigO(n^{w^*})$ retained Pauli strings is evaluated on the $\mps$ by one transfer-matrix sweep costing $\bigO(n\chi^3)$.
    Summing the three parts gives \cref{mtd:eq:hybrid_runtime}.
\end{proof}
The two hypotheses match in scale.
Condition (ii) requires every $k$-qubit reduced state with $k = 2w^*\kh^{\Upsilon\Gamma}$ to be close to maximally mixed,
so a contiguous $k$-block has entanglement entropy $S_B \approx k$.
An $\mps$ of bond dimension $\chi$ has $S_B \le 2\log_2\chi$ on any contiguous block.
Hence a state in the $\lpd$ regime needs $\chi \ge 2^{k/2} = 2^{\Omega(w^*)}$ for an exact forward representation,
and the switching time is where the forward budget reaches the scale $\chi = 2^{\Theta(w^*)}$ at which the backward guarantee begins.
At this scale the runtime \cref{mtd:eq:hybrid_runtime} is polynomial in $n$ for constant $w^*$.

Both methods are individually limited to constant times.
In certain nonintegrable systems,
the entanglement entropy grows ballistically as
$S_A(t) \sim t$ \cite{kimBallisticSpreadingEntanglement2013},
so the $\mps$ bond dimension $\chi \ge 2^{S_A}$ grows exponentially with $t$
and the MPS simulation is efficient only for $t_F \in\Theta(1)$.
Meanwhile, $\lpd$ requires $t_B \in \Theta(1)$,
the $\lpd$ time limit from \cref{thm:runtime}.
The total simulation time is therefore limited to
$t \in \Theta(1)$.
When the target precision is $\eps \sim 1/\mathrm{poly}(n)$,
\cref{thm:truncation_threshold} gives $w^* = \bigO(\log n)$,
hence $\chi = 2^{\bigO(\log n)} = \mathrm{poly}(n)$
and the overall runtime is quasi-polynomial in $n$.

In practice the evaluation stage dominates the cost of the hybrid.
Each retained Pauli string is a tensor product, so its expectation in the $\mps$ is one transfer-matrix sweep, $\bigO(\ell\chi^3)$ with cached environments for a string spanning $\ell$ sites, and summed over the $\bigO(n^{w^*})$ retained strings this product of the two truncation budgets can dominate the wall-clock time at the largest scales.
Our implementation amortizes shared Pauli prefixes, merges symmetry-related strings, and batches the contractions, which introduce no further approximation.
When the retained operator exceeds the budget for exact contraction, we evaluate it by importance sampling: strings are drawn with probability proportional to $\abs{c_P}$, giving an unbiased estimator whose cost is set by the number of samples rather than by the number of strings, at the price of a statistical error that is controlled by the sample count and reported with the values.
Since independent samples need no communication beyond the final average, distributing them over several GPUs, each holding a copy of the $\mps$, is a straightforward extension that would reduce the wall-clock time of the evaluation essentially in proportion to the number of GPUs.

\mysection[3]{Comparison with cluster expansion}
On the QMFI instance of \cref{fig:typical_case_pauli}, we compare $\lpd$ with the cluster expansion of Wild and Alhambra \cite{wildClassicalSimulationShortTime2023}.
For product input states and local observables at $t<t^*=1/(2e\mathfrak{d})$,
where $\mathfrak{d}$ is the degree of the interaction graph of unit-norm local terms ($t^*\approx0.04$ for this model),
its runtime is independent of $n$ and polynomial in $1/\eps$.
This guaranteed scaling is stronger in both $n$ and $\eps$ than that of $\lpd$,
whose runtime $\bigO(r\,n^{w^*})$ with $w^*=\ko+\bigO(\log(1/\eps)/\log(t_0/t))$ is guaranteed only for $t<t_0$ and sufficiently entangled input states (\cref{thm:runtime}).
Beyond $t^*$, its analytic continuation reaches any constant time with a runtime polynomial in $1/\eps$ of degree $e^{\pi t/t^*}$,
so both guarantees are confined to constant evolution time.
Numerically, the cluster expansion is far more accurate at short times,
but double-precision cancellation stops its truncated series at $t\approx2.6$,
and its analytic continuation, even with parameters tuned against the exact solution, fails by $t=3$,
whereas $\lpd$ at the truncation threshold $w^*=5$ of \cref{fig:typical_case_pauli} stays within $10^{-1}$ of the ideal value up to $t=5$,
empirically and beyond its guarantee.
Moreover, the convergence bound of the cluster expansion is independent of the input state,
and its accuracy horizon is numerically unchanged from product to Haar-random input states,
whereas the Pauli truncation error of $\lpd$ decreases with the entanglement of the input state (\cref{thm:entangle_ob_bound,fig:hybrid}).
We put the figure and more details in \cref{apd:sec:ce_comparison}.

\vspace{10pt}
\emph{\textbf{Data availability}}.--
The data and code that support the findings of this study are available at \url{https://github.com/Jue-Xu/LPD}.
The Lean~4 formalization of the truncation-error analysis is available at \url{https://github.com/Jue-Xu/Lean4LPD}.

\emph{\textbf{Author contributions}}.--
J.X., Q.Z., C.Z., and S.Z. proposed the research.
J.X., Q.Z., and X.Z. completed the theoretical proofs of the entanglement bounds.
J.X., C.Z., and S.Z. completed the theoretical proofs of the Pauli truncation error bounds and runtime analysis.
J.X. performed numerical simulations.
All authors discussed the results and wrote the manuscript.

\emph{\textbf{Competing interests}}.--
The authors declare no competing interests.

\emph{\textbf{Acknowledgements}.}--
We thank Zoë Holmes, Daochen Wang, Wenjun Yu, Tianfeng Feng, and Zhongxia Shang for helpful discussions.
We also thank the anonymous referee who reviewed this work at both QCTiP and QSim for their constructive comments.
We acknowledge the \href{https://github.com/tbegusic/spd/blob/main/spd/SparsePauliDynamics.py}{SPD} python package by Begušić \cite{begusicRealTimeOperatorEvolution2025}
and also appreciate the review paper on Pauli propagation by Rudolph et al. \cite{rudolphPauliPropagationComputational2026} for the  \textsf{Julia} package \href{https://github.com/MSRudolph/PauliPropagation.jl}{PauliPropagation.jl}.
We also acknowledge Claude.ai for assistance.
J.X. and Q.Z. acknowledge funding from Innovation Program for Quantum Science and Technology via Project 2024ZD0301900, National Natural Science Foundation of China (NSFC) via Project No. 12347104 and No. 12305030, Guangdong Basic and Applied Basic Research Foundation via Project 2023A1515012185, Hong Kong Research Grant Council (RGC) via No. 27300823, N\_HKU718/23, and R6010-23, Guangdong Provincial Quantum Science Strategic Initiative No. GDZX2303007, HKU Seed Fund for Basic Research for New Staff via Project 2201100596.

\onecolumngrid
\appendix
\clearpage
\renewcommand{\addcontentsline}{\oldacl}
\renewcommand{\tocname}{Appendix Contents}
\tableofcontents
\numberwithin{theorem}{section}
\numberwithin{lemma}{section}
\numberwithin{corollary}{section}
\numberwithin{proposition}{section}
\numberwithin{definition}{section}

\section{Preliminaries}\label{apd:pre}

In this paper, we adopt the notion of Schatten norms to quantify the errors in different cases.
    The Schatten $p$-norm of an operator $A$ is defined as
    $\norm{A}_p:=\qty(\Tr(\abs{A}^p))^{1/p}=\qty(\sum_{i\ge 1} \sigma_i^p(A))^{1/p}$
    where $\sigma(A)$ are the singular values of $A$, i.e. the eigenvalues of $\abs{A}:=\sqrt{A^\dagger A}$.
We use the notation $\norm{\cdot}$ to denote the operator norm (Schatten $\infty$-norm) if not specified.
Besides, we extensively use the normalized Schatten 2-norm, also known as the Pauli 2-norm.
For simplicity, we just call it 2-norm in the rest of paper if no confusion arises.
\begin{definition}[Normalized Schatten 2-norm, or Pauli 2-norm]\label{def:norm}
        We denote the normalized Schatten 2-norm as $\nfnorm{A}:=2^{-n/2}\sqrt{\Tr(AA^\dagger)}$.
        It is equivalent to the Pauli 2-norm $\norm{A}_{\mathrm{Pauli},2}$ used in some references.
The Pauli $p$-norm is defined as
$\norm{A}_{\mathrm{pauli},p}:=\qty(\sum_{P\in \pbasis_n} \abs{\alpha_P}^p)^{1/p}$
where $\alpha_P$ is the non-zero coefficients of Paulis in $A$.
It is also known as the \emph{normalized Frobenius} or \emph{Hilbert-Schmidt norm}.
\end{definition}

Here are some useful properties of the Schatten norms:
\begin{itemize}
    \item

Hölder's inequality:
    $\forall 1/p+1/q=1$  and  $1\le p,q\le \infty$,
    $\abs{\Tr(A^\dagger B)}\le \norm{A}_p\norm{B}_q$.
\item
    As a corollary of the Hölder's inequality, the Schatten 1-norm (trace norm) satisfies
    $\norm{AB}_{1}\le \norm{A}_2\norm{B}_2$ (Cauchy-Schwarz inequality)
    and $\abs{\Tr(A\rho)}^2 \le \Tr(A^\dagger A \rho)$.
    And $\norm{AB}_{1}\le \norm{A}_1\norm{B}_{\infty}$.
    \item

Unitary invariance: For any unitary matrices $U$ and $V$ of appropriate dimensions $\norm{UAV}_p = \norm{A}_p$.

    \item
    Sub-multiplicativity:  For all $p \in [1, \infty]$
    $\norm{AB}_p \le \norm{A}_p \norm{B}_p$,
\end{itemize}

\begin{definition}[State $k$-design \cite{meleIntroductionHaarMeasure2024}]\label{def:state_k_design}
    Let $\eta$ be a probability distribution over a set of state $S\subseteq \mathbb{C}^d$.
    The distribution $\eta$ is said to be a state $k$-design
    (or also spherical $k$-design) if and only if:
    $
        \bbE_{\ket{\psi}\sim \statee_k} \qty[ \op{\psi}^{\otimes k} ]
        =
        \bbE_{\ket{\psi}\sim \statee_{\haar}} \qty[ \op{\psi}^{\otimes k} ].
        $
    Specifically, we denote the ensemble of 2-design states as $\statee_2$.
\end{definition}

An $n$-qubit Pauli operator
is $P\in\pbasis_n$ where $\pbasis_n:=\qty{I,X,Y,Z}^{\otimes n}$ and  $\qty{I,X,Y,Z}$ is the set of Pauli matrices.
\begin{definition}[Pauli basis]\label{def:pauli_basis}
    The (normalized) $n$-qubit Pauli basis is the set
    $\npbasis_n:= \qty{2^{-n/2}P:P\in\qty{ I,X,Y,Z}^{\otimes n}}$,
    that is,
    $\forall s, s' \in \npbasis_n, \; \Tr(s s')=\delta_{s,s'}$.
\end{definition}
We use the lower case letters for normalized Pauli operators such as $s\in\npbasis$ where we omit the subscript $n$ if no confusion arises,
while the upper case letters for non-normalized Pauli operators such as $P\in\pbasis$.
Throughout this paper, we will use both notations $\npbasis$ and $\pbasis$ accordingly to avoid introducing unnecessary normalization factors.

\begin{definition}[Support]\label{def:support}
    For an operator $O=\sum_{P\in \pbasis} \alpha_P P$,
    its \emph{support} is the union of the supports of all non-zero Pauli operators, where $\supp(P)$ denotes the set of qubits on which the Pauli string $P$ is non-identity, i.e.,
        $\supp(O):=\bigcup_{P:\alpha_P\neq 0} \supp(P)$.
\end{definition}
\begin{definition}[Pauli weight]\label{def:pauli_weight}
    The \emph{weight} of a Pauli operator $P$ is defined as the number of qubits on which it acts non-trivially (the number of non-identities), i.e., $|P|=\abs{\supp(P)}$.
    We call an operator $O=\sum_{P\in\pbasis}\alpha_P P$ \emph{$k$-local} if every Pauli operator with $\alpha_P\neq0$ has weight $\abs{P}\le k$ (not to be confused with $\abs{\supp(O)}$).
\end{definition}

\section{Proofs of lemmas and theorems}\label{apd:sec:proof}

In this section, we present the explicit proofs of the lemmas and theorems in the main text.

\subsection{Proof of Lemma 1}
In this subsection, we provide the Pauli 2-norm (equivalently, normalized Schatten 2-norm)
upper bound on the expectation value of any observable with respect to random states.
Furthermore, we prove the similar 2-norm bound for the expectation value of a local observable with an entangled state.
Therefore, we can bound the truncation error of our algorithm for entangled input states, without resorting to randomness.
\begin{lemma}[Upper bound on average expectation value]\label{apd:lem:2design_ob_bound}
    The average squared expectation value of any observables $O$
    over an ensemble of 2-design states $\ket{\psi}\sim \statee_2$ has the 2-norm upper bound
    \begin{equation}\label{apd:eq:2design_observable_bound}
        \bbE_{\ket{\psi}\sim \statee_2} \qty[\abs{\bra{\psi} O\ket{\psi}}^2]
        \le \nfnorm{O}^2,
    \end{equation}
    where the squared normalized Schatten (Pauli) 2-norm $\nfnorm{O}^2:=\frac{1}{2^n}\Tr(O^\dagger O)$.
\end{lemma}
\begin{proof}
    Directly by the property of 2-design states \cite{meleIntroductionHaarMeasure2024},
    the average squared expectation value has
    \begin{equation}
        \bbE_{\ket{\psi}\sim \statee_2} \qty[\abs{\bra{\psi} O \ket{\psi}}^2]
        = \frac{1}{2^n(2^n+1)} \qty(\Tr(O)^2 + \Tr(O^2))\le \nfnorm{O}^2 .
    \end{equation}
    as desired.
    Note that this 2-norm bound can be generalized to more general random states,
    i.e., the low-average ensemble of states defined in \cite{schusterPolynomialtimeClassicalAlgorithm2025}.
\end{proof}
Next, to show the Pauli 2-norm bound for a \emph{sufficiently entangled} state,
we need the following lemma on the upper bound of the expectation value of a \emph{local observable},
which is from \cite{zhaoEntanglementAcceleratesQuantum2025}.
Here, we restate the proof for completeness.
\begin{lemma}[Upper bound on local observable expectation value]\label{apd:lem:local_observable_2bound}
    Given a $w$-local observable $O= \sum_{j} O_j$
    where each $O_j$ acts non-trivially on the subsystem with $\supp(O_j)$ and each has weight at most $w$ (i.e., $\max_j \abs{\supp(O_j)} \le w$),
    the squared expectation value of $O$ with respect to any pure state $\ket{\psi}$ has the upper bound
    \begin{equation}\label{apd:eq:local_observable_bound}
        \abs{\bra{\psi} O \ket{\psi}}^2
        \le \nfnorm{O}^2 + \sum_{j,j'} \norm{O_j^\dagger O_{j'}} \Tr\abs{\rho_{j,j'}-\frac{I_{\supp(O^\dagger_j O_{j'})}}{2^{\abs{\supp(O^\dagger_j O_{j'})}}}}.
    \end{equation}
    where $\rho_{j,j'} := \Tr_{[n]\backslash \supp(O^\dagger_j O_{j'})}(\op{\psi})$ is the reduced density matrix (RDM) of $\ket{\psi}$ on the subsystem $\supp(O^\dagger_j O_{j'})$.
    One can further bound the second term (trace distance to the locally maximally mixed state) by the entanglement entropy of the subsystem as
    \begin{equation}\label{apd:eq:local_observable_bound_entropy}
        \abs{\bra{\psi} O \ket{\psi}}^2
        \le \nfnorm{O}^2 + \sum_{j,j'} \norm{O_j^\dagger O_{j'}}
        \sqrt{2 \abs{\supp(O^\dagger_j O_{j'})} -2S(\rho_{j,j'})}.
    \end{equation}
\end{lemma}
\begin{proof}
    First, for any $O=\sum_j O_j$ and $\ket{\psi}$, the square expectation value obeys
    \begin{equation}
        \abs{\bra{\psi} O \ket{\psi}}^2
        \le \bra{\psi} O^\dagger O \ket{\psi}
        = \sum_{j,j'} \bra{\psi} O_j^\dagger O_{j'} \ket{\psi},
    \end{equation}
    where each term in the sum, $O^\dagger_{j} O_{j'}$ acts non-trivially on the subsystem $\supp(O^\dagger_j O_{j'})$.
    Let us denote the local factor of each term by $L_{j,j'}:=2^{-(n-\abs{\supp(O^\dagger_j O_{j'})})}\,\Tr_{[n]\backslash \supp(O_j^\dagger O_{j'})}(O^\dagger_j O_{j'})$, i.e., the operator on $\supp(O^\dagger_j O_{j'})$ with $O^\dagger_j O_{j'}=L_{j,j'}\otimes I_{[n]\backslash\supp(O^\dagger_j O_{j'})}$ (the normalization $2^{-\abs{\supp^c}}$ makes the partial trace the actual tensor factor),
    so that $2^{-n}\cdot \Tr(O^\dagger_j O_{j'})=\Tr(L_{j,j'})\cdot 2^{-\abs{\supp(O^\dagger_j O_{j'})}}$ where $n$ is the complete system size.
    Then each (generally complex, for $j\ne j'$) term decomposes exactly as
    \begin{align}
        \bra{\psi} O_j^\dagger O_{j'} \ket{\psi}
        =& \Tr(L_{j,j'} \rho_{j,j'}) \tag{RDM} \\
        =& \Tr[L_{j,j'} (\rho_{j,j'} - \bar{I}_{jj'})] + \Tr(L_{j,j'})/2^{\abs{\supp(O^\dagger_j O_{j'})}} \\
        =& \Tr[L_{j,j'} (\rho_{j,j'} - \bar{I}_{jj'})] + 2^{-n}\cdot \Tr(O^\dagger_j O_{j'}) ,
    \end{align}
    where $\bar{I}_{jj'}:=\frac{I_{\supp(O^\dagger_j O_{j'})}}{2^{\abs{\supp(O^\dagger_j O_{j'})}}}$ denotes the maximally mixed state on the subsystem $\supp(O^\dagger_j O_{j'})$.
    Summing over $j,j'$, the second terms sum \emph{exactly} to $2^{-n}\Tr(O^\dagger O)=\nfnorm{O}^2$ (no orthogonality of the decomposition is needed), and bounding each first term by H\"older's inequality, $\abs{\Tr[L_{j,j'}(\rho_{j,j'}-\bar I_{jj'})]}\le\norm{L_{j,j'}}\Tr\abs{\rho_{j,j'}-\bar I_{jj'}}$ with $\norm{L_{j,j'}}=\norm{O_j^\dagger O_{j'}}$, gives the desired upper bound
    \begin{equation*}
        \abs{\bra{\psi} O \ket{\psi}}^2
        \le \nfnorm{O}^2 + \sum_{j,j'} \norm{O_j^\dagger O_{j'}}
        \Tr\abs{\rho_{j,j'}-\bar{I}_{jj'}}.
    \end{equation*}
    Moreover, the trace distance of $\rho_{j,j'}$ and $\bar{I}_{jj'}$ can be bounded by the relative entropy as
    \begin{equation}\label{apd:eq:trace_distance_entropy}
        \Tr\abs{\rho_{j,j'}-\bar{I}_{jj'}}
        \le \sqrt{2\ln(2)\, S\qty(\rho_{j,j'}||\bar{I}_{jj'})}
        =\sqrt{2\ln(2)\qty[\abs{\supp(O^\dagger_j O_{j'})} -S(\rho_{j,j'})]}
        \tag{Pinsker's ineq}
    \end{equation}
    where $S(\rho||\sigma):=\Tr(\rho \log \rho - \rho \log \sigma)$ is the relative entropy between $\rho$ and $\sigma$
    in base-2; the factor $\ln(2)$ converts to nats, for which Pinsker's inequality $\frac{1}{2}\qty(\Tr\abs{\rho-\sigma})^2\le \ln(2)\,S(\rho||\sigma)$ holds.
    Since $\ln(2)<1$, we relax the right-hand side to $\sqrt{2\abs{\supp(O^\dagger_j O_{j'})}-2S(\rho_{j,j'})}$ for a cleaner constant, which gives \cref{apd:eq:local_observable_bound_entropy}.
\end{proof}
Now, we have the similar 2-norm upper bound of a local expectation values for a sufficiently entangled state.
\begin{corollary}[Upper bound on local observable expectation value with entangled states]\label{apd:cor:entangle_ob_bound}
    Given a $w$-local observable $O= \sum_{j} O_j$ where each $O_j$ has weight at most $w$,
    if the input state is a $\Delta$-approximate $k$-uniform pure state $\ket{\psi_\Delta}$ with $\Delta\le \nfnorm{O}^2/(\sum_j\norm{O_j})^2$ and $k\ge 2w$,
    then the squared expectation value of $w$-local observable $O$ with respect to $\ket{\psi_\Delta}$ has the Pauli 2-norm upper bound
    $\abs{\bra{\psi_\Delta} O \ket{\psi_\Delta}}^2\le 2\nfnorm{O}^2.$
    In the words of entanglement entropy, if the input state $\ket{\psie}$ has the subsystem entanglement entropy $S(\rho_{j,j'})\ge\abs{\supp(O^\dagger_j O_{j'})}-\frac{1}{2}\nfnorm{O}^4/(\sum_j\norm{O_j})^4$ for all subsystem $\supp(O^\dagger_j O_{j'})$,
    then the same 2-norm bound holds
    \begin{equation}
    \label{apd:eq:entangle_observable_bound}
        \abs{\bra{\psie} O \ket{\psie}}^2
        \le 2\nfnorm{O}^2.
    \end{equation}
\end{corollary}
\begin{proof}
    For a $\Delta$-approximate $k$-uniform state
    (i.e., $\norm{\Tr_{[n]\backslash A}(\op{\psi_\Delta})-I_A/2^{\abs{A}}}_{\tr}\le \Delta$ for \emph{every} subset $A$ of at most $k$ qubits --- the condition is required on all size-$\le k$ subsystems, not one designated subset),
    the trace distance part (the second term in \cref{apd:eq:local_observable_bound}) is bounded by
    \begin{align}
        &\sum_{j,j'} \norm{O_j^\dagger O_{j'}} \Tr\abs{\rho_{j,j'}-I_{\supp(O^\dagger_j O_{j'})}/2^{\abs{\supp(O^\dagger_j O_{j'})}}} \\
        \le &\sum_{j,j'} \norm{O_j^\dagger}\norm{ O_{j'}} \Delta
        =\qty(\sum_j \norm{O_j})^2 \Delta \le \nfnorm{O}^2.
    \end{align}
    Then, we have the desired 2-norm bound
    $\abs{\bra{\psi_\Delta} O \ket{\psi_\Delta}}^2 \le 2\nfnorm{O}^2$,
    which is without the cross terms in \cref{apd:eq:local_observable_bound}
    and matches the average-case upper bound \cref{apd:eq:2design_observable_bound} (up to a constant factor).
    And the entanglement entropy condition $S(\rho_{j,j'})\ge\abs{\supp(O^\dagger_j O_{j'})}-\frac{1}{2}\nfnorm{O}^4/(\sum_j\norm{O_j})^4$ is required by
    \begin{equation}
        \sum_{j,j'} \norm{O_j^\dagger O_{j'}}
        \sqrt{2 \abs{\supp(O^\dagger_j O_{j'})} -2S(\rho_{j,j'})}
        \le \qty(\sum_j \norm{O_j})^2 \max_{j,j'}\sqrt{2 \abs{\supp(O^\dagger_j O_{j'})} -2S(\rho_{j,j'})}
        \le \nfnorm{O}^2.
    \end{equation}
\end{proof}

\subsection{Trotter formula error analysis}\label{apd:trotter_error}

The $p$th-order Trotter-Suzuki formula \cite{suzukiGeneralTheoryFractal1991} is constructed recursively as:
\begin{align}\label{apd:eq:suzuki}
    \pf_p(t)=\pf_{p-2}(u_pt)^2\pf_{p-2}((1-4u_p)t)\pf_{p-2}(u_pt)^2,
\end{align}
where $u_p=1/(4-4^{1/(p-1)})$ and the second-order Suzuki formula is $\pf_2(t)=\prod_{\gamma=1}^{\Gamma} e^{-\ii H_\gamma t/2} \prod_{\gamma=\Gamma}^{1} e^{-\ii H_\gamma t/2}$.
This recursive construction incurs a $\Upsilon=2 \cdot 5^{p/2-1}$ overhead of gate counts (depth) in the high-order Trotter formula (even $p\ge2$; $\Upsilon=1$ for the first-order formula $\pf_1$).
To balance the performance with the overhead of high-order product formula, it is usually to use the second-order Trotter-Suzuki formula.
The Trotter error the $p$th-order product formula is defined as $\opnorm{\pf_p - U}$
where the operator norm (also known as the Schatten $\infty$-norm) captures the worst-case (input state) analysis \cite{childsTheoryTrotterError2021}.
We first prove the worst-case Trotter error bound in expectation for any input state.
\begin{lemma}[Worst-case Trotter error in expectation]\label{apd:lem:worst_trotter_error}
    Given a Hamiltonian $H=\sum_{\gamma=1}^\Gamma H_{\gamma}$ in the decomposition into commuting groups, any input state $\rho$ and observable $O$,
    the worst-case Trotter error of the $r$-step $p$th-order Suzuki product formula \cref{apd:eq:suzuki} is bounded by
    \begin{equation}
        \abs{\Tr[(\pf^{\dagger r}_p O \pf_p^r - U^\dagger O U)\rho]}
        = \bigO(\opnorm{O} \Lambda \cdot t^{p+1}/r^{p})
    \end{equation}
    where $U:=e^{-\ii H t}$ is the ideal evolution unitary.
    Consequently, the required number of Trotter steps to achieve error tolerance $\eps$ in expectation for any input states is
    $r\in \bigO((\opnorm{O}\Lambda\eps^{-1})^{1/p} t^{1+1/p})$.
\end{lemma}
\begin{proof}
    We have the following upper bound on the worst-case Trotter error in expectation
    \begin{align}
        \abs{\Tr[(\pf_p^{\dagger r} O \pf_p^{r} - U^\dagger O U)\rho]}
        &\le \norm{\pf_p^{\dagger r} O \pf_p^{r} - U^\dagger O U}_{\infty} \cdot \norm{\rho}_1 \tag{Hölder's ineq} \\
        &= \norm{\pf_p^{\dagger r} O \pf_p^{r} - U^\dagger O U}_{\infty} \tag{$\norm{\rho}_1=1$} \\
        &\le \norm{\pf_p^{\dagger r} O \pf_p^{r} - \pf_p^{\dagger r} O U}_{\infty} + \norm{\pf_p^{\dagger r} O U - U^\dagger O U}_{\infty} \tag{Triangle ineq} \\
        &\le 2 \norm{O\pf_p^{r}-OU}_\infty \tag{Unitarity} \\
        &\le 2 \norm{O}_{\infty} \cdot \norm{\pf_p^r - U}_{\infty}
        \tag{Sub-multiplicativity} \\
        &= \bigO(\opnorm{O} \Lambda \cdot t^{p+1}/r^{p})
        \tag{$\opnorm{\pf_p^r- U}=\bigO(\Lambda\cdot t^{p+1}/r^{p})$},
    \end{align}
    where $\opnorm{\pf_p^r- U}$ is the worst-case Trotter error of any states and
    $\Lambda:=\sum_{\gamma_1,\dots,\gamma_{p+1}=1}^{\Gamma} \opnorm{[H_{\gamma_{p+1}},[H_{\gamma_p},...,[H_{\gamma_2},H_{\gamma_1}]]]}$ is the operator norm of the nested commutator \cite{childsTheoryTrotterError2021}.
\end{proof}
When considering the state information of the input state,
e.g., an ensemble of random states or an entangled state,
we can improve the worst-case Trotter error bound from the operator norm $\opnorm{\cdot}$ to the normalized Schatten (Pauli) 2-norm $\nfnorm{\cdot}$.
Assume an ensemble of random input states, the Theorem~3 of Ref.~\cite{zhaoHamiltonianSimulationRandom2021,yuObservabledrivenSpeedupsQuantum2025} proved the average-case Trotter error bound in expectation value.
We restate it as follow and refer to the original paper for the proof.
\begin{lemma}[Average-case Trotter error in expectation with random input state]\label{apd:thm:average_trotter_error}
    For a 2-design ensemble $\scrE_2$ of quantum states and an observable $O$,
    the average Trotter error
    $\bbE_{\rho\sim \statee_2}\qty[\abs{\Tr[(\pf_p^{\dagger r} O \pf^r_p - U^\dagger O U)\rho]}]$
    of the standard $r$-step $p$th-order Suzuki product formula \cref{apd:eq:suzuki} has the upper bound
        $\bigO(\fLambda \nfnorm{O} t^{p+1}r^{-p})$,
    where $\fLambda:=\sum_{\gamma_1,\dots,\gamma_{p+1}=1}^{\Gamma} \nfnorm{[H_{\gamma_{p+1}},[H_{\gamma_p},...,[H_{\gamma_2},H_{\gamma_1}]]]}$ is the normalized Schatten (Pauli) 2-norm of the nested commutator and $H=\sum_{\gamma=1}^\Gamma H_{\gamma}$ is the decomposition of Hamiltonian terms into commuting groups.
    Consequently, the required number of Trotter steps $r$ for average Trotter error $\eps$ is
    $ r\in \bigO((\fLambda\nfnorm{O}\eps^{-1})^{1/p} t^{1+1/p})$.
\end{lemma}
 Theorem~3 of \cite{zhaoEntanglementAcceleratesQuantum2025} showed the 2-norm bound of the Trotter error can be achieved without randomness.
 Instead, it requires the input state to be entangled by the \cref{apd:cor:entangle_ob_bound}.
\begin{lemma}[Pauli 2-norm bound of the Trotter error in expectation with an entangled state]\label{apd:thm:trotter_error_entangled}
    Let $O$ be a local observable.
    Assume that all intermediate Trotter-evolved states $\ket{\psi_d}:=\pf_p^{\,d}\ket{\psie}$, $d=0,\dots,r-1$, satisfy the entanglement condition of \cref{apd:cor:entangle_ob_bound} with respect to the local nested-commutator error operators of one Trotter step.
    Then the Trotter error in expectation obeys
    \begin{equation}
        \abs{\Tr[(\pf_p^{\dagger r} O \pf_p^r - U^\dagger O U)\op{\psie}]}
        =\bigO\qty(\fLambda\, \opnorm{O}\, t^{p+1}r^{-p}) .
    \end{equation}
    Consequently, the required number of Trotter steps to achieve error tolerance $\eps$ in expectation is
    $r\in \bigO((\fLambda\opnorm{O}\eps^{-1})^{1/p} t^{1+1/p})$.
\end{lemma}
\begin{proof}
    By H\"older's inequality and unitarity,
    \begin{align}
        \abs{\Tr[(\pf_p^{\dagger r} O \pf_p^r - U^\dagger O U)\op{\psie}]}
        &\le 2\opnorm{O}\cdot \norm{(\pf_p^{r}-U)\ket{\psie}}
        \tag{H\"older's ineq} \\
        &\le 2\opnorm{O}\sum_{d=0}^{r-1}\norm{\qty(\pf_p(t/r)-U(t/r))\ket{\psi_d}}
        \tag{Telescoping \& unitarity},
    \end{align}
    which reduces the problem to the per-step \emph{vector-norm} Trotter error evaluated on the intermediate states.
    By Theorem~1 and Corollary~1 of the Supplemental Material of Ref.~\cite{zhaoEntanglementAcceleratesQuantum2025} --- whose error operator is the representation \cite{childsTheoryTrotterError2021} of the Trotter remainder as a sum of \emph{local nested commutators} of Hamiltonian terms --- each step obeys
    $\norm{(\pf_p(t/r)-U(t/r))\ket{\psi_d}}=\bigO(\fLambda\,(t/r)^{p+1})$,
    provided $\ket{\psi_d}$ is sufficiently entangled on the supports of those commutators.
    Summing the $r$ steps yields the claim.
\end{proof}
\begin{remark}[Operator norm and the per-step route]
    (i) The bound features the operator norm $\opnorm{O}$: the normalized Schatten 2-norm is \emph{not} submultiplicative
    (e.g., $A=B=2^{n/2}\op{0}^{\otimes n}$ gives $\nfnorm{AB}=2^{n/2}\gg\nfnorm{A}\nfnorm{B}=1$),
    so a chain of the form $\nfnorm{\pf^{\dagger r}O\pf^r-U^\dagger OU}\le2\nfnorm{O}\nfnorm{\pf^r-U}$ is unavailable;
    the valid H\"older bound $\nfnorm{AB}\le\nfnorm{A}\opnorm{B}$ produces $\opnorm{O}$, which equals $\nfnorm{O}$ for Pauli observables.
    (ii) The entanglement condition applies to all intermediate states rather than only the input, because the per-step error operators act on $\ket{\psi_d}$.
    (iii) The exactly evolved $U^\dagger OU$ is not a bounded-weight observable, so \cref{apd:cor:entangle_ob_bound} cannot be applied to the difference $\pf_p^{\dagger r}O\pf_p^r-U^\dagger OU$ directly; the per-step vector-norm route above avoids this.
\end{remark}
\begin{remark}\label{apd:trotter_n_indep}
    The Trotter error bound can be further improved
    when the evolved observable is constrained in the low-weight subspace and one-step Trotter is a shallow circuit.
    For example, Ref.\cite{yuObservabledrivenSpeedupsQuantum2025} showed that the Trotter error in local observables of the short-time nearest-neighbor Hamiltonian dynamics is system-size independent.
\end{remark}

In addition, the error mitigation techniques can be applied to
not only physical errors
\cite{caiQuantumErrorMitigation2023,lerchEfficientQuantumEnhancedClassical2026,zhangCliffordPerturbationApproximation2024}
but also algorithmic errors \cite{endoMitigatingAlgorithmicErrors2019}.
Specifically, Ref.~\cite{watsonExponentiallyReducedCircuit2025} applied the extrapolation scheme to reduce the Trotter steps $r$ from $1/\eps$ dependence to $\poly\log(1/\eps)$.
In the framework of $\lpd$, we could use the symbolic calculation \cite{rudolphPauliPropagationComputational2026} to enable the extrapolation without repeating overhead (cf. \cite{zhangCliffordPerturbationApproximation2024}).

\subsection{Pauli truncation error analysis}\label{apd:pauli_error}
In this section, we analyze the Pauli truncation error incurred by our \emph{Low-weight Pauli Dynamics} $\lpd$ algorithm.
To begin with, we recall the important equation of Pauli rotation $e^{-\ii G \dt/2}$ generated by a Pauli operator $G\in\pbasis_n$ (so that $G^\dagger=G$ and $G^2=I$) acting on a Pauli operator $P\in\pbasis_n$
\begin{equation}\label{apd:eq:pauli_rotation_branch}
     e^{\ii G \dt/2} P e^{-\ii G \dt/2} =
    \begin{cases}
        P, & [G,P] = 0 \\
        \cos(\dt) P + \ii \sin(\dt) \cdot GP, & \qty{G,P}=0.
    \end{cases}
\end{equation}
\textbf{This equation has three crucial implications}:
(1) the transition is sparse in the Pauli basis;
(2) the weight change is limited by the locality of $G$;
(3) the transition amplitude is damped by $\sin(\dt)$.
More explicitly,
when the generator $G$ and $P$ commute $[G,P]=0$,
then no rotation occurs and the Pauli operator remains unchanged (no weight change);
when they anti-commute $\{G,P\}=0$,
then the weight of $P$ changes by at most $\kh-1$ if $G$ has weight at most $\kh$.
\textit{And importantly, if the rotation angle $\dt$ is small,
then the newly generated Pauli $GP$ has its coefficient damped by small factor $\sin(\dt)\sim \dt$.}
See the illustration in \cref{fig:one_layer}(a).

Next, we are going to prove the 2-norm upper bound on the Pauli truncation error in expectation for an ensemble of random states as well as an entangled state.

\begin{lemma}[2-norm upper bound of average Pauli Truncation error]\label{apd:thm:ob_norm_bound_expectation_error}
    The average Pauli error in the expectation value
    with respect to 2-design states $\statee_2$ is
    upper bounded by the normalized Schatten (Pauli) 2-norm of the two observables' difference
    \begin{equation}
        \bbE_{\rho\sim \statee_2}  \qty[\abs{\Tr[(\tO^{(r)} - \tO^{(r)}_{\le w^*})\rho]}]
        \le
        \nfnorm{\tO^{(r)} - \tO^{(r)}_{\le w^*}},
    \end{equation}
    where $\tO^{(r)}:=\pf^{\dagger r} O \pf^{r}$ is the $r$-step Trotter evolved observable without truncation
    with $\pf:=\pf_p(t/r)$ the one-step Trotter unitary of \cref{apd:eq:suzuki} (for the first-order formula, $\pf=\prod_{\gamma=1}^{\Gamma}e^{-\ii H_\gamma t/r}$, i.e.\ $\Gamma$ layers),
    and $\tO_{\le w^*}^{(r)}$ is the low-weight approximation of $\tO^{(r)}$ by $\lpd$
    (i.e., truncating high-weight Paulis above the threshold $w^*$ at the end of each Trotter step).
\end{lemma}
\begin{proof}
    By the Jensen's inequality,
    the average truncation error is upper bounded by the square root of the sum of the squared truncation errors as $\bbE[\abs{\Delta}]\le \sqrt{\bbE[\Delta^2]}$.
    \begin{align}
        \bbE_{\rho\sim \statee_2} \qty[\abs{\Tr[(\tO^{(r)} - \tO^{(r)}_{\le w^*})\rho]}]
        &\le \sqrt{\bbE_{\rho\sim \statee_2} \qty[\Tr[(\tO^{(r)} - \tO^{(r)}_{\le w^*})\rho]^2]} \tag{Jensen's ineq} \\
        &\le \nfnorm{\tO^{(r)} - \tO^{(r)}_{\le w^*}} \tag{\cref{apd:lem:2design_ob_bound}}.
    \end{align}
    The second inequality is by \cref{apd:lem:2design_ob_bound} that upper bounds the root squared mean expectation over 2-design states by the normalized Schatten 2-norm of the observable difference.
\end{proof}
Next, we have the similar 2-norm upper bound of the Pauli truncation error in expectation for an entangled state.
    We define the Pauli truncation error in expectation with an entangled state $\ket{\psie}$ as
    \begin{equation}\label{apd:eq:pauli_truncation_error_entangled}
        \Delta \tilde{\mu}_{\le w^*}
        := \abs{\bra{\psie}(\tO^{(r)} - \tO^{(r)}_{\le w^*})\ket{\psie}}.
    \end{equation}
\begin{proposition}[2-norm upper bound on the per-step truncated component in expectation with an entangled state]\label{apd:thm:ob_norm_bound_expectation_error_entangled}
    Let $\tO^{(d)}_{\le w^*}:=\Pi_{\le w^*}\qty(\pf^\dagger\,\tO^{(d-1)}_{\le w^*}\,\pf)$ denote the \emph{truncated} ($\lpd$) trajectory started at $\tO^{(0)}_{\le w^*}:=O$, for which we assume $\ko\le w^*$, where $\Pi_{\le w^*}$ is the projection onto Pauli weights $\le w^*$,
    and let
    \begin{equation}\label{apd:eq:step_component}
        \tO^{(d)}_{\ge w^*+1}:=(1-\Pi_{\le w^*})\;\pf^\dagger\, \tO^{(d-1)}_{\le w^*}\,\pf
    \end{equation}
    denote the complementary high-weight component truncated at the end of Trotter step $d$, so that $\tO^{(d)}_{\le w^*}+\tO^{(d)}_{\ge w^*+1}=\pf^\dagger\,\tO^{(d-1)}_{\le w^*}\,\pf$, which for $d\ge2$ differs from $\tO^{(d)}$ in general.
    If the $\kh$-local Hamiltonian is arranged into $\Gamma$ layers with pairwise disjoint supports, so that one Trotter step consists of $\Upsilon\Gamma$ such layers, then every Pauli operator in $\tO^{(d)}_{\ge w^*+1}$ has weight at most $w^*\kh^{\Upsilon\Gamma}$, uniformly in $r$ and $n$.
    If a pure state $\ket{\phi}$ is $\Delta$-approximate $k$-uniform with $k\ge 2 w^*\kh^{\Upsilon\Gamma}$ and $\Delta$ as required in \cref{apd:cor:entangle_ob_bound}, then
    \begin{equation}
        \abs{\bra{\phi}\tO^{(d)}_{\ge w^*+1}\ket{\phi}}\le \sqrt{2}\,\nfnorm{\tO^{(d)}_{\ge w^*+1}} .
    \end{equation}
\end{proposition}
\begin{proof}
    The weight bound follows from the light cone of a single Trotter step with disjoint-support layers (cf.\ the proof of \cref{apd:thm:lightcone}): the input of the step has weights $\le w^*$ after the previous truncation, and each of the $\Upsilon\Gamma$ layers multiplies the weight by at most $\kh$.
    The expectation bound is then \cref{apd:cor:entangle_ob_bound} applied to the bounded-weight observable $\tO^{(d)}_{\ge w^*+1}$.
    We emphasize that the corollary is \emph{not} applied to the total difference $\tO^{(r)}-\tO^{(r)}_{\le w^*}$: the untruncated $\tO^{(r)}$ contains Pauli weights that grow with $r$ (up to $\Theta(r)$ for geometrically local $H$ and $\ko\kh^{\Upsilon\Gamma r}$ in general), so for $r=\poly(n)$ it is not a bounded-weight observable and the entanglement bound does not apply to it.
\end{proof}
Therefore, the two input-state assumptions enter differently:
for \emph{random (2-design)} input states, \cref{apd:thm:ob_norm_bound_expectation_error} bounds the truncation error directly by the Pauli 2-norm of the total observable difference $\nfnorm{\tO^{(r)} -\tO_{\le w^*}^{(r)}}$ (no locality of the difference is needed there),
whereas for an \emph{entangled} input state the bound proceeds per Trotter step, through the bounded-weight components $\tO^{(d)}_{\ge w^*+1}$ of \cref{apd:thm:ob_norm_bound_expectation_error_entangled} evaluated on the intermediate evolved states (\cref{apd:thm:triangle} below); the two routes meet in the same sum $\sum_d\nfnorm{\tO^{(d)}_{\ge w^*+1}}$.
\begin{remark}[When the entanglement condition can be satisfied]\label{apd:rmk:entanglement_condition}
    The condition required in \cref{apd:thm:ob_norm_bound_expectation_error_entangled} is $\Delta$-approximate $k$-uniformity with $k=2w^*\kh^{\Upsilon\Gamma}\in\bigO(1)$ and $\Delta\le\nfnorm{\tO^{(d)}_{\ge w^*+1}}^2/(\sum_j\norm{O_j})^2$, where $\tO^{(d)}_{\ge w^*+1}=\sum_j O_j$ is the decomposition into terms grouped by support, in the notation of \cref{apd:cor:entangle_ob_bound}.
    ``Sufficiently entangled'' therefore means \emph{high entanglement of constant-size subsystems} --- every $k$-qubit reduced state close to maximally mixed (local scrambling) --- and \emph{not} global entanglement across extensive cuts; in particular it is compatible with the purity of the global state (exact pure $k$-uniform states exist for suitable pairs $(n,k)$, e.g.\ from graph-state and MDS-code constructions, though not for every $k\le n/2$ --- notably no $2$-uniform state of $n=4$ qubits exists).
    The condition must hold for all intermediate Trotter-evolved states $\pf^{\,d}\ket{\psie}$ (cf.\ \cref{apd:thm:triangle}): a state drawn from a 2-design ensemble satisfies it \emph{with high probability at every fixed intermediate time} (the ensemble property is preserved by any fixed unitary evolution, and the marginal bound below is exponentially small in $n$); it is physically natural for scrambling dynamics, where local entanglement grows and then saturates; and it can be monitored along a run through the purities of the $k$-qubit reduced states.
    Concrete states satisfying it include exact $k$-uniform stabilizer states ($\Delta=0$; cf.\ the input-state discussion in \cref{apd:thm:runtime}) and typical Haar or 2-design states, for which $\bbE\,\Tr\abs{\rho_k-I_k/2^k}\le\sqrt{2^{3k}/2^{n}}$ is exponentially small in $n$.
    Product states do not satisfy it, consistent with the numerics, where the product-state truncation error is the largest.
    For components $\tO^{(d)}_{\ge w^*+1}=\sum_j O_j$ with $J$ comparably weighted terms, the ratio $\nfnorm{\tO^{(d)}_{\ge w^*+1}}^2/(\sum_j\norm{O_j})^2$ shrinks as $1/J$ (more generally as $1/J_{\mathrm{eff}}$ with $J_{\mathrm{eff}}:=(\sum_j a_j)^2/\sum_j a_j^2$ for term weights $a_j$, which ranges from $1/J$ up to order one for strongly unequal weights), so $\Delta$ must be tracked in that regime.
\end{remark}
To bound the norm of the truncated high-weight Paulis, we introduce some notations.
Let
\begin{equation}\label{apd:eq:weight_w_component}
    O_{=w}:=\sum_{s\in \npbasis_{n},\abs{s}=w} \Tr(Os)s
\end{equation}
denote the set of Pauli operators (with coefficients) in $O$ of weights $=w$
and $O_{\ge w}$ denote the set of Pauli operators in $O$ of weight at least $w$.
The squared Pauli (normalized Schatten) 2-norm of $O_{=w}$ reads
$\nfnorm{O_{=w}}^2 =\frac{1}{2^n}\sum_{s\in \npbasis_{n},\abs{s}=w} \Tr(Os)^2$
(the normalization $2^{-n}$ follows from $\nfnorm{A}^2=2^{-n}\Tr(A^\dagger A)$ and the orthonormality $\Tr(ss')=\delta_{s,s'}$ of the normalized Pauli basis).
And the squared \emph{Pauli 2-norm of the high-weight component} of an observable $O$ reads
\begin{equation}
    \nfnorm{O_{\ge w^*+1}}^2
    =\frac{1}{2^n}\sum_{s\in \npbasis_{n},\abs{s}\ge w^*+1} \Tr(Os)^2
    \equiv\sum_{w= w^*+1}^{n} \nfnorm{O_{=w}}^2,
\end{equation}
where the sum is over all $n$-qubit normalized Pauli operators $s$ with weight $\ge w^*+1$.

\begin{figure*}[!t]
    \centering
    \includegraphics[width=0.95\linewidth]{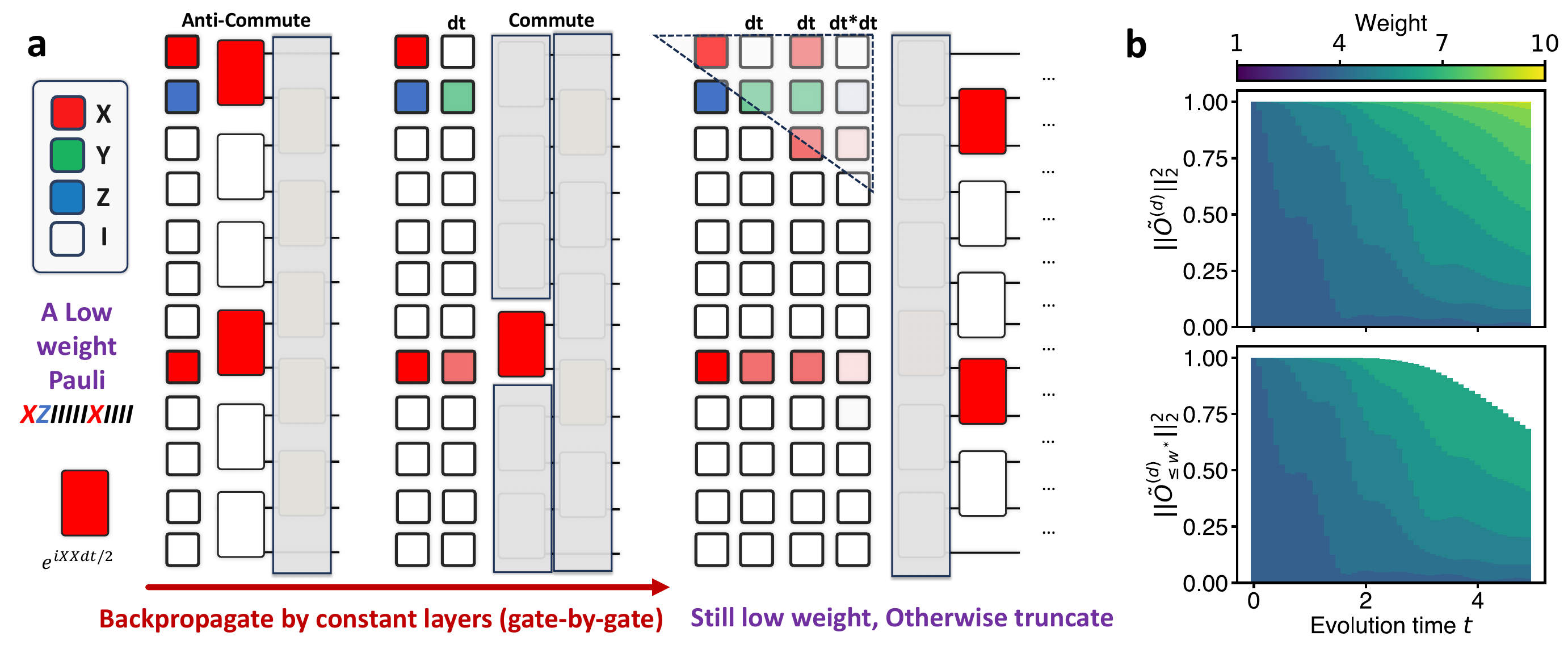}
    \caption{
    The Pauli branch and weight change in one Trotter step.
    (a) An illustration of constant layers of local Pauli rotations acting on a local (low-weight) Pauli operator.
    For simplicity, assume the initial Pauli operator is \textsf{XZIIIIIXIIII} with weight 3,
    the two-layer brickwork Pauli rotations are $e^{-\ii XX \dt/2}$ colored red.
    In the first layer, only two Pauli rotations (red) has support on the initial Pauli operator.
    One commutes with the initial Pauli operator and thus does not change it;
    the other anti-commutes and thus generates a new Pauli operator \textsf{IYIIIIIXIIII} with weight 2,
    the coefficient damped by $\sin(\dt)$ (with the lighter color).
    Similarly, in the second layer, one Pauli rotation commutes and the other anti-commutes,
    generating another two new Paulis.
    At the end of one Trotter step, all Pauli operators are low-weight (below the threshold) due to truncation.
    Since high-weight Paulis flowed from low-weight Paulis experienced many Pauli rotations, their contribution must be damped by small angle $\dt$ many times.
    (b)
    The 2-norm distribution over Pauli weights with(above)/without(below) truncation.
    Without truncation, the top figure shows that in the short-time regime,
    the low-weight Paulis (dark green) dominate the 2-norm distribution,
    while the high-weight Paulis (light green) gradually accumulate 2-norm in longer time.
    So, it is feasible to truncate high-weight Paulis in the short-time regime
    and Pauli truncation is unlikely to work in long-time (i.e., $t\sim n$) regime.
    The bottom figure shows that with the high-weight truncation,
    the 2-norm distribution is strictly limited to low-weight Paulis,
    and the 2-norm gradually loses as time evolves.
    }
    \label{fig:one_layer}
\end{figure*}

Since our $\lpd$ algorithm truncates high-weight Paulis at the end of each Trotter step,
we can bound the total Pauli truncation error by the sum of errors in every step as follows.
\begin{proposition}[Triangle bound of Pauli truncation error]\label{apd:thm:triangle}
    Given an input local observable $O$ and an input state $\ket{\psie}$ such that every intermediate Trotter-evolved state $\ket{\psi_{d}}:=\pf^{\,d}\ket{\psie}$, $d=0,\dots,r-1$, satisfies the entanglement condition of \cref{apd:thm:ob_norm_bound_expectation_error_entangled} for the truncated component $\tO^{(r-d)}_{\ge w^*+1}$,
    the total Pauli truncation error in expectation of $\lpd$ has the upper bound
    \begin{equation}
        \Delta \tilde{\mu}_{\le w^*} \le
        \sqrt{2}\sum_{d=1}^{r} \nfnorm{\tO_{\ge w^*+1}^{(d)}},
    \end{equation}
    the sum of the Pauli 2-norms of the high-weight components truncated at the end of every Trotter step.
\end{proposition}
\begin{proof}
    The difference between the untruncated and the truncated evolution telescopes in the Heisenberg picture:
    \begin{align}
        \tO^{(r)}-\tO^{(r)}_{\le w^*}
        &= \pf^\dagger\,\tO^{(r-1)}\,\pf-\qty(\pf^\dagger\,\tO^{(r-1)}_{\le w^*}\,\pf-\tO^{(r)}_{\ge w^*+1})
        \tag{By $\lpd$} \\
        &= \pf^\dagger\qty(\tO^{(r-1)}-\tO^{(r-1)}_{\le w^*})\pf+\tO^{(r)}_{\ge w^*+1}
        \tag{Rearrange} \\
        &= \pf^{\dagger 2}\qty(\tO^{(r-2)}-\tO^{(r-2)}_{\le w^*})\pf^{2}+\pf^\dagger\,\tO^{(r-1)}_{\ge w^*+1}\,\pf+\tO^{(r)}_{\ge w^*+1}
        \tag{Same for step $r-1$} \\
        &=\sum_{d=1}^{r} \pf^{\dagger(r-d)}\, \tO^{(d)}_{\ge w^*+1}\, \pf^{\,r-d},
        \tag{Iterate to step $0$}
    \end{align}
    where the first line writes both evolutions as one more Trotter step, $\tO^{(r)}=\pf^\dagger\,\tO^{(r-1)}\,\pf$ and $\tO^{(r)}_{\le w^*}=\pf^\dagger\,\tO^{(r-1)}_{\le w^*}\,\pf-\tO^{(r)}_{\ge w^*+1}$ (\cref{apd:eq:step_component}),
    and the iteration ends at $\tO^{(0)}-\tO^{(0)}_{\le w^*}=O-O=0$;
    each summand is the step-$d$ truncated component, further evolved by the remaining $r-d$ Trotter steps.
    Moving the residual evolution onto the state,
    \begin{align}
        \Delta \tilde{\mu}_{\le w^*}
        &\le \sum_{d=1}^{r}\abs{\bra{\psie}\,\pf^{\dagger(r-d)}\, \tO^{(d)}_{\ge w^*+1}\, \pf^{\,r-d}\,\ket{\psie}}
        =\sum_{d=1}^{r}\abs{\bra{\psi_{r-d}}\,\tO^{(d)}_{\ge w^*+1}\,\ket{\psi_{r-d}}}
        \tag{Schr\"odinger picture} \\
        &\le \sqrt{2} \sum_{d=1}^{r}\nfnorm{\tO^{(d)}_{\ge w^*+1}}
        \tag{\cref{apd:thm:ob_norm_bound_expectation_error_entangled}} ,
    \end{align}
    where each application of the entanglement bound involves only the bounded-weight observable $\tO^{(d)}_{\ge w^*+1}$, evaluated on the correspondingly evolved (entangled) state --- which is exactly the entanglement condition assumed in the proposition.
\end{proof}

Further, to upper bound the norm of high-weight components in the $\kh$-Hamiltonian dynamics,
we define the \emph{high-weight norm} of the evolved observable $O^{(g)}$ as the Pauli 2-norm of its component above the threshold $w_m$,
\begin{equation}\label{apd:eq:def_high_weight_norm}
    \cumnorm_{\ge m}^{(g)}:=\nfnorm{O^{(g)}_{\ge w_m+1}}
    =\qty(\sum_{w>w_m} \nfnorm{O_{=w}^{(g)}}^2)^{1/2},
\end{equation}
matching the definition of the high-weight component used in the main text,
where $w_m:=\ko+(m-1)\cdot(\kh-1)$ denotes the maximum weight that can be achieved by $m- 1$ anti-commuting $\kh$-local Pauli rotations from weight $\ko$
because every $\kh$-local Pauli rotation applied to a Pauli operator can increase the weight of the Pauli operator by at most $\kh-1$.
Then, we have a damped flow recursion as follows.
\begin{lemma}[Damped local norm flow]\label{apd:thm:local_flow_k_local}
    Given an initial $\ko$-local observable $O$ of Paulis with weights $\le\ko$
    and the product of local Pauli rotations $U_g:=\prod_{l=1}^{g} e^{-\ii G_l\cdot \dt/2}$ with each Pauli weight $\abs{G_l}\le\kh$, where $\dt\in[0,\pi/2]$ denotes the conjugation angle (conjugation by $e^{-\ii G\dt/2}$ rotates an anti-commuting Pauli by $\dt$, cf.\ \cref{apd:eq:pauli_rotation_branch}; for a physical rotation $e^{-\ii\alpha_l G_l\tau}$ the conjugation angle is $2\alpha_l\tau$, absorbed below in $\alpha:=2\max_l\abs{\alpha_l}$; rotations carrying different angles $\dt_l$ are covered by taking $\dt:=\max_l\abs{\dt_l}$, since $\abs{\sin(\dt_l)}\le\sin(\dt)$ for $\dt\le\pi/2$),
    consider the evolved observable $O^{(g)}:= U_g^\dagger O U_g$.
    The high-weight norm \cref{apd:eq:def_high_weight_norm}
    satisfies, for every rung $m\ge1$ and every $g\ge1$, the recursion relation
    \begin{equation}\label{apd:eq:worst_local_flow_k_local}
        \cumnorm_{\ge m}^{(g)}\le
        \cumnorm_{\ge m}^{(g-1)}+\sin(\dt)\cdot \cumnorm_{\ge m-1}^{(g-1)},
    \end{equation}
    where $w_m:=\ko+(m-1)(\kh-1)$.
\end{lemma}
\begin{proof}
    Let $x^{(g)}$ denote the real coefficient vector of $O^{(g)}$ in the Pauli basis, $O^{(g)}=\sum_{P\in\pbasis_n} x^{(g)}_P\, P$ with $x^{(g)}_P=2^{-n}\Tr(P\,O^{(g)})$ (below we write $s,s'$ for generic elements of $\pbasis_n$); since $\pbasis_n$ is orthonormal with respect to the normalized Hilbert--Schmidt inner product $\langle A,B\rangle:=2^{-n}\Tr(A^\dagger B)$, we have $\nfnorm{O^{(g)}}=\norm{x^{(g)}}_{\ell_2}$, with no stray dimensional factors.
    Conjugation by the rotation $e^{-\ii G_g\dt/2}$ acts on $x^{(g-1)}$ as an \emph{orthogonal} matrix $A$:
    every Pauli $s$ commuting with $G_g$ is fixed, while the anti-commuting Paulis split into disjoint unordered pairs $\qty{s,\; s'}$ with $s'=\pm\ii G_g s$ the canonical (signless) partner, on which $A$ is the planar rotation $x_s\mapsto \cos(\dt)\,x_s \mp \sin(\dt)\,x_{s'}$, cf.\ \cref{apd:eq:pauli_rotation_branch} (note $s'$ is again a Hermitian Pauli with $\qty{G_g,s'}=0$; the unordered pairing $s\leftrightarrow s'$ partitions the anti-commuting sector, the sign of $\pm\ii G_g s$ being absorbed into the matrix entry --- the operator-valued map $s\mapsto\ii G_g s$ itself squares to $-1$, not to the identity).
    Split the Pauli basis into the high-weight region $R:=\qty{s: \abs{s}>w_m}$ and its complement $B$, and write
    $x^{(g)}_R=A_{RR}\,x^{(g-1)}_R+A_{RB}\,x^{(g-1)}_B$.
    First, $\norm{A_{RR}}\le\norm{A}=1$, because a submatrix of an orthogonal matrix has operator norm at most one.
    Second, $A_{RB}$ has at most one non-zero entry in each row and each column (the pairing $s\leftrightarrow \ii G_g s$), each of magnitude at most $\sin(\dt)$;
    a non-zero entry connects $s\in R$ to its partner $\ii G_g s\in B$, whose weight is at least $\abs{s}-(\kh-1)>w_m-(\kh-1)=w_{m-1}$, because anti-commutation forces an odd --- hence non-zero --- number of sites at which $G_g$ and $s$ carry different non-identity Paulis, and at every such site the product is non-identity, so at most $\abs{G_g}-1\le\kh-1$ sites of $s$ can cancel.
    Hence
    $\norm{A_{RB}\,x^{(g-1)}_B}_{\ell_2}\le \sin(\dt)\cdot\nfnorm{O^{(g-1)}_{\ge w_{m-1}+1}}=\sin(\dt)\,\cumnorm^{(g-1)}_{\ge m-1}$,
    and the claim follows from Minkowski's inequality.
    Note also that the truncation $\Pi_{\le w^*}$ sets coefficients to zero, so it can only decrease every $\cumnorm_{\ge m}$; the recursion therefore holds along the whole $\lpd$ trajectory with truncations interleaved.
\end{proof}
For one whole \emph{layer} of disjointly supported rotations, the inflow into the high-weight region can moreover be bounded by an operator norm --- uniformly in the system size and in the number of Pauli terms of the observable.
This is the key to the system-size-independent per-step accounting used in \cref{apd:thm:one_step_truncation_error} below.
\begin{lemma}[Layer inflow]\label{apd:thm:layer_inflow}
    Let $V_\gamma=\prod_{l\in\gamma}e^{-\ii \alpha_l G_l \tau}$ be one layer of Pauli rotations with each Pauli weight $\abs{G_l}\le\kh$ and pairwise \emph{disjoint supports} (hence mutually commuting), each with conjugation angle $\dt_l:=2\alpha_l\tau$ of magnitude $\abs{\dt_l}\le\dt\le\pi/2$, and let $A$ denote its orthogonal action on Pauli coefficient vectors.
    For the splitting $R=\qty{s:\abs{s}>w_m}$, $B=\pbasis_n\setminus R$, decompose $A_{RB}=\sum_{j\ge1}A^{(j)}$, where $A^{(j)}$ collects the branches taking exactly $j$ anti-commuting rotations of the layer (raising the weight by at most $j(\kh-1)$, i.e., $j$ rungs).
    Then
    \begin{equation}\label{apd:eq:layer_inflow}
        \norm{A^{(j)}}\le \binom{w_{m+j}}{j}\sin^j(\dt)\le \frac{\qty(w_{m+j}\sin(\dt))^j}{j!} .
    \end{equation}
\end{lemma}
\begin{proof}
    Since the layer terms commute, conjugating a Pauli $s'$ by $V_\gamma$ gives the exact branch expansion
    \begin{equation}
        V_\gamma^\dagger\, s'\, V_\gamma=\sum_{T\subseteq \mathcal{A}(s')}\qty(\prod_{l\in T}\sin(\dt_l))\qty(\prod_{l\in\mathcal{A}(s')\setminus T}\cos(\dt_l))\,(\pm)\qty(\prod_{l\in T}\ii G_l)\,s' ,
    \end{equation}
    where $\mathcal{A}(s')$ is the set of layer terms anti-commuting with $s'$.
    Because the supports within the layer are disjoint, each qubit of $s'$ meets at most one term, so $\abs{\mathcal{A}(s')}\le\abs{s'}$.
    We bound $\norm{A^{(j)}}\le\sqrt{\norm{A^{(j)}}_{1\to1}\,\norm{A^{(j)}}_{\infty\to\infty}}$ (Schur test) by the maximum column and row absolute sums.
    A column (fixed source $s'\in B$, $\abs{s'}\le w_m$) has at most $\binom{\abs{s'}}{j}\le\binom{w_m}{j}$ non-zero entries, each of magnitude at most $\sin^j(\dt)$.
    A row (fixed destination $s\in R$ reachable by a $j$-jump, so $\abs{s}\le w_m+j(\kh-1)=w_{m+j}$) has its sources of the form $s'=\qty(\prod_{l\in T}\ii G_l)s$ with $T\subseteq\mathcal{A}(s)$, $\abs{T}=j$: at most $\binom{w_{m+j}}{j}$ entries, each of magnitude at most $\sin^j(\dt)$.
    Combining, $\norm{A^{(j)}}\le\sin^j(\dt)\sqrt{\binom{w_m}{j}\binom{w_{m+j}}{j}}\le\binom{w_{m+j}}{j}\sin^j(\dt)$.
\end{proof}
The recursion relations lead to the following upper bounds on the high-weight norm.
\begin{corollary}[High-weight norm cumulation]\label{apd:cor:norm_cumulation_jump}
    Given an initial $\ko$-local observable $O$ (so $\cumnorm^{(0)}_{\ge m}=0$ for $m\ge1$) and conjugation angles of modulus at most $\dt\le\pi/2$:
    (i) after any $g$ single Pauli rotations,
    \begin{equation}
        \cumnorm_{\ge m}^{(g)} \le \binom{g}{m} \sin^m(\dt)\, \nfnorm{O};
    \end{equation}
    (ii) if within every layer only the single-jump inflow is counted (the blocks $A^{(1)}$ of \cref{apd:thm:layer_inflow}), then after any $T$ disjoint-support layers (see \cref{apd:rmk:multijump} for the additive multi-jump and overshoot sectors, which are not bounded by this display but are collected in the majorant \cref{apd:eq:composition_majorant} and absorbed into the constant $c_0$ of \cref{apd:thm:one_step_truncation_error}),
    \begin{equation}\label{apd:eq:layer_cumulation}
        \cumnorm_{\ge m}^{(T)} \le \binom{T}{m}\sin^m(\dt) \prod_{j=2}^{m+1} w_j \cdot \nfnorm{O} .
    \end{equation}
    Both bounds hold along the $\lpd$ trajectory with truncations interleaved, since truncation only decreases every $\cumnorm_{\ge m}$ and the total Pauli 2-norm.
\end{corollary}
\begin{proof}
    (i) Induction over $m$.
    \emph{Base case} $m=0$: $\cumnorm_{\ge0}^{(l)}\le\nfnorm{O^{(l)}}=\nfnorm{O}$ for all $l$, by unitary invariance of the 2-norm (this is where the single-norm definition \cref{apd:eq:def_high_weight_norm} is essential: the previous sum-of-norms quantity admits no bound independent of $l$ and the system size, since the observable spreads over up to $\ko+l(\kh-1)$ weight classes).
    \emph{Inductive step}: iterating \cref{apd:eq:worst_local_flow_k_local} and inserting the induction hypothesis,
    \begin{equation}
        \cumnorm^{(g)}_{\ge m}\le\sin(\dt)\sum_{l=0}^{g-1}\cumnorm^{(l)}_{\ge m-1}
        \le\sin^m(\dt)\sum_{l=0}^{g-1}\binom{l}{m-1}\,\nfnorm{O}=\binom{g}{m}\sin^m(\dt)\,\nfnorm{O}
    \end{equation}
    by the hockey-stick identity.
    (ii) The same induction, per layer: by \cref{apd:thm:layer_inflow}, the single-jump inflow into rung $m$ has operator norm at most $\binom{w_{m+1}}{1}\sin(\dt)=w_{m+1}\sin(\dt)$, so climbing rung $j$ contributes a factor $w_{j+1}\sin(\dt)$ and $T$ layers contribute the binomial:
    $\cumnorm_{\ge m}^{(T)}\le \binom{T}{m}\prod_{j=1}^{m}\qty(w_{j+1}\sin(\dt))\,\nfnorm{O}$.
\end{proof}
\begin{remark}[Multi-jump inflow]\label{apd:rmk:multijump}
    Within one layer a source can take $j\ge2$ anti-commuting rotations and climb $j$ rungs at once, with block norm at most $\eps_j^{(m)}:=(w_{m+j}\sin\dt)^j/j!$ by \cref{apd:eq:layer_inflow}.
    Such a $j$-jump carries the \emph{same} total power $\sin^j(\dt)$ as $j$ consecutive single jumps --- what it changes is the layer-slot combinatorics, not the order in the small parameter --- and it feeds rung $m$ from rung $m-j$ rather than from rung $m-1$.
    Multi-jumps are therefore \emph{additive} inflows and cannot be written as a multiplicative correction to \cref{apd:eq:layer_cumulation}: that bound vanishes for $T<m$ (the binomial $\binom{T}{m}$ is zero), whereas a single layer can already truncate mass (we call the normalized Pauli 2-norm of the component of an operator on a set of Pauli strings its \emph{mass} on that set; conjugation conserves the total mass and truncation only decreases it), as for $O=Z_1Z_2Z_3$ under the disjoint layer $\qty{X_1X_4,X_2X_5,X_3X_6}$ with $w^*=5$ ($m^*=2$), where the triple jump produces a weight-$6$ Pauli with coefficient $\sin^3(\dt)$.
    The lemmas below carry these sectors through the ladder:
    \cref{apd:thm:first_passage} retains all jump lengths in the recursion and solves it by a first-passage majorant;
    \cref{apd:thm:entry_factor,apd:thm:part_factor} are the two elementary estimates that control its excess over the single-jump sector;
    and \cref{apd:thm:multijump_sum,apd:thm:c0_bound} sum the inflow over the Trotter steps and collect the excess in the constant $c_0\le2$ of \cref{apd:thm:one_step_truncation_error}, which tends to $1$ as $r\to\infty$.
    Note that the combinatorial part of the correction is independent of $\dt$ --- it is a layer-counting, not a small-parameter, effect --- whereas the entry overshoot is controlled by the condition $r\ge8e^2w_2\,\alpha t$, which keeps the ladder's expansion parameter $w_2\sin(\dt)$ below one.
    The finite-$r$ conditions are mild: $m^*\in\bigO(\log(1/\eps)/\log(t_0/t))$ is polylogarithmic, $w_2\in\bigO(1)$, and $r$ must grow polynomially for Trotter accuracy anyway;
    moreover, for $t<t_0$ the last condition is implied by $r\ge8e^2w_2/(\Gamma(\kh-1))$, since $c_0\ge1$ gives $\alpha t<1/(\Gamma(\kh-1))$, so all conditions are independent of $t$ as well.
\end{remark}
Throughout the following lemmas, $\kh\ge2$, the truncation threshold is parameterized as $w^*=\ko+m^*(\kh-1)=w_{m^*+1}$, and $w_2:=\ko+(\kh-1)$, $c:=\ko/(\kh-1)$, so that $w_m=(\kh-1)(m-1+c)$, and $\beta:=2e\,w_2\sin(\dt)$.
\begin{lemma}[First-passage majorant]\label{apd:thm:first_passage}
    Under the assumptions of \cref{apd:thm:layer_inflow} and $\beta<1$ (so that the entry factors $E_\nu$ below converge), consider the $\lpd$ trajectory of a $\ko$-local observable $O$ (with truncations interleaved), let $\cumnorm_{\ge m}^{(T)}$ be its high-weight norm \cref{apd:eq:def_high_weight_norm} after $T$ layers, and set $\cumnorm_{\ge\nu}^{(T)}:=\nfnorm{O}$ for $\nu\le0$ (the reservoir of Paulis below the first rung).
    Then for every rung $m\ge1$ and every layer,
    \begin{equation}\label{apd:eq:multijump_recursion}
        \cumnorm_{\ge m}^{(T)}\le\cumnorm_{\ge m}^{(T-1)}+\sum_{j\ge1}\eps_j^{(m)}\,\cumnorm_{\ge m-j}^{(T-1)},
        \qquad \eps_j^{(m)}:=\frac{\qty(w_{m+j}\sin(\dt))^j}{j!},
    \end{equation}
    and consequently, for every $m\ge1$ and $T\ge0$,
    \begin{equation}\label{apd:eq:composition_majorant}
        \cumnorm_{\ge m}^{(T)}\le\nfnorm{O}\sum_{k=1}^{\min(m,T)}\binom{T}{k}\sum_{0<\nu_1<\dots<\nu_k=m}E_{\nu_1}\prod_{i=2}^{k}\eps_{\nu_i-\nu_{i-1}}^{(\nu_i)},
        \qquad E_\nu:=\sum_{j\ge\nu}\eps_j^{(\nu)} .
    \end{equation}
    The right-hand side of \cref{apd:eq:composition_majorant} is non-decreasing in $T$, and its all-ones term ($k=m$, all $\nu_i=i$) is \cref{apd:eq:layer_cumulation} times the entry-factor inflation $E_1/\eps_1^{(1)}$.
\end{lemma}
\begin{proof}
    \emph{Recursion.}
    Split the Pauli basis at $w_m$ into $R=\qty{s:\abs{s}>w_m}$ and $B=\pbasis_n\setminus R$ as in \cref{apd:thm:layer_inflow}, and decompose the inflow block of the layer as $A_{RB}=\sum_{j\ge1}A^{(j)}$, with $\norm{A^{(j)}}\le\eps_j^{(m)}$ by \cref{apd:eq:layer_inflow}.
    A $j$-jump raises the weight by at most $j(\kh-1)$, so every source of a destination above $w_m$ has weight above $w_m-j(\kh-1)=w_{m-j}$ and is counted by $\cumnorm_{\ge m-j}^{(T-1)}$;
    for $j\ge m$ the source may lie anywhere below the first rung, and its norm is bounded by the total Pauli 2-norm of the current observable, which is at most $\nfnorm{O}$ by unitary invariance and because truncation only removes coefficients.
    The diagonal block obeys $\norm{A_{RR}}\le1$ (a submatrix of an orthogonal matrix), truncation only decreases every $\cumnorm_{\ge m}$, and Minkowski's inequality gives \cref{apd:eq:multijump_recursion}.
    \emph{Majorant.}
    Induction on $T$, inserting \cref{apd:eq:multijump_recursion}; for $T=0$ both sides vanish since $O$ is $\ko$-local.
    A path with $k$ jumps either takes no jump in the last layer, which contributes the slot count $\binom{T-1}{k}$, or takes its last jump, of length $j$, into rung $m$ in that layer:
    for $1\le j<m$ this contributes $\eps_j^{(m)}$ times the majorant of rung $m-j$ with $k-1$ jumps,
    and for $j\ge m$ the path starts in the reservoir, the whole first jump --- including every overshoot $j>m$, which lands above rung $m$ from below rung $1$ and is absent from a plain composition sum --- being collected in the entry factor $E_m$.
    Pascal's rule $\binom{T-1}{k}+\binom{T-1}{k-1}=\binom{T}{k}$ closes the induction:
    each path leaves the reservoir once and then climbs through strictly increasing rungs, at most one jump per layer.
    Monotonicity in $T$ follows from $\binom{T}{k}\le\binom{T+1}{k}$ and the non-negativity of every term.
    For the all-ones term, $\eps_1^{(\nu)}=w_{\nu+1}\sin(\dt)$ gives $\prod_{\nu=1}^{m}\eps_1^{(\nu)}=\sin^m(\dt)\prod_{j=2}^{m+1}w_j$ with $\binom{T}{m}$ layer slots, i.e., \cref{apd:eq:layer_cumulation}, and the entry factor replaces $\eps_1^{(1)}$ by $E_1$.
\end{proof}
\begin{lemma}[Entry factor]\label{apd:thm:entry_factor}
    If $\beta<1$, then $\eps_{j+1}^{(\nu)}\le\beta\,\eps_j^{(\nu)}$ for all $j\ge\nu\ge1$, and hence
    \begin{equation}\label{apd:eq:entry_bound}
        E_\nu\le\frac{\eps_\nu^{(\nu)}}{1-\beta}\le\eps_\nu^{(\nu)}\,(1+2\beta)\le\eps_\nu^{(\nu)}\,(1+4e\beta),
    \end{equation}
    where the last two inequalities assume $\beta\le1/2$.
    Moreover, if $\sin(\dt)\le\alpha t/r$, then the step-number condition $r\ge8e^{2}w_2\,\alpha t$ implies $\beta\le1/(4e)$, i.e., $4e\beta\le1$.
\end{lemma}
\begin{proof}
    We have
    $\eps_{j+1}^{(\nu)}/\eps_j^{(\nu)}=\qty(w_{\nu+j+1}/w_{\nu+j})^{j}\cdot w_{\nu+j+1}\sin(\dt)/(j+1)$.
    Since $\nu\ge1$ and $c\ge0$, $w_{\nu+j+1}/w_{\nu+j}=1+1/(\nu+j-1+c)\le1+1/j$, so the first factor is at most $(1+1/j)^j\le e$.
    Since $\nu\le j$, $w_{\nu+j+1}=(\kh-1)(\nu+j+c)\le(\kh-1)(2j+c)\le2j\,w_2$, so the second factor is at most $2w_2\sin(\dt)\,j/(j+1)\le\beta/e$.
    Summing the geometric series gives the first inequality of \cref{apd:eq:entry_bound}; the others are $1/(1-\beta)\le1+2\beta$ for $\beta\le1/2$ and $2\le4e$.
    Finally, $\sin(\dt)\le\alpha t/r\le1/(8e^{2}w_2)$ gives $\beta=2e\,w_2\sin(\dt)\le1/(4e)$.
\end{proof}
\begin{lemma}[Part factor]\label{apd:thm:part_factor}
    A part of size $j\ge1$ ending at the boundary $\sigma\ge j$ costs, relative to the $j$ single jumps it replaces,
    \begin{equation}\label{apd:eq:part_factor}
        \rho(j,\sigma):=\frac{\eps_j^{(\sigma)}}{\prod_{\nu=\sigma-j+1}^{\sigma}\eps_1^{(\nu)}}
        =\frac{(\sigma+j-1+c)^j}{j!\,\prod_{i=\sigma-j+1}^{\sigma}(i+c)}
        \le\qty(\frac94)^{j-1},
    \end{equation}
    with equality for $j=1$ and for $j=\sigma=2$, $\ko=0$.
\end{lemma}
\begin{proof}
    The equality is direct substitution: both sides carry $\sin^j(\dt)$ and $(\kh-1)^j$, which cancel.
    With $s:=\sigma+c\ge j$ we have $\rho=\frac{1}{j!}\prod_{l=0}^{j-1}\frac{s+j-1}{s-l}$, and every factor $1+(j-1+l)/(s-l)$ is decreasing in $s$, so $\rho(j,\sigma)\le\rho(j,j)\big|_{c=0}=(2j-1)^j/(j!)^2=:f(j)$.
    Now $f(1)=1$, $f(2)=9/4$, and for $j\ge2$,
    $f(j+1)/f(j)=\frac{2j+1}{(j+1)^2}\qty(1+\frac{2}{2j-1})^{j}\le\frac{2j+1}{(j+1)^2}\,e^{2j/(2j-1)}\le\frac59\,e^{4/3}<\frac94$,
    because both $(2j+1)/(j+1)^2$ and $2j/(2j-1)$ decrease in $j$; induction gives $f(j)\le(9/4)^{j-1}$.
    (It is the $(j-1)$-th root of $f(j)$ that is maximal at $j=2$; $f$ itself peaks at $f(4)\approx4.17$.)
\end{proof}
\begin{lemma}[Summed inflow with all sectors]\label{apd:thm:multijump_sum}
    Let one Trotter step consist of $\Gamma$ layers of $\kh$-local Pauli rotations with pairwise disjoint supports and conjugation angles of magnitude at most $\dt\le\pi/2$, followed by the truncation at $w^*=w_{m^*+1}$, let $O$ be $\ko$-local with $\ko\le w^*$, and assume $\sin(\dt)\le\alpha t/r$ and $\beta\le1/2$ (the input state does not enter).
    Then the mass truncated over all $r$ steps obeys
    \begin{equation}\label{apd:eq:multijump_sum}
        \sum_{d=1}^{r} \nfnorm{\tO_{\ge w^*+1}^{(d)}}
        \le (1+4e\beta)\,\exp\qty[\frac94\,\frac{(m^*+1)^2}{(r+1)\,\Gamma}]\,
        \qty(\frac{r+1}{r}\,\Gamma\,\alpha t)^{m^*+1}\,\frac{\prod_{j=2}^{m^*+2}w_j}{(m^*+1)!}\cdot\nfnorm{O} .
    \end{equation}
\end{lemma}
\begin{proof}
    Write $m:=m^*+1$.
    At the beginning of every step $\cumnorm_{\ge m}=0$ (by the previous truncation, and for the first step because $O$ is $\ko$-local with $\ko\le w^*$), so iterating \cref{apd:eq:multijump_recursion} over the $\Gamma$ layers of step $d$ gives
    $\nfnorm{\tO^{(d)}_{\ge w^*+1}}\le\Gamma\sum_{j\ge1}\eps_j^{(m)}\,\overline{\cumnorm}_{\ge m-j}^{(d\Gamma)}$,
    where $\overline{\cumnorm}_{\ge\nu}$ denotes the right-hand side of \cref{apd:eq:composition_majorant} for $\nu\ge1$ and $\nfnorm{O}$ for $\nu\le0$;
    bounding every layer of the step by the majorant at its end is legitimate because the majorant, unlike the actual norms, is monotone in $T$.
    Merging the final jump with the chain below it, the right-hand side becomes a sum over the ordered decompositions $m=j_1+\dots+j_K$ into $K$ parts, $1\le K\le m$, with $\sigma_i:=j_1+\dots+j_i$:
    a decomposition with $K$ parts has weight $E_{\sigma_1}\prod_{i\ge2}\eps_{j_i}^{(\sigma_i)}$ and $\binom{d\Gamma}{K-1}$ layer slots.
    We sum over the steps first, with $\sum_{d=1}^{r}\binom{d\Gamma}{K-1}\le\qty((r+1)\Gamma)^{K}/(\Gamma\,K!)$, so that no sector is ever compared with a vanishing binomial.
    By \cref{apd:eq:entry_bound} (charged once per path) and \cref{apd:eq:part_factor}, the weight of a decomposition is at most $(1+4e\beta)\,(9/4)^{m-K}\sin^{m}(\dt)\prod_{j=2}^{m+1}w_j$.
    With the excess $D:=m-K$, there are $\binom{m-1}{D}\le m^{D}/D!$ decompositions, and $m!/K!\le m^{D}$, so the sector $D$ is at most $\frac{1}{D!}\qty[\frac94\,m^2/((r+1)\Gamma)]^{D}$ times the single-jump sector $D=0$, which equals $(1+4e\beta)\,\qty((r+1)\Gamma\sin(\dt))^{m}\prod_{j=2}^{m+1}w_j/m!$ (the prefactor $\Gamma$ cancelling the $1/\Gamma$ of the slot sum).
    Summing over $D\ge0$ and inserting $\sin(\dt)\le\alpha t/r$, all powers of $r$ collect into $((r+1)/r)^{m}$, which is \cref{apd:eq:multijump_sum}.
\end{proof}
\begin{lemma}[The constant $c_0$]\label{apd:thm:c0_bound}
    Relaxing $(r+1)\Gamma\ge r\Gamma$ in \cref{apd:eq:multijump_sum}, the correction relative to the single-jump sector is
    \begin{equation}\label{apd:eq:multijump_factor}
        (1+4e\beta)\,\exp\qty[\frac94\,\frac{(m^*+1)^2}{r\,\Gamma}]
        \qquad\text{with}\qquad
        r\;\ge\;\frac{8\,(m^*+1)^2}{\Gamma}
        \quad\text{and}\quad
        r\;\ge\;8e^{2}w_2\,\alpha t ,
    \end{equation}
    and distributing the $(m^*+1)$-th root over the three $r$-dependent factors defines
    \begin{equation}\label{apd:eq:c0}
        c_0:=\frac{r+1}{r}\,\exp\qty[\frac94\,\frac{m^*+1}{r\,\Gamma}]\,(1+4e\beta)^{1/(m^*+1)}
        \;\le\;\frac65\cdot\frac{64}{55}\cdot\sqrt2\;<\;2
        \qquad\text{for } m^*\ge1 \text{ and } r\ge5
    \end{equation}
    under the two conditions of \cref{apd:eq:multijump_factor}.
\end{lemma}
\begin{proof}
    $r\ge5$ gives $(r+1)/r\le6/5$.
    The first condition gives $\frac94\frac{m^*+1}{r\Gamma}\le\frac{9}{32(m^*+1)}\le\frac{9}{64}$ for $m^*\ge1$, and $e^{y}\le1/(1-y)$ gives $\exp[\cdot]\le64/55$.
    The second condition gives $4e\beta\le1$ by \cref{apd:thm:entry_factor}, so $(1+4e\beta)^{1/(m^*+1)}\le2^{1/2}$ for $m^*\ge1$.
    The product is $\frac{384}{275}\sqrt2\approx1.975$.
    Neither $m^*\ge1$ nor $r\ge5$ can be dropped: with $4e\beta=1$, $(r,m^*,\Gamma)=(5,0,2)$ gives $c_0\approx3.01$ and $(4,1,8)$ gives $c_0\approx2.03$, both satisfying the two conditions of \cref{apd:eq:multijump_factor}.
\end{proof}

With the layer-inflow bound \cref{apd:thm:layer_inflow} and the first-passage majorant \cref{apd:eq:composition_majorant}, we can bound the mass truncated at the end of each Trotter step, and hence --- via the triangle bound \cref{apd:thm:triangle} --- the total Pauli truncation error, in terms of the evolution time and the locality parameters.
\begin{theorem}[Pauli truncation error of all Trotter steps]\label{apd:thm:one_step_truncation_error}
    Consider a $\kh$-local Hamiltonian ($\kh\ge2$) $H=\sum_l^L\alpha_l G_l=\sum_{\gamma=1}^{\Gamma} H_\gamma$ whose $\Gamma\in\bigO(1)$ layers consist of terms with pairwise disjoint supports,
    a $\ko$-local input observable $O$,
    and an input state whose intermediate Trotter evolutions satisfy the entanglement condition of \cref{apd:thm:triangle}.
    Let $\alpha:=2\max_l\abs{\alpha_l}$, parameterize the truncation threshold as $w^*=\ko+m^*(\kh-1)$, write $c:=\ko/(\kh-1)$, and assume $m^*\ge1$ and $r\ge \max\qty(5,\,m^*,\,8(m^*+1)^2/\Gamma,\,8e^2(\ko+\kh-1)\alpha t)$ (the multi-jump conditions \cref{apd:eq:multijump_factor}).
    Then the total Pauli truncation error \cref{apd:eq:pauli_truncation_error_entangled} of all $r$ Trotter steps is at most
    \begin{equation}\label{apd:eq:total_truncation_error}
        \Delta \tilde{\mu}_{\le w^*} \le
        2\sqrt{2}\;\qty(\frac{t}{t_0})^{m^*+1}\, (m^*+2)\,\qty(e\,(m^*+2))^{c}
         \cdot \nfnorm{O},
        \qquad
        t< t_0:=\frac{1}{c_0\,\Gamma\,(\kh-1)\,\alpha}\in\Omega(1),
    \end{equation}
    where $c_0\le2$ is the constant of \cref{apd:thm:c0_bound} ($c_0\to1$ as $r\to\infty$), and $\Gamma$ is replaced by $\Upsilon\Gamma$ for the $p$th-order formula with even $p\ge2$ and $\Upsilon=2\cdot5^{p/2-1}$ (the statement as displayed is the first-order case, $\Upsilon=1$).
\end{theorem}
\begin{proof}
    Since $w^*=\ko+m^*(\kh-1)=w_{m^*+1}$, the truncated region $\qty{w\ge w^*+1}$ is exactly the region above rung $m^*+1$, and the mass truncated at the end of step $d$ is $\nfnorm{\tO^{(d)}_{\ge w^*+1}}=\cumnorm^{(d)}_{\ge m^*+1}$ evaluated at the end of step $d$.
    At the beginning of every step the previous truncation guarantees $\cumnorm_{\ge m^*+1}=0$, while every lower rung $m\le m^*$ obeys the majorant \cref{apd:eq:composition_majorant} along the whole truncated trajectory.
    During step $d$ (consisting of $\Gamma$ disjoint-support layers), the new mass above rung $m^*+1$ is fed by the layer inflow from all lower rungs; bounding each of the $\Gamma$ layers by the majorant at the end of the step --- legitimate because \cref{apd:eq:composition_majorant} is monotone in $T$, unlike the actual norms --- gives
    \begin{align}
        \nfnorm{\tO^{(d)}_{\ge w^*+1}}
        &\le \Gamma \sum_{j\ge1}\eps_j^{(m^*+1)}\; \overline{\cumnorm}_{\ge m^*+1-j}^{(d\Gamma)} ,
        \tag{\cref{apd:thm:layer_inflow}, \cref{apd:eq:composition_majorant}}
    \end{align}
    where $\overline{\cumnorm}_{\ge\nu}$ denotes the right-hand side of \cref{apd:eq:composition_majorant} for $\nu\ge1$ and the reservoir value $\nfnorm{O}$ for $\nu\le0$ (so the terms with $j\ge m^*+1$ sum to the entry factor $E_{m^*+1}\nfnorm{O}$).
    The leading (all-single-jump) sector of the right-hand side is $\Gamma\,w_{m^*+2}\sin(\dt)\binom{d\Gamma}{m^*}\sin^{m^*}(\dt)\prod_{j=2}^{m^*+1}w_j\cdot\nfnorm{O}$, up to the entry-factor inflation $E_1/\eps_1^{(1)}\le1+4e\beta$ of \cref{apd:eq:entry_bound}; the multi-jump and overshoot sectors are additive (\cref{apd:rmk:multijump}) and do not vanish when $\binom{d\Gamma}{m^*}$ does, so the sectors are compared only after the summation over the Trotter steps.
    Summing over the steps with
    $\sum_{d=1}^{r}\binom{d\Gamma}{K}\le \frac{\qty((r+1)\Gamma)^{K+1}}{\Gamma\,(K+1)!}$ for every sector (integer $\Gamma\ge1$)
    and $\sin(\dt)\le\alpha t/r$, all powers of $r$ collect into $((r+1)/r)^{m^*+1}$, and \cref{apd:thm:multijump_sum,apd:thm:c0_bound} give
    \begin{align}
        \sum_{d=1}^{r} \nfnorm{\tO_{\ge w^*+1}^{(d)}}
        &\le (1+4e\beta)\exp\qty[\frac94\frac{(m^*+1)^2}{r\Gamma}]\,\qty(\frac{r+1}{r}\,\Gamma\,\alpha t)^{m^*+1}\,\frac{\prod_{j=2}^{m^*+2}w_j}{(m^*+1)!}\cdot\nfnorm{O}
        \label{apd:eq:total_high_weight_norm} \\
        &\le \qty(c_0\,\Gamma\,(\kh-1)\,\alpha t)^{m^*+1}\,(m^*+2)\,\qty(e\,(m^*+2))^{c}\cdot\nfnorm{O} ,
        \nonumber
    \end{align}
    where we used
    $\prod_{j=2}^{m^*+2}w_j=(\kh-1)^{m^*+1}\prod_{j=1}^{m^*+1}(j+c)\le(\kh-1)^{m^*+1}\,(m^*+2)!\,\qty(e(m^*+2))^{c}$
    and absorbed $(r+1)/r$ together with the multi-jump factor \cref{apd:eq:multijump_factor} into $c_0\le2$ by \cref{apd:eq:c0} (both tend to $1$ as $r\to\infty$).
    The bound tends to zero as $m^*$ grows precisely under the short-time condition
    \begin{equation}\label{apd:eq:time_condition}
        t< t_0=\frac{1}{c_0\,\Gamma\,(\kh-1)\,\alpha} .
    \end{equation}
    Finally, \cref{apd:thm:triangle} converts the summed high-weight norms into the truncation error in expectation, $\Delta\tilde\mu_{\le w^*}\le\sqrt{2}\sum_{d}\nfnorm{\tO^{(d)}_{\ge w^*+1}}$, which gives \cref{apd:eq:total_truncation_error} (since $t_0=1/(c_0\Gamma(\kh-1)\alpha)$, the second line above is $(t/t_0)^{m^*+1}(m^*+2)\qty(e(m^*+2))^{c}\nfnorm{O}$ and the estimate in fact yields the prefactor $\sqrt2$; the stated $2\sqrt2$ is weaker by a factor $2$ and holds a fortiori).
\end{proof}
\begin{remark}[Sharpness of $t_0$]\label{apd:rmk:comparison}
    The bound carries no factor exponential in $\Gamma$ (the per-layer operator-norm accounting), no factor $e$ (the rung-weighted product $\prod_j w_j$ cancels the factorial in $\binom{T}{m}$), and the initial locality $\ko$ enters only through the polynomial prefactor $(m^*+2)^{1+c}$.
    The remaining wall at $t_0$ is a genuine convergence radius of the method rather than an artifact:
    the inflow chances per rung are proportional to the rung weight (saturated by worst-case Hamiltonians in which every encounter anti-commutes), and a ladder with rates $\propto w_m$ has amplitude $\asymp (t/t_0)^m$ at rung $m$ --- so enlarging $w^*$ does not extend the accessible time range, in close analogy with the cluster-expansion radius $t^*=1/(2e\mathfrak{d})$ of Wild and Alhambra \cite{wildClassicalSimulationShortTime2023}.
\end{remark}
\begin{remark}[Beyond the convergence radius]\label{apd:rmk:beyond_t0}
    The wall at $t_0$ concerns a ladder with a fixed base, i.e., a fixed truncation threshold applied to an input observable of fixed locality $\ko$.
    Two routes past it are visible from the analysis; we record them as possibilities and do not develop them here.
    First, re-basing the ladder: splitting the evolution into segments of duration below $t_0$ and restarting the truncation analysis from the truncated observable of the previous segment, with a threshold that grows from segment to segment, would extend the guarantee to every constant $t$ at a threshold growing exponentially in $t/t_0$ --- polynomial cost at constant $t$, but with a doubly exponential dependence on $t/t_0$ that mirrors the degree $e^{\pi t/t^*}$ of the cluster-expansion continuation \cite{wildClassicalSimulationShortTime2023}.
    Second, geometric locality: for geometrically local Hamiltonians, Lieb--Robinson bounds \cite{liebFiniteGroupVelocity1972,bravyiLiebRobinsonBoundsGeneration2006} confine the Heisenberg-evolved observable to a light cone of radius $vt+\bigO(\log(1/\eps))$, and running $\lpd$ without truncation inside the cone reproduces the standard guarantee \cite{osborneDynamics1DQuantum2006,haahQuantumAlgorithmSimulating2018} of polynomial cost at $t=\bigO(\log n)$ in one dimension; this is a property of the cone rather than of weight truncation, whose advantage inside the cone is empirical.
\end{remark}

Since we have expressed the upper bound of the Pauli truncation error in terms of the truncation weight threshold $w^*$,
we can now determine the required $w^*$ to achieve a desired Pauli truncation error tolerance.
\begin{corollary}[Truncation weight threshold with an entangled input state]\label{apd:thm:truncation_threshold_entangled}
    Under the assumptions of \cref{apd:thm:one_step_truncation_error}, with evolution time $t< t_0$,
    the truncation weight threshold
    \begin{equation}\label{eq:truncation_weight_bound}
        w^* = \ko+(\kh-1)\,m^*,
        \qquad
        m^* \in \bigO\qty(\frac{\log(1/\eps)+\qty(1+\frac{\ko}{\kh-1})\log\log(1/\eps)}{\log(t_0/t)}) ,
    \end{equation}
    suffices to achieve the Pauli truncation error tolerance $\eps\,\nfnorm{O}$.
    In particular, $w^*=\ko+\bigO\qty(\frac{\log(1/\eps)}{\log(t_0/t)})$ up to a doubly logarithmic correction; it is a constant for constant $t<t_0$ and constant $\eps$, and $\bigO(\log n)$ for $\eps=1/\poly(n)$.
\end{corollary}
\begin{proof}
    Setting the right-hand side of \cref{apd:eq:total_truncation_error} below $\eps\nfnorm{O}$ requires
    $(m^*+1)\log(t_0/t)\ge\log\qty(2\sqrt2\,(m^*+2)^{1+c}e^{c}/\eps)$ with $c=\ko/(\kh-1)$,
    which holds for $m^*$ as stated.
\end{proof}
Recall \cref{apd:thm:ob_norm_bound_expectation_error}, the truncation threshold also applies to the random input states scenario,
so we give the probabilistic bound that means most input states will satisfy the truncation threshold requirement.
This is the noiseless short-time analogue of the average-case guarantee of Ref.~\cite{schusterPolynomialtimeClassicalAlgorithm2025}, with the damping $\sin(\dt)$ of the norm flow in place of the noise-induced damping $e^{-\gamma}$ per layer; unlike noise, it acts only for $t<t_0$.
\begin{corollary}[Truncation weight threshold with random input states]\label{apd:thm:probability_bound}
    Under the assumptions of \cref{apd:thm:multijump_sum,apd:thm:c0_bound}, with evolution time $t<t_0$ of \cref{apd:eq:time_condition},
    if the input state is drawn from the 2-design states (cf. \cref{apd:lem:2design_ob_bound}),
    then the truncation weight threshold
    \begin{equation}
        w^* = \ko+\bigO\qty((\kh-1)\,\frac{\log(1/(\eps^2\delta))}{\log(t_0/t)})
    \end{equation}
    suffices to achieve the Pauli truncation error tolerance $\eps\,\nfnorm{O}$ with success probability $1-\delta$.
\end{corollary}
\begin{proof}
    For a real random variable $X$ and $a>0$, the Markov's inequality states that
        $\Pr[\abs{X}\ge a] = \Pr[X^2\ge a^2] \le \frac{\bbE[X^2]}{a^2}.$
    Write $D:=\tO^{(r)} - \tO^{(r)}_{\le w^*}$ and apply $\nfnorm{\cdot}$ to the telescoping identity in the proof of \cref{apd:thm:triangle}: by the triangle inequality and the unitary invariance of the Schatten $2$-norm,
    $\nfnorm{D}\le\sum_{d=1}^{r}\nfnorm{\pf^{\dagger(r-d)}\,\tO^{(d)}_{\ge w^*+1}\,\pf^{\,r-d}}=\sum_{d=1}^{r}\nfnorm{\tO^{(d)}_{\ge w^*+1}}$,
    which \cref{apd:thm:multijump_sum,apd:thm:c0_bound} bound as in \cref{apd:eq:total_high_weight_norm}, without any assumption on the input state.
    Taking $X=\bra{\psi}D\ket{\psi}$ with $\ket{\psi}\sim\statee_2$ and $a=\eps\nfnorm{O}$, and using $\bbE_{\ket{\psi}\sim\statee_2}\qty[\bra{\psi}D\ket{\psi}^2]\le\nfnorm{D}^2$ (\cref{apd:lem:2design_ob_bound}), the failure probability is at most $\delta$ once $\nfnorm{D}\le\eps\sqrt{\delta}\,\nfnorm{O}$, i.e., by replacing $\eps$ with $\eps\sqrt{\delta}$ in \cref{eq:truncation_weight_bound}; since $\log(1/(\eps\sqrt{\delta}))=\log(1/(\eps^2\delta))/2$, the stated $\bigO(\cdot)$ is unchanged.

    The two routes from the norm bound to an expectation-value error differ in where the state enters.
    Here the 2-design average $\bbE_{\ket{\psi}\sim\statee_2}\qty[\bra{\psi}D\ket{\psi}^2]\le\nfnorm{D}^2$ holds for the total difference $D$ regardless of its weight, so the residual evolution stays on the operator and no condition on the intermediate states is needed (a 2-design ensemble remains a 2-design under any fixed unitary).
    For a single entangled input, by contrast, the bound $\abs{\bra{\phi}A\ket{\phi}}\le\sqrt2\,\nfnorm{A}$ of \cref{apd:cor:entangle_ob_bound} requires a bounded-weight $A$, which forces the per-step, Schr\"odinger-picture argument of \cref{apd:thm:triangle} and its entanglement condition on every intermediate state.
\end{proof}

\subsection{Runtime and memory analysis}\label{apd:sec:proof_runtime}

Since every Trotter step is the same and all Paulis in the observable are low-weight (less than $w^*$) at the beginning of each Trotter step,
we only need to analyze the runtime and memory requirement of one Trotter step.
If $w^*$ is independent of system size $n$,
the number of Pauli operators generated in one Trotter step will be polynomial in $n$
such that the runtime and memory complexity of one Trotter step are polynomial in $n$.
    Here, we give the detailed analysis.
\begin{lemma}[Proliferation of Pauli operators in one truncated Trotter step]\label{apd:thm:lightcone}
    Consider a $w^*$-local observable $O=\sum_j^J  \beta_j P_j$  with $J\in \poly(n)$
    and a $\kh$-local Hamiltonian $H=\sum_l^L \alpha_l G_l$ with $\kh\ge2$ and $\kh\in\bigO(1)$ whose terms are arranged into $\Gamma\in\bigO(1)$ layers with pairwise \emph{disjoint supports} within each layer, so that one Trotter step consists of $\Gamma$ (for the $p$th-order formula, $\Upsilon\Gamma$) such layers.
    If $w^*\in\bigO(1)$, then the number of Pauli operators in the observable at the end of one Trotter step, before truncation, is $\bigO(n^{w^*})$.
\end{lemma}
\begin{proof}
    After the truncation ending the previous step, the observable contains at most $\sum_{w=1}^{w^*}3^w\binom{n}{w}\le \qty(3en/w^*)^{w^*}\in\bigO(n^{w^*})$ Pauli operators.
    Within one layer, a Pauli of weight $w$ meets at most $w$ rotations (disjoint supports), and this set of rotations is unchanged along the branching, since the generators of a layer commute with the one that has just fired; each met rotation at most doubles the number of branches (cf.\ the sparse transition rule \cref{apd:eq:pauli_rotation_branch}) and raises the weight by at most $\kh-1$.
    Hence after layer $\gamma$ the weight is at most $w^*\kh^{\gamma}$, and the branch multiplicity per initial Pauli is bounded by $\prod_{\gamma=1}^{\Upsilon\Gamma}2^{w^*\kh^{\gamma-1}}\le 2^{w^*\kh^{\Upsilon\Gamma}}\in\bigO(1)$ for $w^*,\kh,\Gamma\in\bigO(1)$.
    The total number of Paulis at the end of the step is therefore $\bigO(n^{w^*})\cdot\bigO(1)=\bigO(n^{w^*})$.
    (For a Pauli supported on a single connected interval in 1D brickwork layers the weight growth is additive, at most $w^*+2\Upsilon\Gamma$; for a general weight-$w^*$ Pauli with disconnected support each isolated site can be extended by its own gate, so only the multiplicative bound holds --- and it suffices here.)
    The disjoint-support assumption is necessary: for the commuting star layer $\qty{X_1X_j}_{j=2}^{L+1}$ (a single commuting group, $\Gamma=1$), the single Pauli $Z_1$ branches into $2^{L}$ operators of weight up to $L+1$ within one layer.
\end{proof}

With the maximum number of Paulis in one Trotter step,
we have the following upper bound of the total runtime (the number of inner products to calculate) and the corresponding memory requirement.
\begin{theorem}[Runtime of $\lpd$ Algorithm]\label{apd:thm:runtime}
    Given a $\kh$-local Hamiltonian $H=\sum_{\gamma}^{\Gamma} H_{\gamma}$ with $\kh\in\bigO(1)$ and $\Gamma\in\bigO(1)$ disjoint-support layers,
    any $\ko$-local observable $O=\sum_j \beta_j P_j$ with $\ko\in\bigO(1)$, normalized so that $\nfnorm{O}=1$,
    evolution time $t< t_0$ (cf.\ \cref{apd:eq:time_condition}) under the assumptions of \cref{apd:thm:one_step_truncation_error} (in particular $\kh\ge2$, $m^*\ge1$ and the lower bounds on $r$ stated there), simulation precision $\eps$,
    and an input state whose intermediate Trotter evolutions satisfy the entanglement condition of \cref{apd:thm:triangle}.
    Assume that the Pauli coefficients of $H$, $O$, and $\rho$ can each be computed in $\bigO(1)$ time, and that the evolved observable is stored in a \emph{support-indexed} data structure (a map from qubit index to the Pauli operators supported on it), so that each rotation touches only the Paulis on its support.
    Then there is a classical algorithm approximating the expectation value $\mu=\Tr(\rho\, U^\dagger O U)$ within error tolerance $2\eps$, running with
    \begin{equation}
        \textup{time: }\bigO\qty((r+1)\cdot n^{w^*} ), \quad
        \textup{memory: }\bigO\qty(n^{w^*} ),
    \end{equation}
    where $r\in\poly(n,1/\eps)$ is the number of Trotter steps, chosen such that the Trotter error in expectation is at most $\eps$ and the lower bounds on $r$ required by \cref{apd:thm:one_step_truncation_error} are met (increasing $r$ only decreases $c_0$, hence never weakens the truncation bound)
    and the truncation threshold is $w^*= \ko+\bigO\qty(\frac{\log(1/\eps)}{\log(t_0/t)})$ by \cref{apd:thm:truncation_threshold_entangled}.
\end{theorem}
\begin{proof}
    The runtime is dominated by the number of transition amplitudes evaluated; each amplitude, i.e., the inner product of two Paulis, costs $\bigO(1)$ time.
    By \cref{apd:thm:lightcone}, throughout the evolution the observable contains $\bigO(n^{w^*})$ Pauli operators.
    With the support-indexed data structure, one Trotter step costs $\bigO(\Upsilon\Gamma\, w^*\kh^{\Upsilon\Gamma})=\bigO(1)$ operations per stored Pauli, i.e., $\bigO(n^{w^*})$ per step.
    (Without it, a naive sweep over all $\Upsilon L$ rotations against all stored Paulis costs $\bigO(L\cdot n^{w^*})=\bigO(n^{w^*+1})$ per step; the data-structure assumption is what removes the factor $L$.)
    The number of Trotter steps follows from \cref{apd:thm:trotter_error_entangled}: if the intermediate states are sufficiently entangled with respect to the nested-commutator error operators of one Trotter step, $r=\bigO\qty((\fLambda\,\opnorm{O}\,\eps^{-1})^{1/p}\, t^{1+1/p})$ steps suffice for a Trotter error $\eps$;
    for an arbitrary input state, \cref{apd:lem:worst_trotter_error} gives the same count with $\Lambda$ in place of $\fLambda$.
    Both counts are polynomial in $n$, because $\opnorm{O}\le\sum_j\abs{\beta_j}\le\sqrt{N_O}\,\nfnorm{O}$ for an observable with $N_O=\bigO(n^{\ko})$ Pauli terms.
    Note that $\fLambda$ grows with system size for extensive Hamiltonians --- e.g., $\fLambda=\Theta(\sqrt n)$ already for 1D nearest-neighbor models \cite{zhaoHamiltonianSimulationRandom2021} --- so $r\in\poly(n)$ in general; for local observables, the observable-driven bounds of Ref.~\cite{yuObservabledrivenSpeedupsQuantum2025} make $r$ independent of $n$ at constant $t$ and $\eps$.
    The truncation threshold $w^*$ is constant for constant $t<t_0$ and $\eps$ by \cref{apd:thm:truncation_threshold_entangled}, so the total evolution cost is $\bigO(r\cdot n^{w^*})$.

    The last step of $\lpd$ evaluates the expectation $\Tr(\tO^{(r)}_{\le w^*}\rho)$.
    Since each Pauli coefficient of $\rho$ is computable in $\bigO(1)$ time, this costs one inner product per stored Pauli by orthogonality, i.e., $\bigO(n^{w^*})$ in total --- no more than one Trotter step.
    A concrete admissible family of input states are exact $k$-uniform \emph{stabilizer} states (e.g., from MDS-code or graph-state constructions): they satisfy the entanglement condition exactly for $k\ge 2w^*\kh^{\Upsilon\Gamma}$, and their Pauli coefficients ($\pm2^{-n/2}$ on the stabilizer group and $0$ elsewhere) are computable in $\bigO(1)$ per Pauli by a symplectic membership check.
    (By contrast, a product state such as $\ket{\vb{0}}$ has efficiently computable coefficients but does \emph{not} satisfy the entanglement condition, so the error guarantee does not apply to it --- consistent with the numerics, where the product-state truncation error is the largest.)
    The total error budget combines the Trotter error $\eps$ (\cref{apd:thm:trotter_error_entangled}) with the truncation error $\eps\,\nfnorm{O}=\eps$ (\cref{apd:thm:truncation_threshold_entangled}; recall the normalization $\nfnorm{O}=1$), giving $2\eps$.
    Correspondingly, the memory requirement is the maximum number of stored Pauli operators, $\bigO(n^{w^*})$.
\end{proof}

\subsection{Efficiency and limitation of hybrid classical protocol}

This subsection formalizes the hybrid classical protocol
that combines $\mps$-based forward state evolution
with $\lpd$ backward observable evolution.
The $\mps$ methods are efficient
when the Schmidt rank across every cut stays bounded.
The average-case truncation error of $\lpd$ instead requires sufficiently entangled input states (\cref{apd:cor:entangle_ob_bound}),
i.e., every constant-size subsystem close to maximally mixed (\cref{apd:rmk:entanglement_condition}).
Since $\lpd$ enforces the Paulis in the evolved observable to be low-weight (at most constant $w^*$),
its condition concerns subsystems of constant size only
(cf.\ mixed Schr\"odinger--Heisenberg tensor-network simulation \cite{begusicFastConvergedClassical2024} and the meet-in-the-middle option of Pauli propagation \cite{rudolphPauliPropagationComputational2026}).
We first recall the cost of $\mps$ simulation,
then show that the two hypotheses match in scale.

\begin{lemma}[$\mps$--$\texttt{TEBD}$ cost; restatement of {\cite{vidalEfficientClassicalSimulation2003,vidalEfficientSimulationOnedimensional2004}}]\label{lem:mps_entanglement}
    If a quantum state $\ket{\psi}$ on $n$ qubits has Schmidt rank
    at most $\chi$ across every contiguous bipartition,
    then $\ket{\psi}$ admits an exact $\mps$ representation
    with bond dimension $\chi$.
    The bipartite entanglement entropy satisfies
    $S_A \le \log \chi$ for every such bipartition.
    Time-evolving an $\mps$ of bond dimension $\chi$ via TEBD
    for $r$ Trotter steps costs $\bigO(n\chi^3 r)$ arithmetic operations,
    where the bond dimension is truncated back to $\chi$ after each gate.
\end{lemma}

When the input state is a product state,
entanglement grows during time evolution,
so the $\mps$ bond dimension required for an exact representation
increases with time.
Meanwhile, \cref{apd:cor:entangle_ob_bound} guarantees that
once the constant-size subsystems are sufficiently entangled,
$\lpd$ provides an accurate backward evolution of the observable.
Combining both methods yields the following guarantee.

\begin{proposition}[Efficiency of hybrid $\mps$--$\lpd$ simulation]\label{prop:hybrid_efficiency}
    Given a one-dimensional $\kh$-local Hamiltonian $H = \sum_\gamma^\Gamma H_\gamma$ with $\Gamma\in\bigO(1)$ disjoint-support layers,
    a $\ko$-local observable $O$,
    a product input state $\rho(0) = \op{\psi(0)}$,
    total evolution time $t = t_F + t_B$ with $t_B < t_0$,
    and $r = r_F + r_B$ Trotter steps chosen as in \cref{apd:thm:runtime},
    suppose that
    (i) throughout the first $r_F$ Trotter steps, after every gate, the evolved state has Schmidt rank at most $\chi$ across every contiguous bipartition, and
    (ii) the states $\pf^{d}\ket{\psi(0)}$ for $d = r_F,\dots,r$ satisfy the entanglement condition of \cref{apd:thm:triangle}, as required in \cref{apd:thm:runtime}.
    Then the hybrid $\mps$--$\lpd$ protocol approximates
    $\Tr(\rho(0)\, e^{\ii Ht}\, O\, e^{-\ii Ht})$
    within additive error $\bigO(\eps\nfnorm{O})$ in time
    \begin{equation}\label{eq:hybrid_runtime}
        \bigO\!\left(n\chi^3 r_F \;+\; n^{w^*}\qty(r_B + n\chi^3)\right)
    \end{equation}
    and memory $\bigO(n\chi^2 + n^{w^*})$,
    where $w^*$ is the truncation threshold of \cref{apd:thm:truncation_threshold_entangled}.
\end{proposition}
\begin{proof}
    \emph{MPS simulation in Schrödinger picture.}
    Under (i), by \cref{lem:mps_entanglement} the TEBD evolution with bond dimension $\chi$ represents $\ket{\psi_F}:=\pf^{r_F}\ket{\psi(0)}$ exactly, at cost $\bigO(n \chi^3 r_F)$ and memory $\bigO(n\chi^2)$.

    \emph{LPD simulation in Heisenberg picture.}
    The hybrid estimate is $\bra{\psi_F}\tO^{(r_B)}_{\le w^*}\ket{\psi_F}$, the $\lpd$ estimate with input state $\ket{\psi_F}$ and $r_B$ backward steps.
    Since $\bra{\psi(0)}\pf^{\dagger r} O \pf^{r}\ket{\psi(0)} = \bra{\psi_F}\pf^{\dagger r_B} O \pf^{r_B}\ket{\psi_F}$,
    it targets the same $r$-step Trotter circuit as $\lpd$ alone,
    so the truncation error bound (\cref{apd:thm:triangle,apd:thm:truncation_threshold_entangled}) and the Trotter error bound of \cref{apd:thm:runtime} apply verbatim under hypothesis (ii),
    with the intermediate states $\pf^{d}\ket{\psi_F}=\pf^{r_F+d}\ket{\psi(0)}$.
    The backward evolution costs $\bigO(r_B\, n^{w^*})$ time and $\bigO(n^{w^*})$ memory by \cref{apd:thm:lightcone}.

    \emph{Evaluation.}
    Each of the $\bigO(n^{w^*})$ retained Pauli strings is a tensor product, so its expectation in the $\mps$ is one transfer-matrix sweep costing $\bigO(n\chi^3)$.
    Summing the three parts gives \cref{eq:hybrid_runtime}.
\end{proof}
The two hypotheses match in scale.
Condition (ii) requires every $k$-qubit reduced state with $k = 2w^*\kh^{\Upsilon\Gamma}$ to be $\Delta$-close to maximally mixed (\cref{apd:rmk:entanglement_condition}),
so a contiguous $k$-block $B$ has entanglement entropy $S_B \approx k$ (exactly $k$ for $\Delta=0$).
For an $\mps$ of bond dimension $\chi$, the reduced state of a contiguous block is bounded by its two cuts, so $S_B \le 2\log_2\chi$ by \cref{lem:mps_entanglement}.
Hence a state in the $\lpd$ regime needs $\chi \ge 2^{k/2} = 2^{\Omega(w^*)}$ for an exact forward representation,
and the switching time is where the forward budget reaches the scale $\chi = 2^{\Theta(w^*)}$ at which the backward guarantee begins.
At this scale the runtime \cref{eq:hybrid_runtime} is polynomial in $n$ for constant $w^*$.

Both methods are individually limited to constant times.
In certain nonintegrable systems,
the entanglement entropy grows ballistically as
$S_A(t) \sim v_E\, t$ with entanglement velocity $v_E$~\cite{kimBallisticSpreadingEntanglement2013},
so the $\mps$ bond dimension $\chi \ge 2^{S_A}$ grows exponentially with $t$
and the forward phase is exact only for $t_F \lesssim w^*$.
Meanwhile, $\lpd$ requires $t_B \in \Theta(1)$,
the $\lpd$ time limit from \cref{apd:thm:runtime}.
The total simulation time is therefore limited to
$t \in \Theta(1)$.
When the target precision is $\eps \sim 1/\mathrm{poly}(n)$,
\cref{apd:thm:truncation_threshold_entangled} gives $w^* = \bigO(\log n)$,
hence $\chi = 2^{\bigO(\log n)} = \mathrm{poly}(n)$
and the overall runtime is quasi-polynomial in $n$.

\section{Additional numerical results}\label{apd:sec:additional_numerics}

\subsection{Pushing the frontier of Quantinuum's digital-magnetism experiment}\label{apd:sec:quantinuum}

We apply the hybrid $\mps$--$\lpd$ protocol to a recognized frontier instance: the digitized two-dimensional transverse-field Ising dynamics realized on Quantinuum's H2 trapped-ion processor~\cite{haghshenasDigitalQuantumMagnetism2026}, a $7\times8$ (56-qubit) experiment whose accompanying classical-baseline study---tensor networks, neural-network states, and operator-truncation (sparse Pauli) methods---found no classical method both efficient and trustworthy for its intermediate-temperature quench at the experimental scale.
Our implementation combines the forward $\mps$ evolution of the product initial state with the backward $\lpd$ evolution of the observable, sparse-Pauli coefficient truncation~\cite{begusicFastConvergedClassical2024,begusicRealTimeOperatorEvolution2025}, and merging of Pauli operators related by the discrete lattice symmetries~\cite{tengLeveragingSymmetryMerging2025}.
The evaluation of the retained operator on the $\mps$, the dominant cost, shares partial contractions between strings with a common prefix and, where the operator exceeds the budget for exact contraction, uses importance sampling with the statistical error reported alongside the values (Methods).
Since independent samples need no communication beyond the final average, distributing them over several GPUs, each holding a copy of the $\mps$, is a straightforward extension.
It would reduce the wall-clock time of this dominant contraction step essentially in proportion to the number of GPUs.

\Cref{apd:fig:quantinuum_6x6} shows the current status on the $6\times6$ (36-qubit) version of the experimental quench protocol ($J=-1$, $h=2$, $\delta t=0.25$, intermediate Trotterization angle $\Delta\theta=2\pi/9$), for the observable $\ev{Z^2_{\mathrm{tot}}}$.
Following Refs.~\cite{haghshenasDigitalQuantumMagnetism2026,mandraHeuristicMatrixProduct2026}, $s$ counts the Trotter steps of the experiment's circuit, whose step size $\delta t$ is fixed;
$s_F$ and $s_B$ are the numbers of forward and backward steps, denoted $r_F$ and $r_B$ in \cref{prop:hybrid_efficiency}.
With forward $\mps$ for $s_F=10$ Trotter steps at bond dimension $\chi=2048$ and backward $\lpd$ (truncation weight $w^*=9$ with sparse-coefficient thresholds $\delta=10^{-5},10^{-6}$) covering the remaining steps, the best of our norm-corrected hybrid estimators roughly halves the root-mean-square deviation from the numerically exact reference over the backward window ($s=11$--$20$), relative to the unrescaled pure-$\mps$ simulation at the same bond dimension. The raw hybrid estimator, depleted by the truncation, is less accurate than the pure $\mps$; the gain comes from the correction for the Pauli 2-norm lost to truncation, and the best estimator is selected against the exact reference, which at the experimental size no longer exists.
The backward $\lpd$ stage thus extends the window over which the estimate stays close to the exact reference at fixed tensor-network resources, illustrating the entanglement--magic complementarity of the hybrid protocol on a two-dimensional system. Since the baseline study, the fidelity-rescaled $\mps$ of Mandr\`a et al.~\cite{mandraHeuristicMatrixProduct2026} has become the most accurate classical estimate reported on this instance; the gain quoted above is measured against the \emph{unrescaled} $\mps$, and the $6\times6$ estimators used here do not include that rescaling.

\begin{figure}[!t]
    \centering
    \includegraphics[width=0.9\linewidth]{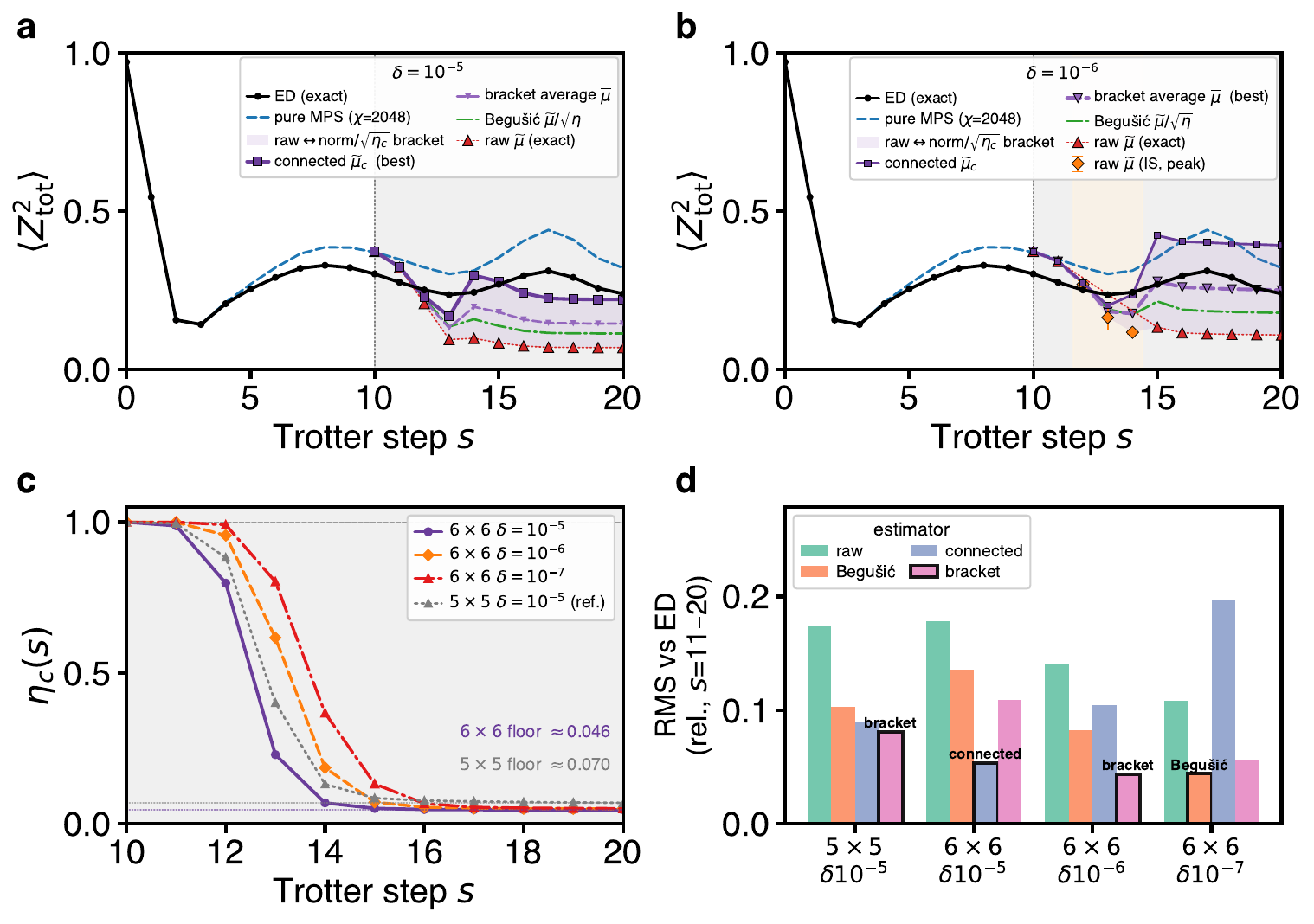}
    \caption{
    Hybrid $\mps$--$\lpd$ simulation of the Quantinuum digital-magnetism protocol~\cite{haghshenasDigitalQuantumMagnetism2026} on a $6\times6$ periodic lattice ($N=36$).
    (a,b) Numerically exact reference (black), pure $\mps$ at $\chi=2048$ (blue dashed), and hybrid estimators (forward $\mps$ for $s_F=10$ steps, backward $\lpd$ with $w^*=9$ and sparse-coefficient threshold $\delta=10^{-5}$ (a) and $10^{-6}$ (b)) for $\ev{Z^2_{\mathrm{tot}}}$: over the backward window $s=11$--$20$ the best norm-corrected estimator stays closer to the exact dynamics than pure $\mps$ at the same bond dimension, which visibly departs.
    Panel (d) gives the RMS deviations.
    The points at $s=12$--$14$ in (b) are importance-sampling estimates of the raw estimator, shown with their statistical error bars, at the steps where the retained operator exceeded the budget for exact contraction.
    (c) Connected norm retention $\eta_c(s)$ of the backward-evolved observable for $\delta=10^{-5},10^{-6},10^{-7}$, with a $5\times5$ reference.
    (d) Estimator-quality comparison across lattice sizes and truncation thresholds.
    }
    \label{apd:fig:quantinuum_6x6}
\end{figure}

At the experimental $7\times8$ scale, the same hybrid simulation pushes two Trotter steps beyond the end of the experiment: from the $s=20$ forward $\mps$ state at $\chi=2048$, backward $\lpd$ at $w^*=6$ gives $s=21$ and $s=22$ (\cref{apd:fig:quantinuum_7x8}).
We do not claim that the forward $\mps$ could not be continued there.
The claim is that backward $\lpd$ reaches those steps from the state we already had.
At $5\times5$, where exact diagonalization can be carried further, backward $\lpd$ continues the same $\chi=2048$ state for five Trotter steps (\cref{apd:fig:quantinuum_7x8}, left).
With no exact reference to select an estimator at $7\times8$, the reported hybrid value is corrected on both sides: the operator side by our norm correction, and the state side by exactly the fidelity rescaling of Ref.~\cite{mandraHeuristicMatrixProduct2026}, with the exponent refitted because our forward states are contracted per qubit rather than row-fused.

\begin{figure}[!t]
    \centering
    \includegraphics[width=0.9\linewidth]{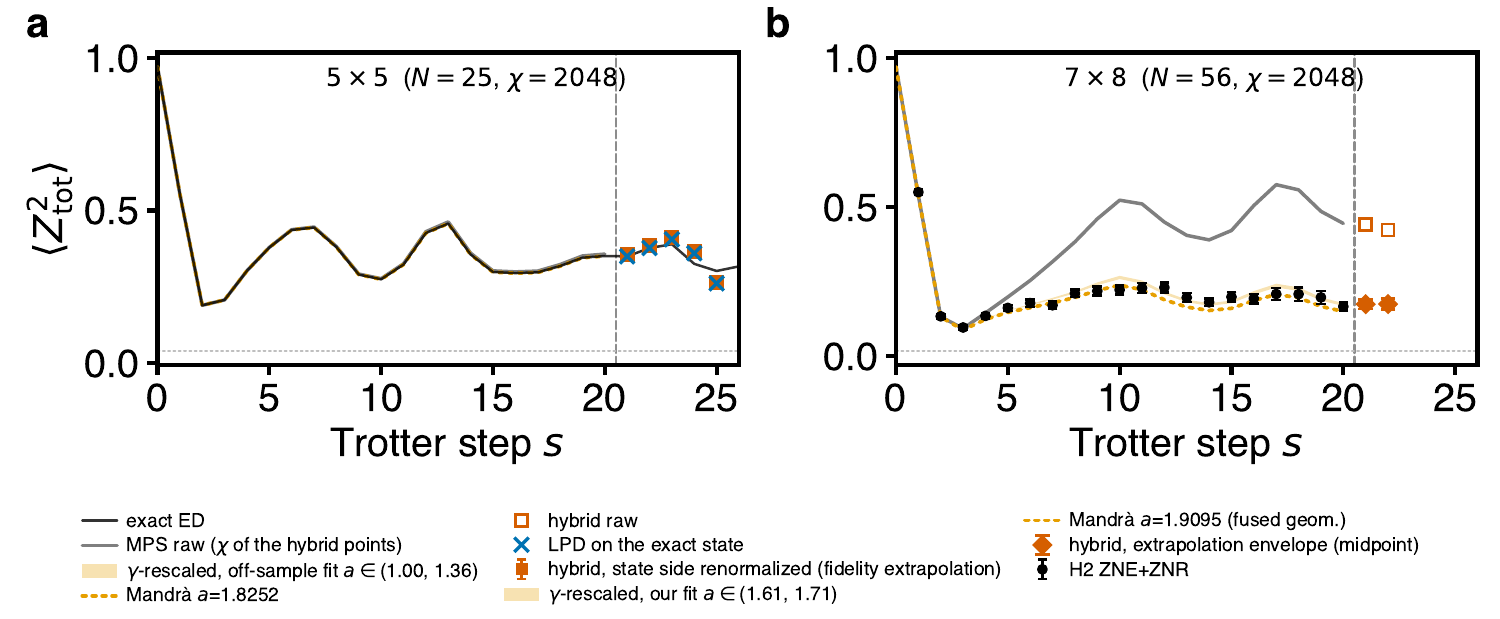}
    \caption{
    Hybrid $\lpd$-$\mps$ pushes beyond the digitized 2D Ising quench of the H2 digital-magnetism experiment~\cite{haghshenasDigitalQuantumMagnetism2026} (intermediate-temperature quench, observable $\ev{Z^2_{\mathrm{tot}}}$).
    Two lattices; the dashed line at $s=20.5$ marks the split between the forward $\mps$ ladder, which we ran to $s=20$ to match the length of the experiment's protocol, and the backward $\lpd$ steps.
    Grey: the raw $\mps$ at $\chi=2048$, the bond dimension carried by the hybrid points of both panels.
    Band: the same $\mps$ estimate rescaled by a power of its accumulated truncation fidelity in the manner of Ref.~\cite{mandraHeuristicMatrixProduct2026}, the coefficient of that exponent refit against exact references on our own calibration grids (right: $16$--$30$ qubits; left: $16$--$24$ qubits, extrapolated to $25$, so that the $5\times5$ reference itself is not used); the band spans the range over which the fitted coefficient moves with the calibration grid, and is a fit-window spread rather than an uncertainty.
    Dotted: the exponent of that reference, drawn for comparison (left: its published value; right: the value of its arXiv v1, fitted on row-fused forward states); at $7\times8$ our own fit restricted to its largest grids comes close to it, so the gap there is a window effect.
    Open squares: the raw hybrid estimator at $s=20+s_B$ (right: with importance-sampling error bars; left: evaluated exactly on the full statevector).
    Filled squares (left): the raw hybrid with the same fidelity rescaling of its $\mps$ state.
    Crosses (left): the same backward observables evaluated on the exact state, which isolates the truncation of the observable from that of the state.
    Filled diamonds (right): the midpoint of the envelope of the doubly-extrapolated estimator---the state-side exponent at both ends of the band together with the value of Ref.~\cite{mandraHeuristicMatrixProduct2026}, the operator-side exponent of the small-lattice law shifted by its one-signed transfer band, and the sampling error.
    It is not a single estimate and not an error bar: the residual bond-dimension dependence of the forward $\mps$, which is not converged at $\chi=2048$, falls outside it.
    Black: exact diagonalization, which at $5\times5$ can be carried past the split, and, at $7\times8$, the error-mitigated H2 data with their published $1\sigma$ bars; there the band overlaps those bars at most steps and runs above them at several.
    At $7\times8$ ($N=56$, $\chi=2048$, backward $\lpd$ at $w^*=6$ and coefficient threshold $\delta=10^{-7}$) the hybrid returns $s=21$ and $s=22$, beyond the last step of the protocol; the corrected envelope lies well below the raw estimator, and both points fall short of the sampling target.
    There the backward depth is cut where the retained operator norm collapses: a point is kept only while $\eta_c(s_B)\ge0.85$, which at $w^*=6$ means $s_B\le2$.
    At $5\times5$ the exact reference checks every step, so no cut is applied: backward $\lpd$ runs five steps with a budget of $10^9$ retained Pauli strings per step (largest coefficients kept, $w^*=16$) instead of a coefficient threshold.
    }
    \label{apd:fig:quantinuum_7x8}
\end{figure}

\subsection{Entangled states: randomness versus entanglement}\label{apd:sec:entangle_vs_random}

The average error over the random-Clifford ensemble in Fig.~3 of the main text leaves open whether the observed suppression stems from the \emph{entanglement} of the sampled states or merely from their \emph{randomness} (self-averaging).
To separate the two hypotheses, we evaluate the same truncation-error operator of Fig.~3d of the main text ($w^*=4$) on state families that disentangle randomness from entanglement, keeping all other parameters identical ($n=10$ mixed-field Ising chain, $t=5$, second-order Trotter with $r=50$ steps, observable $Z_1$).
The families are:
uniformly random computational-basis states and Haar-random product states (both exact state \emph{1-designs}: maximally random, zero entanglement);
the fixed N\'eel product state of the main text;
GHZ and the 1D cluster (graph) state (non-random, with half-chain entropy $S_{n/2}=1$~bit);
random Clifford states (an exact state 2-design, highly entangled);
and Haar-random states ($S_{n/2}\approx 4.28$~bits on average).

The outcome (\cref{apd:fig:entangled_families}) is unambiguous.
The two maximally random but unentangled 1-design families show truncation errors several times \emph{larger} than the highly entangled ensembles (random Clifford and Haar), both time-averaged and at the final time $t=5$, where they lie above \emph{every} entangled family considered: randomness alone, and the self-averaging it could induce, does not reduce the error.
The mechanism is sharp: random bitstrings leave diagonal ($Z$-type) strings entirely unsuppressed, and even a Haar-random product state suppresses a weight-$w$ Pauli string only as $3^{-w/2}$, whereas for any state 2-design $\mathbb{E}\ev{\Delta}^2=\nfnorm{\Delta}^2/(2^n+1)$, suppressing every string to the $2^{-n/2}$ scale; our Clifford and Haar data coincide with each other and with this analytic prediction, confirming that only the second moment of the ensemble enters.
Conversely, the non-random entangled states (GHZ and cluster) are single stabilizer states, so their instantaneous error oscillates coherently over a wide envelope rather than self-averaging ($\ev{P}\in\{0,\pm1\}$ for every Pauli string $P$): time-averaged, GHZ lands \emph{between} the random-product and 2-design levels, while the cluster state exceeds even the random-product level over most of $1.5\lesssim t\lesssim4$; at the final time both suppress the error as strongly as the highly entangled ensembles.
The cleanest separation is therefore \cref{apd:fig:entangled_families}(b): the evolved bitstrings form an orthonormal basis at every $t_F$ --- an exact state 1-design, never a 2-design --- so the randomness is held fixed while only the entanglement grows, and the ensemble-mean error falls to the Haar level.

\begin{figure}[!t]
    \centering
    \includegraphics[width=0.9\linewidth]{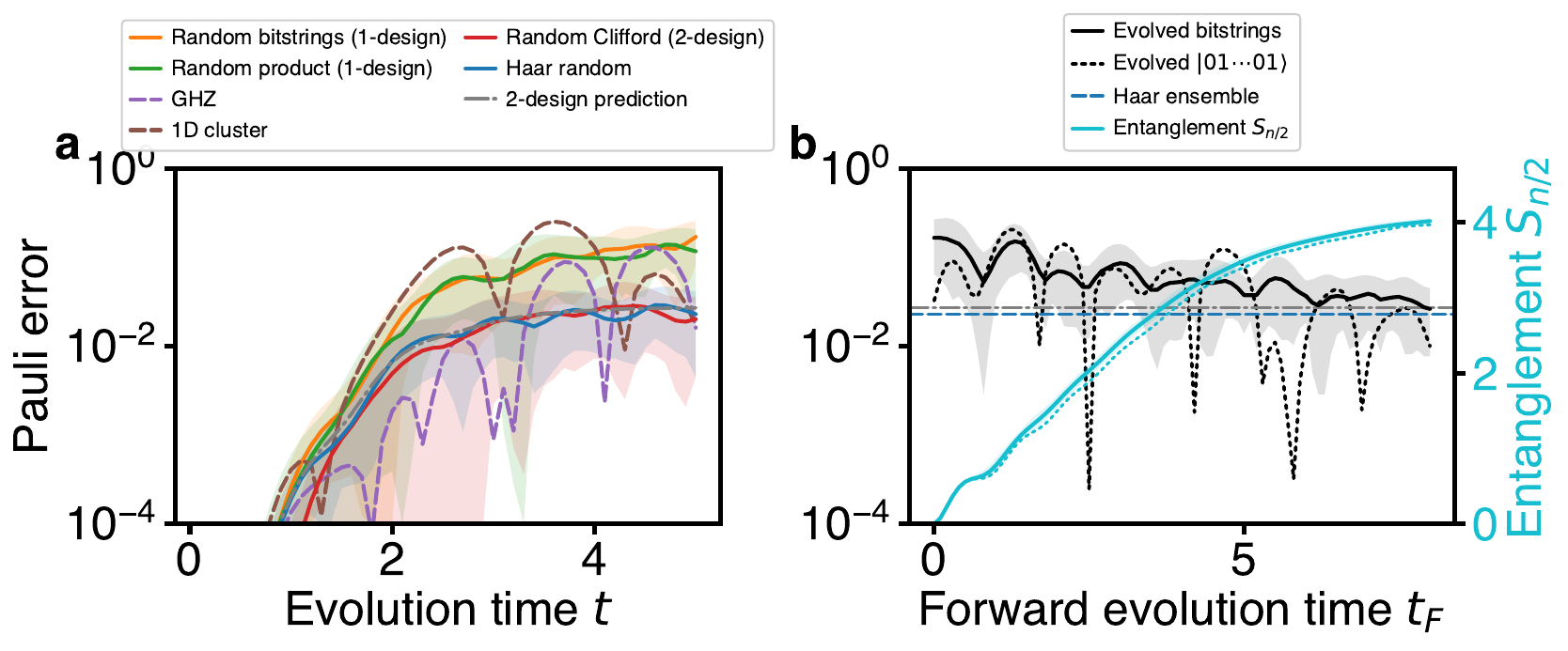}
    \caption{
    Pauli truncation error of the fixed backward-evolved observable of Fig.~3d of the main text ($n=10$, $t=5$, $w^*=4$) on entangled states.
    (a) Time-resolved error for six state families separating randomness from entanglement (ensembles of $30$ states: mean $\pm$ one standard deviation; gray dash-dotted: the analytic 2-design prediction; the N\'eel product state of the main text is omitted from the plot, its trajectory being Fig.~4 of the main text).
    The maximally random but unentangled 1-designs (random bitstrings, random product states) are the \emph{worst} families at the final time; the random-Clifford (2-design) and Haar ensembles coincide with each other and with the prediction; the non-random entangled GHZ and 1D-cluster states ($S_{n/2}=1$~bit) oscillate coherently over a wide envelope rather than self-averaging.
    (b) The same error operator evaluated on product states evolved by the exact dynamics $e^{-\ii Ht_F}$ (black, left axis: mean $\pm$ std over $20$ random bitstrings; dotted: the $\ket{01\cdots01}$ trajectory) and their half-chain entanglement entropy in bits (cyan, right axis), versus the forward time $t_F$.
    As the entanglement grows from $0$ to $4.0$~bits, the ensemble-mean error decreases from the random-product level to the Haar-ensemble level (blue dashed), which it reaches at the largest entanglement of the scan, $t_F=8$.
    }
    \label{apd:fig:entangled_families}
\end{figure}

To exhibit the error--entanglement relation directly, we evolve product states by the exact dynamics $e^{-\ii Ht_F}$ for a forward time $t_F\in[0,8]$, so that the half-chain entanglement grows continuously from $0$ to $4.0$~bits, and evaluate the expectation of the \emph{fixed} backward-evolved error operators of Fig.~3 of the main text (at $t=5$) on these states.
As shown in \cref{apd:fig:entangled_families}(b), the ensemble-mean truncation error decreases, with the expected oscillations, from the random-product level to the Haar floor, which it reaches as the entanglement approaches its maximum; the same holds at $w^*=5$.
These input states are generated by the very dynamics being simulated, so the observed error reduction is precisely the one exploited by the hybrid $\mps$--$\lpd$ protocol; both extremes of Fig.~3 of the main text are recovered as endpoints of one continuous family.
\subsection{Comparison with the cluster expansion}\label{apd:sec:ce_comparison}

The cluster expansion (CE) of Wild and Alhambra~\cite{wildClassicalSimulationShortTime2023} is the closest rigorous alternative to $\lpd$ for noiseless short-time dynamics: for product initial states and local observables it attains an $n$-independent $\poly(1/\eps)$ cost within the convergence radius $t^*=1/(2e\mathfrak{d})$, where $\mathfrak{d}$ is the interaction-graph degree and time is measured in units in which every local term has at most unit strength, and reaches any fixed real time by analytic continuation.
This guaranteed scaling is stronger in both $n$ and $\eps$ than that of $\lpd$,
whose runtime $\bigO(r\,n^{w^*})$ with $w^*=\ko+\bigO(\log(1/\eps)/\log(t_0/t))$ is guaranteed only for $t<t_0$ and sufficiently entangled input states (\cref{thm:runtime}).
Beyond $t^*$, its analytic continuation reaches any constant time with a runtime polynomial in $1/\eps$ of degree $e^{\pi t/t^*}$,
so both guarantees are confined to constant evolution time.
To make the comparison quantitative, we implemented the CE---to our knowledge for the first time---verified the implementation against a literal evaluation of Eqs.~(8)--(9) of Ref.~\cite{wildClassicalSimulationShortTime2023} (the size-$M$ truncated expansion coincides with the degree-$M$ Taylor polynomial of $\ev{A(t)}$ in $t$), and benchmarked it on the identical mixed-field Ising instance of Fig.~3 of the main text.

\Cref{apd:fig:wa_ce_benchmark} summarizes the outcome.
In practice the truncated CE is limited not by a convergence radius but by arithmetic precision. The root test on the thermodynamic-limit coefficients estimates a convergence radius $\approx1.5/J$; at finite $n$ the series is entire, so higher orders push the accessible horizon further, to $t\approx2.6$, until double-precision cancellation stops it: the terms of the series grow large and cancel to an $\bigO(1)$ value, exhausting the sixteen significant digits of double-precision arithmetic, the earlier the larger $n$. The analytic continuation of Ref.~\cite{wildClassicalSimulationShortTime2023}, oracle-tuned (its parameters set against the exact solution, the setting most favourable to the CE) and run in 120-digit arithmetic, holds $10^{-2}$ accuracy to $t\approx2$ and fails by $t=3$.
The guaranteed cost beyond $t^*$ is a polynomial in $1/\eps$ whose degree is $e^{\pi t/t^*}$: on this instance $\mathfrak d=4$, and normalizing the strongest local term, the on-site field $h_xX_j+h_yY_j$ of strength $\sqrt{0.8^2+0.9^2}\,J\approx1.20J$, to unit strength gives $t^*=1/(2e\mathfrak{d}\cdot1.20J)\approx0.038/J$, so the degree is ${\approx}5\times10^{35}$ at $t=1$ and ${\approx}4\times10^{178}$ at $t=5$.
By contrast, $\lpd$ tracks the exact result to $t=5$ within $10^{-1}$ both at $w^*=4$ ($17\,850$ retained Pauli strings, under a minute on a laptop) and at the threshold $w^*=5$ of Fig.~3 of the main text.

\begin{figure}[!t]
    \centering
    \includegraphics[width=0.9\linewidth]{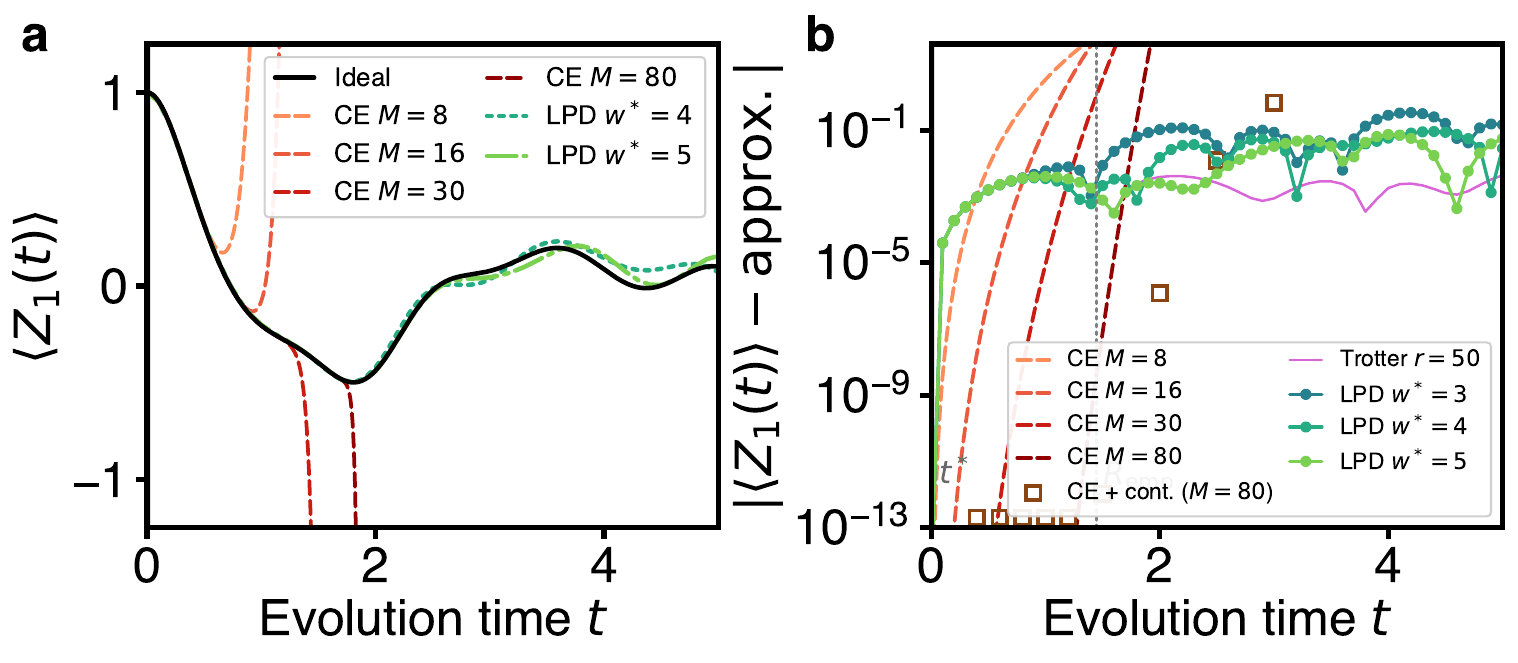}
    \caption{
    Cluster expansion versus $\lpd$ on the instance of Fig.~3 of the main text.
    (a) $\ev{Z_1(t)}$: exact (black), truncated CE at increasing orders $M$ (reds), $\lpd$ at increasing $w^*$ (viridis). The CE curves depart from the exact result in order of $M$ near the convergence radius, while $\lpd$ tracks it to $t=5$.
    (b) Errors against the exact result; brown squares: the CE with the analytic continuation of Ref.~\cite{wildClassicalSimulationShortTime2023} at $M=80$ in 120-digit arithmetic, which holds $10^{-2}$ accuracy to $t\approx2$ and fails by $t=3$. The gray band at the left edge is the entire rigorously guaranteed CE window $t\le t^*<0.04/J$; the dotted line marks the root-test estimate of the thermodynamic-limit convergence radius, ${\approx}1.5$.
    }
    \label{apd:fig:wa_ce_benchmark}
\end{figure}

The state dependence of the two truncation strategies is the conceptually important contrast.
The CE convergence bound is state-independent (the product-state assumption of Ref.~\cite{wildClassicalSimulationShortTime2023} enters only through classical access to $\bigO(M)$-local marginals), so input entanglement can never reduce the CE cost; and for the Loschmidt-echo results the product assumption is structural.
We verified this numerically: the CE coefficient growth and accuracy horizon are unchanged from product to GHZ, random-stabilizer, and Haar inputs spanning half-chain entropies $S_{n/2}\in[0,4.3]$~bits.
Order truncation is entanglement-blind, whereas the weight-truncation error of $\lpd$ \emph{decreases} with the input entanglement (\cref{apd:sec:entangle_vs_random}).

\end{document}

%% file: style/macros.tex
\newcommand{\supp}{\textup{\textsf{supp}}}

\newcommand{\lpd}{\textup{\texttt{LPD}}}
\newcommand{\mps}{\textup{\texttt{MPS}}}

\newcommand{\bigO}{\mathcal{O}}

\newcommand{\ii}{\textup{i}}

\newcommand{\ob}{O}

\newcommand{\eps}{\epsilon}

\newcommand{\poly}{\textup{poly}}

\newcommand{\dt}{\textup{d} t}

\newcommand{\scrE}{\mathscr{E}}
\newcommand{\statee}{\mathscr{E}}

\newcommand{\bbE}{\mathbb{E}}

\newcommand{\vertiii}[1]{{\left\vert\kern-0.25ex\left\vert\kern-0.25ex\left\vert #1
		\right\vert\kern-0.25ex\right\vert\kern-0.25ex\right\vert}}
\newcommand{\Vertiii}[1]{{\vert\kern-0.25ex\vert\kern-0.25ex\vert #1
		\vert\kern-0.25ex\vert\kern-0.25ex\vert}}

\newcommand{\Input}{\textup{input}}

\newcommand{\nfnorm}[1]{\norm{#1}_{\bar{2}}}

\newcommand{\cumnorm}{\mathcal{N}}
\newcommand{\ko}{k_o}
\newcommand{\kh}{k_h}
\newcommand{\fLambda}{\Lambda_{\bar{2}}}
\newcommand{\pbasis}{\mathbb{P}}
\newcommand{\npbasis}{\mathbb{\overline{P}}}
\newcommand{\opnorm}[1]{\norm{#1}_{\mathrm{\infty}}}
\newcommand{\psie}{\psi_{S}}

\newcommand{\pf}{\tilde{U}}
\newcommand{\tO}{\tilde{O}}

\newcommand{\haar}{\mathrm{Haar}}